\documentclass[a4paper,notitlepage,fleqn,12pt]{article}
\usepackage[tbtags]{amsmath}
\usepackage{amsfonts}
\usepackage{amsthm}
\usepackage{color}
\usepackage{graphicx,amssymb}
\usepackage{multirow}
\usepackage{rotating}
\usepackage{mathrsfs}
\usepackage{epstopdf}
\usepackage{enumerate}
\usepackage{pdflscape}
\usepackage{geometry}
\usepackage{natbib}
\usepackage{changepage}
\usepackage{dcolumn}
\usepackage{tabularx}
\usepackage{subfig}
\usepackage{mathabx}
\usepackage[mathscr]{euscript}
\usepackage{diagbox}

\usepackage{floatrow}
\usepackage{tikz}
\usetikzlibrary{graphs}
\usepackage{mathdots}
\usepackage{algorithm}
\usepackage{algorithmicx}
\usepackage{algpseudocode}
\newcolumntype{d}[1]{D{.}{.}{#1}}
\newcolumntype{Y}{>{\raggedleft\arraybackslash}X}
\newcolumntype{Z}{>{\centering\arraybackslash}X}
\DeclareMathOperator*{\vectorize}{vec}

\newtheorem{lemma}{Lemma}[section]
\newtheorem{proposition}{Proposition}
\newtheorem{cor}{Corollary}
\newtheorem{assum}{Assumption}
\newtheorem{thm}{Theorem}
\theoremstyle{remark}
\newtheorem{remark}{Remark}

\newtheorem{definition}{Definition}

\DeclareMathOperator*{\argmin}{argmin}

\DeclareMathOperator*{\diag}{diag}

\newcommand{\bm}[1]{\mbox{\boldmath{$#1$}}}
\numberwithin{equation}{section}
\newcommand{\cm}[1]{\mbox{\boldmath$\mathscr{#1}$}}
\newcommand{\cmt}[1]{\mbox{\boldmath\scriptsize$\mathscr{#1}$}}
\usepackage[colorlinks,linkcolor=blue,citecolor=blue]{hyperref}

\title{\vspace{-10mm}HAR-It\^o models and high-dimensional HAR modeling for high-frequency data}
\author{Huiling Yuan, Kexin Lu, Yifeng Guo and Guodong Li
	\\ \textit{University of Hong Kong} }

\begin{document}
\maketitle

\begin{abstract}
	It is an important task to model realized volatilities for high-frequency data in finance and economics and, as arguably the most popular model, the heterogeneous autoregressive (HAR) model has dominated the applications in this area. However, this model suffers from three drawbacks: (i.) its heterogeneous volatility components are linear combinations of daily realized volatilities with fixed weights, which limit its flexibility for different types of assets, (ii.) it is still unknown what is the high-frequency probabilistic structure for this model, as well as many other HAR-type models in the literature, and (iii.) there is no high-dimensional inference tool for HAR modeling although it is common to encounter many assets in real applications.
	To overcome these drawbacks, this paper proposes a multilinear low-rank HAR model by using tensor techniques, where a data-driven method is adopted to automatically select the heterogeneous components. In addition, HAR-It\^o models are introduced to interpret the corresponding high-frequency dynamics, as well as those of other HAR-type models. 
	Moreover, non-asymptotic properties of the high-dimensional HAR modeling are established, and a projected gradient descent algorithm with theoretical justifications is suggested to search for estimates.
	Theoretical and computational properties of the proposed method are verified by simulation studies, and the necessity of using the data-driven method for heterogeneous components is illustrated in real data analysis.
\end{abstract}

\textit{Keywords and phrases:} Diffusion process; Heterogenous autoregressive model; High-dimensional analysis; High-frequency data; Non-asymptotic property; Tensor technique.

\newpage

\section{Introduction}
\label{section1}

Volatility analysis is one of the most important tasks in finance and economics \citep{Engle:1982,Bollers:1986} and, with the widespread availability of high-frequency data, more and more recent discussions have concentrated on modeling realized volatilities, which can be constructed from high-frequency intraday observations. Examples include the realized generalized autoregressive conditional heteroscadestic (GARCH) model \citep{Hansen:2012}, high-frequency based volatility model \citep{Shephard:2010}, heterogeneous autoregressive (HAR) model \citep{Corsi:2009}, multiplicative error model \citep{Engle:2006}, and mixed data sampling model \citep{Ghysels:2006}. 
As arguably the most popular one among these approaches, the HAR model has a simple cascade structure \citep{Corsi:2009}.
Specifically, the short-term, medium-term and long-term volatility components are first identified, and they are usually the linear combinations of daily realized volatilities.
We then regress the future realized volatility on the three heterogeneous components, and this leads to an autoregressive form.
Surprisingly the simple HAR model can even outperform a few powerful deep neural networks for some real applications in the area of machine learning \citep{Bucci:2020}.

Since its appearance, the HAR model has attracted huge amount of attention from the literature.
First, due to the importance of jump effects for high-frequency data \citep{Ait:2009,Anderson:2012}, \cite{An:2007}  included the estimated jumps as the fourth heterogeneous component to improve the forecasting performance; see also \cite{CorsiandPirinoandReno:2010} for threshold bipower variation and \cite{PattonandSheppard:2015} for signed jumps.
\cite{Corsi:2012} further added one more component of leverage effects.
Secondly, some efforts have been spent on generalizing HAR models. \cite{BollersandPattonandQuaedvlieg:2016} suggested a HAR model with time-varying coefficients in terms of parametric forms; see also \cite{BekiermanandManner:2018}.
Moreover, \cite{Chen:2018} considered a similar model with coefficients having a nonparametric form, \cite{MM:2008} proposed a multiple-regime smooth transition HAR model, and \cite{CorsiandMittnikandPigorschandPigorsch:2008} considered a more sophisticated HAR model with GARCH errors.
The jump effect has also been discussed for these generalized HAR models \citep{ BuccheriandCorsi:2021,Caporin2022}.
Thirdly, besides the ordinary least squares, more efficient estimation methods have been discussed in the literature, including the weighted least squares \citep{PattonandSheppard:2015} and robust regression  \citep{ClementsandPreve:2021} methods.
Fourthly, to achieve higher forecasting accuracy, a transformation is usually employed to realized volatilities before estimation, and one usually chooses the logarithmic \citep{Corsi:2009,Chen:2018,Bekaert:2014} or Box-Cox transformation \citep{Taylor:2017}.
Finally, besides the stock market, the HAR model has also been applied to Bitcoin \citep{Trucios:2019, Hu:2021}, energy \citep{Lyocsa:2021, Luo:2022},  agricultural commodities \citep{Degiannakis:2022, LuoandTony:2022}, and many others.

Most of the above studies on HAR models are limited to the univariate case, while it is common in practice to forecast the realized volatilities of many assets. One popular way is to extend the HAR model to a multivariate version for the modeling \citep{HongLeeHwang:2020}.
For example, \cite{Bubak:2011} considered a vector HAR model with multivariate GARCH error terms to improve the forecasting accuracy, and similar model settings can be found in \cite{SoucekandTodorova:2013}.
\cite{BauerandVorkink:2011} conducted the prediction for realized covariances, while the vector HAR model plays a key role in the modeling; see also \cite{Oh:2016,BollersandPattonandQuaedvlieg:2018}.
This paper focuses on the vector HAR model for predicting realized volatilities and, when the number of assets $N$ is large, the resulting model will have a much larger number of parameters, which has the rate of $O(N^2)$. As a result, some dimension reduction methods will be needed to make the forecasting feasible. 
\cite{Reinsel:1983} proposed a vector autoregressive model with the low-rank assumption being imposed to the row space of coefficient matrices, and the method was applied to vector HAR models by \cite{Cubadda:2017}. This reduces the number of parameters to $O(N)$, while one may want to know whether the low-rank assumption can be assumed to the column space, or even both row and column spaces, of coefficient matrices; see Remark \ref{rem-sec2} for more details.

On the other hand, for most HAR-type models in the above, the three heterogeneous volatility components are set to the daily, weekly and monthly realized volatilities, where the later two are the simple averages of 5- and 22-day realized volatilities, respectively \citep{Corsi:2009}. However, this setting has been shown to be limited by more and more empirical evidences.
\cite{ChenandHardleandPigorsch:2010} first questioned the involvement of long-term volatility components, measured by monthly realized volatilities, and argued that a structural break, together with a small-order autoregressive model, can outperform HAR models.
\cite{Audrino:2016} empirically evaluated the appropriateness of the three components by using the Lasso method \citep{Tibshirani:1996} to automatically select variables of an autoregressive model with a large order, and the negative result was confirmed for nine stocks from the US market.
Note that the three components correspond to the average of $j$-day realized volatilities with $j=1$, 5 and 22, respectively.
\cite{Kohler:2021} considered many other combinations of $j$'s and empirically showed that the combination of $(1,5,22)$ performed worse; see also \cite{HongLeeHwang:2020}.
In fact, the three components in HAR models can be interpreted as factors along lags, and this motivates us to consider a data-driven method to choose these factors; see model \eqref{har-rv} and Remark \ref{rem-factor} for details.
The first contribution of this paper is to propose a multilinear low-rank HAR model in Section 3.1 to forecast realized volatilities by using tensor techniques, where the low-rank assumption is imposed to both row and column spaces of coefficient matrices, and the heterogeneous volatility components are selected automatically. The asymptotic normality of its ordinary least squares estimation is discussed in Section 3.3.

Recently some efforts have been spent on constructing high-frequency models for intraday asset prices such that the corresponding low-frequency integrated volatilities have a GARCH-type representation.
Examples include the GARCH-It\^{o} model \citep{KW:2016}, the factor GARCH-It\^o model \citep{KimFan}, the realized GARCH-It\^o model \citep{song2020realized}, the overnight GARCH-It\^o volatility model \citep{KW:2023}, the exponential realized GARCH-It\^o model \citep{Kim2023}, and many others.
These GARCH-It\^{o}-type models provide a theoretical bridge to reconcile low-frequency GARCH volatility representations and high-frequency volatility processes, and hence a better interpretation and more reliable inference can be achieved by harnessing GARCH models and realized volatilities.
As a natural, yet non-trivial follow-up of the GARCH-It\^{o} literature, the second contribution of this paper is to propose the univariate and multivariate HAR-It\^o models in Section 2, and their low-frequency integrated volatility admits a univariate and a multivariate HAR representations, respectively. 
To the best of our knowledge, this is the first attempt in the literature to explore the high-frequency dynamics for HAR-type models.

In the meanwhile, all current used HAR modeling tools are for the low-dimensional setting with a fixed number of assets. However it is common to encounter many financial assets in real applications, and it is urgent to design high-dimensional inference tools for HAR models, where the number of assets $N$ may diverge. 
The third contribution of this paper is to fill this gap by conducting the high-dimensional HAR modeling for high-frequency data in Section 4. Specifically, the proposed model has a form similar to that of vector autoregressive models with measurement errors, and Section 4.1 establishes the non-asymptotic properties of its high-dimensional estimation. A projected gradient descent algorithm is suggested to search for estimates in Section 4.2, and its theoretical justifications have also been provided. 

In addition, Section 5 conducts simulation experiments to evaluate the finite-sample performance of the proposed methodology, and its usefulness is further demonstrated by empirical examples in Section 6. Section 7 gives a short conclusion and discussion, and all technical proofs are relegated to the supplementary file.

\section{HAR-It\^o models for high-frequency data} \label{Sec2}
\subsection{Univariate HAR-It\^o model}
\label{Sec2.1}

This section proposes the univariate and multivariate HAR-It\^o models for high-frequency asset prices such that their low-frequency integrated volatilities have univariate and multivariate HAR representations at Propositions \ref{prop:1} and \ref{prop:2}, respectively.

Let $X_{t}$ be the log price of an asset at time $t\in \mathbb{R}_{+}$, where $\mathbb{R}_{+}=[0,\infty)$.
Denote by $\mathcal{F}_t$ the collection of all information up to time $t$, and then $\{\mathcal{F}_t, t\geq 0\}$ is a filtration.
To account for jump components in financial industry, we consider a jump diffusion model,
\begin{equation}
	\label{jump-diffusion}
	dX_{t}=\mu_t dt+\sigma_{t}dB_{t}+L_{t} d \Lambda_{t},
\end{equation}
where $\mu_t$ is a drift term, $\sigma_{t}$ is an instantaneous volatility process, $B_{t}$ is a standard Brownian motion,  $\Lambda_{t}$ is a standard Poisson process with constant intensity $\lambda>0$, the jump sizes $\{L_{t},t\geq 0\}$ are independent and identically distributed ($i.i.d.$) with finite fourth moment, the processes of $\sigma_{t}$, $B_{t}$ and $\Lambda_{t}$ are all adapted to the filtration $\{\mathcal{F}_t\}$, and the jump sizes are independent of the three processes.
Moreover, let $Z_{t}=\int_{[t]}^{t}dW_{t}$, where $W_{t}$ is another standard Brownian motion with respect to $\mathcal{F}_t$ such that $dB_{t}dW_{t}=\rho dt$ with $|\rho|<1$ almost surely, and $\rho$ can be seen as the leverage effect between the log price and the instantaneous volatility. We then define the HAR-It\^{o} model below.

\begin{definition}
	\label{def:har-ito}
	The log prices $\{X_t, t\in \mathbb{R}_{+}\}$ at \eqref{jump-diffusion} are said to follow an univariate HAR-It\^{o} model if the instantaneous volatility satisfies
	\begin{equation}\label{eq:har-ito}
		\begin{array}{rl}
			\sigma_{t}^{2}=\sigma_{[t]}^{2}&+(t-[t])\left(\omega-\sigma_{[t]}^{2}\right) +\alpha_1 \displaystyle\int_{[t]}^{t}\sigma_{s}^{2}ds
			+\sum\limits_{l=2}^{P}\alpha_l \displaystyle\int_{[t]-l+1}^{[t]-l+2}\sigma_{s}^{2}ds\\
			&+	\beta \displaystyle \int_{[t]}^{t} L_{s}^{2} d \Lambda_{s}+\upsilon([t]+1-t)Z_{t}^{2},
		\end{array}
	\end{equation}
	where $[t]$ denotes the integer part of $t$, i.e. $[t]=n-1$ for $n-1<t\leq n$, $P$ is a positive integer, the parameters of $\omega$, $\beta$, $\upsilon$, $\alpha_1$ are all positive, and $\alpha_l \geq 0$ with $2 \leq l \leq P$.
\end{definition}

Denote $\bm{\alpha}=(\alpha_1, \alpha_2, \cdots, \alpha_{P})^{\top}$ and $\bm{\theta}=(\lambda,\rho,\omega, \beta, \upsilon, \bm{\alpha}^{\top})^{\top}$, and the above definition relies on the parameter vector $\bm{\theta}$.
In addition, the instantaneous volatility $\sigma_{t}^{2}$ at \eqref{eq:har-ito} is a process with continuous time and is defined for all $t\in \mathbb{R}_{+}$. Consider the integer time points only, and then the process has the form of
\begin{equation*}
	\sigma_{n}^{2}=\omega
	+\sum\limits_{l=1}^{P}\alpha_l y_{n-l+1}+\beta J_n,
\end{equation*}
with probability one, where
\begin{equation*}
	y_n=IV_n=\displaystyle\int_{n-1}^{n}\sigma_{s}^{2}ds \hspace{5mm}\text{and}\hspace{5mm}
	J_n=\displaystyle\int_{n-1}^{n} L_{s}^{2} d \Lambda_{s}
\end{equation*}
are the integrated volatility and jump variation, respectively. 
Moreover, from \eqref{eq:har-ito}, the daily volatility (or integrated volatility) in the past will contribute to the spot volatility (or instantaneous volatility) by the introduction of $\alpha_l$ with $2\leq l\leq P$.
This paper will concentrate on the integrated volatility, $IV_n$ or $y_n$.

As in \cite{song2020realized}, we denote $\omega_{L}=\mathbb{E}(L_{t}^{2})$ and $M_{t}=L_{t}^{2}-\omega_{L}$.
Moreover, let
\begin{equation*}
	\omega^{g}=\varrho_{1}\omega+\varrho_{2}\beta\omega_L\lambda+(\varrho_{2}-2\varrho_{3})v, \hspace{5mm} 
	\beta^{g}=(\varrho_{1}-\varrho_{2})\beta\hspace{5mm}\text{and}\hspace{5mm}
	\alpha_{l}^{g}=(\varrho_{1}-\varrho_{2})\alpha_{l}
\end{equation*}
for all $1 \leq l \leq P$, where $\varrho_{1}=\alpha_1^{-1}(e^{\alpha_1}-1)$, $\varrho_{2}=\alpha_1^{-2}(e^{\alpha_1}-1-\alpha_1)$ and $\varrho_{3}=\alpha_1^{-3}(e^{\alpha_1}-1-\alpha_1-{\alpha_1^{2}}/{2})$. 

\begin{proposition}\label{prop:1}
Suppose that $n \geq P+1$ and $0< \alpha_1<1$. Then the integrated volatility for univariate HAR-It\^{o} models at Definition \ref{def:har-ito} satisfies
	\begin{align}
		\label{har-model}
		y_n=\omega^{g}+\sum_{j=1}^{P}\alpha_{j}^{g}y_{n-j}+\beta^{g} J_{n-1}+\varepsilon_{n}
	\end{align}
	with probability one, where $\varepsilon_{n}=\varepsilon_{n}^{c}+\varepsilon_{n}^{J}$,
	\begin{align*}
		\varepsilon_{n}^{c}=&2 v \alpha_1^{-2} \displaystyle \int_{n-1}^{n}\left\{\alpha_1\left(n-t-\alpha_1^{-1}\right) e^{\alpha_1(n-t)}+1\right\} Z_{t} d Z_{t},\\ 
		\varepsilon_{n}^{J}=&\beta \alpha_{1}^{-1}\left\{\displaystyle\int_{n-1}^{n}\left(e^{\alpha_{1}(n-t)}-1\right) [M_{t} d \Lambda_{t}+\omega_{L}(d \Lambda_{t}-\lambda d t)]\right\},
	\end{align*}
and $\varepsilon_{n}^{c}$ and $\varepsilon_{n}^{J}$ are innovation terms from the continuous diffusion process and the jump component, respectively.
Moreover, both $\{\varepsilon_{n}^{c}\}$ and $\{\varepsilon_{n}^{J}\}$ are $i.i.d.$ sequences with mean zero and finite variance.
\end{proposition}

From the above proposition, the integrated volatility $\{y_n\}$ admits an iterative form of autoregression with an exogenous variable of $J_{n-1}$, and it will be a stationary low-frequency time series if $1-\sum_{j=1}^{P}\alpha_{j}^{g}z^j\neq 0$ for all complex values satisfying $|z|<1$.
Moreover, the term of $Z_{t}=\int_{[t]}^{t}dW_{t}$ plays a key role in designing the instantaneous volatility $\sigma_t$ at \eqref{eq:har-ito}. We may consider a more general form of $Z_{t}=\int_{[t]}^{t}z_tdW_{t}$ with $z_t$ being a nonrandom continuous function, and a result similar to Proposition \ref{prop:1} can be established. 
However, when $z_t$ is a random function, say $z_t=\sigma_t$ in \cite{KW:2016}, the derived innovations $\{\varepsilon_{n}^{c}\}$ will be a martingale difference sequence only although we can still have an autoregressive form similar to \eqref{har-model}.
In fact, by choosing suitable random functions for $z_t$, the derived innovation $\varepsilon_{n}^{c}$ can even be heavy-tailed \citep{Shin2021}.
This paper will focus on the design at Definition \ref{def:har-ito} since we aim at the  HAR model.

When there is no jump component, i.e. $L_t=0$ with probability one for all $t$, the iterative form at \eqref{har-model} can be further simplified into
\begin{equation}
	\label{har-extension}
	y_n=\omega^{g}+\sum\limits_{j=1}^{P}\alpha_{j}^{g}y_{n-j}+\varepsilon_n,
\end{equation}
and the integrated volatility $y_n$ can be estimated consistently by the realized volatility, denoted by $RV_n$.
Let $IV_n=y_n$, $IV_{n}^{(w)}=5^{-1}\sum_{l=1}^{5}y_{n+1-l}$ and $IV_{n}^{(m)}=22^{-1}\sum_{l=1}^{22}y_{n+1-l}$ be the daily, weekly and monthly integrated volatilities, respectively, and we consider a special case of \eqref{har-extension} below,
\begin{equation}\label{har-rv}
IV_{n}= \beta_{0}+\beta_{d} IV_{n-1}+\beta_{w} IV_{n-1}^{(w)}+\beta_{m} IV_{n-1}^{(m)}
+\varepsilon_{n},
\end{equation}
where $P=22$, $\omega^{g}=\beta_{0}$,  $(\alpha_1^g,\dots,\alpha_P^g)=(\beta_{d},\beta_{w},\beta_{m})\mathbf{U}_{\mathrm{C}}^{\top}$, and 
\[
\mathbf{U}_{\mathrm{C}}^{\top} =	\begin{pmatrix}
&	1 &\quad 0 &\quad 0& \quad 0 &\quad 0 &\quad 0 &\quad\cdots &\quad 0\\
&	1/5&\quad 1/5& \quad 1/5 & \quad 1/5 &\quad  1/5&\quad 0 & \quad\cdots& \quad 0\\
&1/22 &\quad 1/22&\quad 1/22 &\quad 1/22 &\quad 1/22 &\quad 1/22 &\quad\cdots& \quad 1/22
\end{pmatrix} \in\mathbb{R}^{3\times 22}.
\]
\cite{Corsi:2009} made use of model \eqref{har-rv} to propose the HAR model for predicting the realized volatility, and the proposed HAR-It\^{o} model at Definition \ref{def:har-ito} and Proposition \ref{prop:1}  provides a probabilistic support for this popular model.
On the other hand, the cascade structure at \eqref{har-rv} reveals that there usually exists a low-dimensional structure among the parameters $\alpha_j$'s and $\alpha_j^g$'s along the direction of lags.


For the case with jumps, the realized volatility $RV_n$ is a consistent estimator of $IV_n+J_n$, while the integrated volatility $IV_n$ can be estimated consistently by the bipower variation, denoted by $BV_n$.
From Proposition \ref{prop:1}, the following HAR-type model can be used to predict the bipower variation,
\[
B V_{n}= \beta_{0}+\beta_{d} B V_{n-1}+\beta_{w} B V_{n-1}^{(w)}+\beta_{m} B V_{n-1}^{(m)}
+\beta_{J} \widehat{J}_{n-1}+\epsilon_{n}^{J},
\]
where $B V_{n-1}^{(w)}$ and $B V_{n-1}^{(m)}$ are weekly and monthly averages of the daily bipower variation $B V_{n}$, respectively, and $\widehat{J}_{n} = \max \{R V_{n}-B V_{n}, 0\}$ is the estimated jump; see \cite{Cubadda:2017}.
The error term $\epsilon_{n}^{J}$ refers to model error $\varepsilon_{n}$ in Proposition \ref{prop:1} and estimation errors when one uses $B V_{n}$ and $\widehat{J}_{n}$ to estimate $IV_n$ and $J_{n}$, respectively.
For other HAR-type models with jump components, such as \cite{An:2007}, in the literature, their high-frequency models can be constructed with a form similar to Definition \ref{def:har-ito}.

\subsection{Multivariate HAR-It\^o model}
\label{Sec2.2}

This subsection defines a multivariate HAR-It\^o model for $N$ financial assets.
Denote by $X_{i, t}$ and $\sigma_{i, t}$ the log price of the $i$-th asset and its instantaneous volatility at $t\in\mathbb{R}_{+}$ for each $1\leq i\leq N$, respectively, and let $\mathbf{X}_{t}=(X_{1, t}, \ldots, X_{N, t})^{\top}$ and $\bm{\sigma}_{t}=(\sigma_{1, t}, \ldots, \sigma_{N, t})^{\top}$.

Define two $N$-dimensional Brownian motions, $\mathbf{B}_{t}=\left(B_{1, t}, \ldots, B_{N, t}\right)^{\top}$ and $\mathbf{W}_{t}=\left(W_{1, t}, \ldots, W_{N, t}\right)^{\top}$, where $d B_{i,t}\cdot d B_{j,t}= \rho_{i,j}^{B} dt$, $d W_{i,t}\cdot d W_{j,t}=\rho_{i,j}^{W} dt$, $d B_{i,t}\cdot dW_{j,t}=\rho_{i,j} dt$, $\rho_{i,i}^{B}=1$, and $\rho_{i,i}^{W}=1$ for all $1\leq i, j\leq N$.
Let $\bm{\rho}^B=(\rho_{i,j}^B)\in\mathbb{R}^{N\times N}$, $\bm {\rho}^W=(\rho_{i,j}^W)\in\mathbb{R}^{N\times N}$ and $\bm{\rho}=(\rho_{i,j})\in\mathbb{R}^{N\times N}$, where
$\bm \rho^B$ and $\bm \rho^W$ are symmetric.
Moreover, let $Z_{i,t}=\int_{[t]}^{t}dW_{i,t}$ for $1\leq i\leq N$, and denote $\mathbf{Z}_t=(Z_{1, t}, \ldots, Z_{N, t})^{\top}$.
For the jump component, let $\bm{\Lambda}_{t}=(\Lambda_{1, t}, \ldots, \Lambda_{N, t})^{\top}$ be an $N$-dimensional Poisson process, and $\mathbf{L}_{t}=(L_{1, t}, \ldots, L_{N, t})^{\top}$ be $N$-dimensional jump sizes. We assume that the jump components for different assets are independent, $\{\mathbf{L}_{t}, t\in\mathbb{R}_{+}\}$ are $i.i.d.$ with finite fourth moment, $\mathbf{L}_{t}$ is independent of $\bm{\Lambda}_{t}$, $\mathbf{B}_{t}$, $\mathbf{W}_{t}$ and $\bm{\sigma}_{t}$, and all the above diffusion processes are adapted to the filtration $\{\mathcal{F}_t,t\in\mathbb{R}_{+}\}$.
As a result, the jump diffusion model can be introduced below,
\begin{equation}
	\label{high-dimensional-har}
	d{X}_{i,t}=\mu_{i,t} dt+\sigma_{i,t}d B_{i,t}+ L_{i,t}d\Lambda_{i,t}, \hspace{5mm}1\leq i\leq N.
\end{equation}
\begin{definition}
	\label{def:mhar-ito}
	The log prices of $N$ assets $\{\mathbf{X}_{t}, t \in\mathbb{R}_{+}\}$ at \eqref{high-dimensional-har} are said to follow a multivariate HAR-It\^{o} model if, for each $1\leq i\leq N$, the instantaneous volatility of the $i$-th asset satisfies
	\begin{equation}
		\label{eq:mhar-ito}
		\begin{array}{rl}
			\sigma_{i,t}^{2}=\sigma_{i,[t]}^{2}&+(t-[t])\left(\omega_{i}-\sigma_{i,[t]}^{2}\right)
			+\sum\limits_{j=1}^{N}\alpha_{i,j}^{(1)} \displaystyle\int_{[t]}^{t}\sigma_{j,s}^{2}ds
			+\sum\limits_{l=2}^{P}\sum\limits_{j=1}^{N}\alpha_{i,j}^{(l)} \displaystyle\int_{[t]-l+1}^{[t]-l+2}\sigma_{j,s}^{2}ds\\
			&+\beta_i \displaystyle\int_{[t]}^{t} L_{i,s}^{2}d\Lambda_{i,s}+ v_i([t]+1-t) Z_{i,t}^{2},
		\end{array}
	\end{equation}
	where $[t]$ denotes the integer part of $t$, i.e., $[t]=n-1$ for $n-1<t\leq n$,
	$P$ is a positive integer, the parameters of $\omega_{i}$, $\beta_{i}$, $\upsilon_i$, $\alpha_{i,j}^{(1)}$'s are all positive, and $\alpha_{i,j}^{(l)}\geq 0$ for all $2\leq l\leq P$.
\end{definition}

Let $\bm \omega=(\omega_1, \cdots, \omega_N )^{\top}\in\mathbb{R}^{N}$, $\bm \alpha_{i}^{(l)}=(\alpha_{i,1}^{(l)}, \cdots, \alpha_{i,N}^{(l)} )^{\top}\in\mathbb{R}^{N}$, $\bm {\alpha}^{(l)}=(\bm \alpha_{1}^{(l)}, \cdots, \bm \alpha_{N}^{(l)})^{\top}\in\mathbb{R}^{N\times N}$, and $\bm{\beta}=\diag\{\beta_1,\ldots,\beta_N\}\in\mathbb{R}^{N\times N}$.
The instantaneous volatility $\bm{\sigma}_{t}$ defined at \eqref{eq:mhar-ito} is an $N$-dimensional process with continuous time and, when being restricted to integer time points, it has the form of
\begin{equation*}
	\bm {\sigma}_{n}^{2}=\bm\omega
	+\sum\limits_{l=1}^{P}\bm {\alpha}^{(l)} \mathbf y_{n-l+1}+	\bm{\beta} \mathbf {J}_n,
\end{equation*}
with probability one, where the integrated volatility and jump variation for the $i$-th asset are defined as
\[
y_{i,n}=\displaystyle \int_{n-1}^{n} \sigma_{i,t}^{2}dt \hspace{5mm}\text{and}\hspace{5mm}
J_{i,n}=\displaystyle \int_{n-1}^{n} L_{i,s}^{2} d \Lambda_{i,s},
\]
respectively, $\mathbf{y}_n=(y_{1,n}, \cdots, y_{N,n})^{\top}$, and $\mathbf{J}_n=(J_{1,n},\cdots,  J_{N,n})^{\top}$.
Moreover, from \eqref{eq:mhar-ito}, the values of $\alpha_{i,j}^{(l)}$ with $2\leq l\leq P$ measure the contribution of the $j$-th stock's daily volatility to the $i$-th stock's spot volatility.

For $1 \leq i \leq N$, let $\omega_{i,L}=E(L_{i,t}^{2})$ and $M_{i,t}=L_{i,t}^{2}-\omega_{i,L}$, and denote  $\bm{\omega}_L=\diag\{\omega_{1,L}, \cdots, \omega_{N,L}\}\in\mathbb{R}^{N\times N}$, $\mathbf{V}=\diag\{v_1, \cdots, v_N\}\in\mathbb{R}^{N\times N}$, and $\bm{\lambda}= (\lambda_1, \cdots, \lambda_N)^{\top}\in\mathbb{R}^N$.  
Moreover, let
$\bm{\varrho}_{1}={\bm{\alpha}^{(1)}}^{-1}(\mathbf{e}^{{\bm{\alpha}}^{(1)}}-\mathbf {I}_{N})\in\mathbb{R}^{N\times N}$, 
$\bm{\varrho}_{2}={\bm{\alpha}^{(1)}}^{-2}(\mathbf {e}^{{\bm{\alpha}^{(1)}}}-\mathbf {I}_{N}-\bm{\alpha}^{(1)})\in\mathbb{R}^{N\times N}$, and
$\bm{\varrho}_{3}={\bm{\alpha}^{(1)}}^{-3}(\mathbf {e}^{\bm{\alpha}^{(1)}}-\mathbf {I}_{N}-\bm{\alpha}^{(1)}-\frac{{\bm{\alpha}^{(1)}}^{2}}{2})\in\mathbb{R}^{N\times N}$, where $\mathbf {e}^{\mathbf {A}}=\sum_{k=0}^{\infty}\frac{{\mathbf{A}}^k}{k!}\in\mathbb{R}^{N\times N}$ for a matrix $\mathbf{A}\in\mathbb{R}^{N\times N}$.
Denote
$
\bm{\omega}^{g} =\bm{\varrho}_{1} \bm{\omega} +\bm{\varrho}_{2}\bm{\beta} \bm \omega_L \bm {\lambda }
+\left(\bm{\varrho}_{2}-2 \bm{\varrho}_{3}\right) \mathbf{V} \mathbf{1}_N\in\mathbb{R}^N,
$
\[
{\bm{\beta}}^{g}=\left(\bm{\varrho}_{1}- \bm{\varrho}_{2}\right) \bm{\beta}\in \mathbb{R}^{N\times N} \hspace{3mm}\text{and}\hspace{3mm}
\mathbf{A}_j=\left(\bm{\varrho}_{1}- \bm{\varrho}_{2}\right)\bm{\alpha}^{(j)}\in \mathbb{R}^{N\times N}\hspace{3mm}\text{with}\hspace{2mm} 1\leq j\leq P,
\]
where $\mathbf{1}_N=(1,\ldots,1)^{\top}\in\mathbb{R}^N$.
For simplicity, we further define $d \bm{\Lambda}_t=(d  \Lambda_{1,t},\ldots, d  \Lambda_{N,t})$ and $\mathbf{Z}_t d\mathbf{Z}_t =(Z_{1,t}dZ_{1,t} , \cdots, Z_{N,t}d Z_{N,t})^{\top}$.

\begin{proposition}
	\label{prop:2}
	Suppose that  $n \geq P+1$, and $\bm{\alpha}^{(1)}$ has spectral radius less than one.
	Then the integrated volatility for multivariate HAR-It\^{o} models at Definition \ref{def:mhar-ito} satisfies
	\begin{align}
		\mathbf{y}_n
		=\bm{\omega}^{g}+\sum\limits_{j=1}^{P}\mathbf{A}_j\mathbf{y}_{n-j}+\bm{\beta}^{g}\mathbf{J}_{n-1}+\bm {\varepsilon}_n
		\label{vector_integrate}
	\end{align}
   with probability one, where $\bm{\varepsilon}_{n}=\bm{\varepsilon}^{c}_{n}+	\bm{\varepsilon}^{J}_{n}\in\mathbb{R}^N$, 
	\begin{align*}
		\bm{\varepsilon}^{c}_{n}=&2 \mathbf{V}{{ \bm{\alpha}}^{(1)}}^{-2}\displaystyle \int_{n-1}^{n}\left\{\left((n-t){\bm\alpha}^{(1)}- \mathbf{I}_N\right)\mathbf{e}^{(n-t){ \bm{\alpha}}^{(1)}}+ \mathbf{I}_N\right\} \mathbf {Z}_{t} d \mathbf {Z}_{t}\in\mathbb{R}^N, \\
		\bm{\varepsilon}_{n}^{J}=&\bm{\beta} {{ \bm{\alpha}}^{(1)}}^{-1}\left\{\displaystyle \int_{n-1}^{n}\left(\mathbf{e}^{{(n-t){ \bm {\alpha}}^{(1)}}}-\mathbf {I}_N\right) [\mathbf{M}_{t} d \bm{\Lambda}_{t}+\bm{\omega}_{L}(d \bm{\Lambda}_{t}- \bm{\lambda} d t)]\right\}\in\mathbb{R}^N,
	\end{align*}
    and $\bm{\varepsilon}^{c}_{n}$ and $\bm{\varepsilon}^{J}_{n}$ are innovation terms from the continuous diffusion process and the jump component, respectively.
	Moreover, $\{\bm{\varepsilon}^{c}_{n}\}$ and $\{\bm{\varepsilon}^{J}_{n}\}$ are $i.i.d.$ sequence with mean zero and finite variance matrices.
\end{proposition}

Similar to the univariate case, the integrated volatility $\{\mathbf{y}_n\}$ has a representation of vector autoregression with exogenous variable of $\mathbf{J}_{n-1}$.   
In the absence of jump components, the representation at \eqref{vector_integrate} will reduce to the following vector HAR model,
\begin{align}
	\label{vector_integrate_2}
	\mathbf{y}_n
	=\bm{\omega}^{g}+\sum\limits_{j=1}^{P}\mathbf{A}_j\mathbf{y}_{n-j}+\bm {\varepsilon}_n,
\end{align}
and the integrated volatility $\{\mathbf{y}_n\}$ can be consistently estimated by the realized volatility.
Let $\mathbf{y}_{n}^{(j)}=j^{-1}\sum_{l=1}^{j}\mathbf{y}_{n+1-l}$ be the average of $j$-day integrated volatilities, and model \eqref{vector_integrate_2} can be rewritten into
$\mathbf{y}_n
=\bm{\omega}^{g}+\sum\limits_{j=1}^{P}\bm{\beta}_j\mathbf{y}_{n-1}^{(j)}+\bm {\varepsilon}_n$,
which is used by \cite{HongLeeHwang:2020} to predict the realized volatility.
As a result, the proposed HAR-It\^o model at Definition \ref{def:mhar-ito} and Proposition \ref{prop:2} provides a probabilistic support for this vector HAR model.

\begin{remark}\label{rem-sec2}
When the number of assets $N$ is large, there may exist a low-dimensional structure in coefficient matrices $\mathbf{A}_j$'s or $\bm{\alpha}^{(j)}$'s since a market of many assets is usually driven by a few summarized forces \citep{bai2008large, lam2012factor}.
To this end, by using the idea of autoregressive index models \citep{Reinsel:1983}, \cite{Cubadda:2017} imposed a low-rank assumption to the following matrix
\[
(\mathbf{A}_1^{\top},\cdots,\mathbf{A}_P^{\top}) = (\bm{\alpha}^{(1)\top},\cdots,\bm{\alpha}^{(P)\top})(\bm{\varrho}_{1}^{\top}- \bm{\varrho}_{2}^{\top}) \in\mathbb{R}^{N\times NP},
\]
corresponding to the row space of $\mathbf{A}_j$'s or $\bm{\alpha}^{(j)}$'s, and a significant improvement of forecasting accuracy can be observed for the realized volatility.
Consider another matrix $(\mathbf{A}_1,\cdots,\mathbf{A}_P) = (\bm{\varrho}_{1}- \bm{\varrho}_{2})(\bm{\alpha}^{(1)},\cdots,\bm{\alpha}^{(P)}) \in\mathbb{R}^{N\times NP}$, corresponding to the column space of $\mathbf{A}_j$'s or $\bm{\alpha}^{(j)}$'s, and the ranks of these two matrices are not equal in general.
This motivates us to consider the low-rankness of the column space of $\mathbf{A}_j$'s or $\bm{\alpha}^{(j)}$'s or even both row and column spaces simultaneously.
\end{remark}

\section{Multilinear low-rank HAR model}
\label{Sec3}
\subsection{Tensor notations and decomposition}
\label{Tensor Decomposition}
This subsection introduces some tensor notations and Tucker decomposition, which will be used to state the new model and to derive inference tools in this and next sections.

Tensors are natural generalizations of matrices for higher-order data which provide useful data representation formats. The order of a tensor is dimensions, also known as ways or modes. A multidimensional array $\cm {X}\in \mathbb{R}^{p_{1} \times \cdots \times p_{d}}$ is called a $d$th-order tensor, and this paper will focus on third-order tensors. We refer readers to \cite{KoldaandBader:2009} for a detailed review on tensor notations and operations. 

We denote tensors by Euler script capital letters $ \cm{X},\cm{Y}, \ldots$ throughout the article. For a matrix $\mathbf{X}$, denote its Frobenius norm, operator norm, nuclear norm, vectorization, transpose, spectral radius, the $j$-th largest singular value by $\|\mathbf{X}\|_{\rm{F}}$, $\|\mathbf{X}\|_{\rm{op}}$, $\|\mathbf{X}\|_{{*}}$,
$\rm{vec}(\mathbf{X})$, $\mathbf{X}^{\top}$, $\rho(\mathbf{X})$ and $\sigma_{j}(\mathbf{X})$, respectively. For a square matrix $\mathbf{X}$, denote the maximum and minimum eigenvalues by $\lambda_{\max}(\mathbf{X})$ and $\lambda_{\min}(\mathbf{X})$, respectively. Furthermore, for a tensor $\cm{X} \in \mathbb{R}^{p_{1} \times p_{2} \times p_{3}}$, let $\|\cm{X}\|_{\mathrm{F}}=\left(\sum_{i=1}^{p_{1}} \sum_{j=1}^{p_{2}} \sum_{k=1}^{p_{3}} \cm{X}_{i j k}^{2}\right)^{1/2}$
be its Frobenius norm.

Matricization, also known as unfolding or flattening, is the process of reordering the elements of a high-order tensor into a matrix. For any third-order tensor $\cm{X} \in \mathbb{R}^{p_{1} \times p_{2} \times p_{3}}$, its mode-$1$ matricization $\cm{X} _{(1)}$ is the $p_{1}$-by-$(p_{2} p_{3})$ matrix by setting the first tensor mode as its rows and collapsing all the others into its columns. Specially, the $(i_1, i_2, i_3)$-th element of $\cm{X}$ is mapped to the $(i_1, j)$-th element of $\cm{X}_{(1)}$, where
\begin{equation*}
	j=1+\sum_{k=2 }^{3}\left(i_{k}-1\right) J_{k} \quad {\rm{with}} \quad J_{k}=\prod_{\ell=2}^{k-1} p_{\ell}.
\end{equation*}
The mode-$2$ and mode-$3$ matricizations can be defined similarly.
The matricization of tensors helps to link the concepts and properties between matrices and tensors. 
The mode-$1$ multiplication $\times_{1}$ of tensor $\cm{X} \in \mathbb{R}^{p_{1} \times p_{2} \times p_{3}}$ and matrix $\mathbf{Y} \in \mathbb{R}^{q_{1} \times p_{1}}$ is defined as
\begin{equation*}
	\cm{X} \times_{1} \mathbf{Y}=\left(\sum_{i=1}^{p_{1}} \cm{X}_{i j k} \mathbf{Y}_{s i}\right)_{1 \leq s \leq q_{1}, 1 \leq j \leq p_{2}, 1 \leq k \leq p_{3}}\in \mathbb{R}^{q_{1} \times p_{2} \times p_{3}},
\end{equation*}
and we can define the model-$2$ and mode-$3$ multiplication similarly.

For any third-order tensor $\cm{X} \in \mathbb{R}^{p_{1} \times p_{2} \times p_{3}}$, its multilinear
ranks $\left(r_{1}, r_{2}, r_{3}\right)$ are defined as the matrix ranks of its one-mode matricizations, namely
$r_{i}=\operatorname{rank}\left(\cm{X}_{(i)}\right)$ with $1\leq i\leq 3$, and they are not necessarily equal for third- and higher-order tensors. 
Accordingly, there exists a \textit{Tucker decomposition} \citep{Tucker:1966, DeLathauwer2000}:
\begin{equation}
	\cm{X}=\cm{G} \times_{1} \mathbf{U}_{1} \times_{2} \mathbf{U}_{2}  \times_{3} \mathbf{U}_{3}=\cm{G} \times{ }_{i=1}^{3} \mathbf{U}_{i},
	\label{Tucker decomposition}
\end{equation}
where $\cm{G} \in \mathbb{R}^{r_{1} \times r_{2} \times r_{3}}$ is the core
tensor, and $\mathbf{U}_{i} \in \mathbb{R}^{p_{i} \times r_{i}}$ for $i=1, 2, 3$ are the factor matrices. The above decomposition can be denoted by 
$\cm{X}=[\![ \cm{G} ; \mathbf{U}_{1}, \mathbf{U}_{2}, \mathbf{U}_{3} ]\!]$ for simplicity.

Note that, for any nonsingular matrices $\mathbf{O}_{i} \in \mathbb{R}^{r_{i} \times r_{i}}$ with $1\leq i\leq 3$, it holds that
$
	[\![ \cm{G} ; \mathbf{U}_{1}, \mathbf{U}_{2}, \mathbf{U}_{3} ]\!]=[\![  \cm{G} \times_{1} \mathbf{O}_{1} \times_{2}  \mathbf{O}_{2} \times_{3}
	\mathbf{O}_{3}; \mathbf{U}_{1} \mathbf{O}_{1}^{-1}, \mathbf{U}_{2}\mathbf{O}_{2}^{-1}, \mathbf{U}_{3} \mathbf{O}_{3}^{-1} ]\!]
$,
i.e. the decomposition is not unique. 
We can consider the higher-order singular value decomposition (HOSVD) of $\cm{X}$, namely the special Tucker decomposition uniquely defined by choosing $\mathbf{U}_i$ as the tall matrix consisting of the top $r_i$ left singular vectors of $\cm{X}_{(i)}$ and then setting $\cm{G}=\cm{X} \times{ }_{1} \mathbf{U}_{1}^{\top} \times{ }_{2} \mathbf{U}_{2}^{\top} \times{ }_{3} \mathbf{U}_{3}^{\top}$.
The factor matrices $\mathbf{U}_i$'s are then orthonormal, while $\cm{G}$ has the  \textit{all-orthogonal} property: for each $1\leq j\leq 3$, the rows of $\cm{G}_{(j)}$ are pairwise orthogonal.

\subsection{Multilinear low-rank HAR model}
\label{Sect3.1}

Consider $N$ financial assets, and their log prices $\mathbf{X}_{t}=(X_{1, t}, \ldots, X_{N, t})^{\top}$ with $t\in\mathbb{R}_{+}$ are generated by the multivariate HAR-It\^{o} model at \eqref{high-dimensional-har} and \eqref{eq:mhar-ito} with no jump component.
Let $\mathbf{y}_{n}=(y_{1,n}, \cdots, y_{N,n})^{\top} \in \mathbb{R}^{N}$ be the centered $N$-dimensional integrated volatility with $\mathbb{E}(\mathbf{y}_{n})=0$, and we next construct a realized measure for it.

For asset $i$, let $ t_{n,k}\in[n-1,n]$ with $0\leq k\leq m$ be time points with intraday observations, where $m$ is the number of observations within day $n$, and $n-1= t_{n,0} < t_{n,1} < \cdots < t_{n, m}=n$. The realized volatility can be defined as $RV_{i,n}=\sum_{k=1}^{m}(X_{i,t_{n,k}}-X_{i,t_{n,k-1}})^2$, and we focus on the centered version, $\widetilde{y}_{i,n}=RV_{i,n}-T^{-1}\sum_{l=1}^{T}RV_{i,l}$, where $T$ is the total number of days. Denote $\widetilde{\mathbf{y}}_{n}=(\widetilde{y}_{1,n}, \cdots, \widetilde{y}_{N,n})^{\top} \in \mathbb{R}^{N}$, and we then have 
\begin{align}
\widetilde{\mathbf{y}}_n=\mathbf{y}_n+\bm {\eta}_n, \hspace{9mm}  1\leq n \leq T,
\label{estimation error}
\end{align}
where $\bm {\eta}_n=({\eta}_{1,n}, \cdots, {\eta}_{N,n})^{\top} \in \mathbb{R}^{N}$ is the estimation error, and the dependence of notations $\widetilde{\mathbf{y}}_n$ and $\bm {\eta}_n$ on $m$ is suppressed for simplicity.
Suppose that the number of observed prices $m$ is the same for all $1\leq n\leq T$ and $1\leq i\leq N$, and $\{t_{n,k}\}$ are synchronized and equally spaced, i.e. $t_{n,k}-t_{n,k-1}= m^{-1}$.
By a method similar to those in \cite{KW:2016,KW:2023} and \cite{song2020realized}, we can show that 
\begin{equation}
\max\limits_{1 \leq n \leq T} \max\limits_{1 \leq i \leq N}E(|{\eta}_{i,n}|^2)  \leq O( m^{-1 / 2}).
\label{condition_estimation error}
\end{equation}


\begin{remark}
	The assumption of synchronization and equal space is for simplicity only, and it can be relaxed to the generalized sampling time \citep{Ait:2010}, refresh time \citep{Barndorff:2011} or previous tick \citep{Zhang:2011}.
	Moreover, besides the realized volatility, we may consider other realized measures, such as the multi-scale realized volatility \citep{Zh:2006}, kernel realized volatility \citep{B:2008}, and pre-averaging realized volatility \citep{J:2009}.
\end{remark}

From Proposition \ref{prop:2}, the low-frequency integrated volatility satisfies  
\begin{equation}
{\mathbf{y}}_{n}=\sum\limits_{j=1}^{P}\mathbf{A}_j	{\mathbf{y}}_{n-j}+\bm{\varepsilon}_{n}, \hspace{5mm}P+1\leq n\leq T,
\label{autoregressive}
\end{equation}
where the intercept $\bm{\omega}^{g}$ is centered out since $\mathbb{E}(\mathbf{y}_{n})=0$, $\mathbf{A}_j$'s are all $N\times N$ coefficient matrices, and innovations $\{\bm {\varepsilon}_n\}$ are $i.i.d$ with $\bm{\varepsilon}_{n}=({\varepsilon}_{1,n}, \cdots, {\varepsilon}_{N,n})^{\top}\in\mathbb{R}^N$, $\mathbb {E}(\bm{\varepsilon}_{n})=0$ and var$(\bm{\varepsilon}_{n})<\infty$. 
The  matrix polynomial for model \eqref{autoregressive} is defined as $\cm{A}(z)=\mathbf{I}_{N}-\mathbf A_{1} z-\cdots-\mathbf A_{P} z^{P}$, where $z \in \mathbb{C}$ with $\mathbb{C}$ being the complex space.

\begin{assum}
	The determinant of $\cm{A}(z)$ is not equal to zero for all $|z|<1$.
	\label{assumption 1}
\end{assum}

	The above assumption is a necessary and sufficient condition for the strict stationarity of a vector autoregression, and it hence makes sure that $\{\mathbf{y}_n\}$ is strictly stationary.
	In the meanwhile, the integrated volatility $\{\mathbf{y}_n\}$ are generated by the It\^{o} process at \eqref{high-dimensional-har} and \eqref{eq:mhar-ito} rather than the equation at \eqref{autoregressive}, i.e., the sequence $\{\mathbf{y}_n\}$ is not a vector autoregressive process.
	As a result, Assumption \ref{assumption 1} is not necessary, while the stationarity of $\{\mathbf{y}_n\}$ can simplify the presentation, as well as technical details, dramatically.
	When Assumption \ref{assumption 1} does not hold, we cannot center the integrated volatility since it is no longer stationary, and this problem can be solved by using the original integrated volatility at \eqref{estimation error} and then including an intercept at \eqref{autoregressive}. The corresponding theoretical discussions can be referred to \cite{Zheng2021}.

\begin{figure}[H]
	\includegraphics[width=0.9\textwidth]{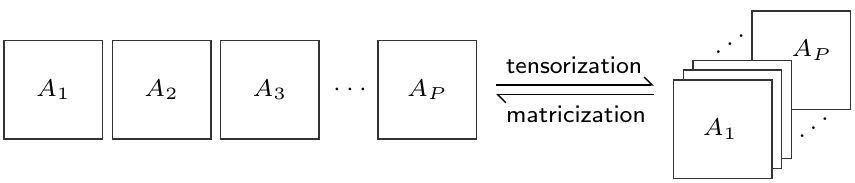}
	\caption{Rearranging $\mathbf{A}_j$s into a third-order tensor \cm{A}.}
	\label{fig_1}
\end{figure}

For model \eqref{autoregressive}, the number of parameters is $N^2P$, which can be very large, and this paper uses tensor techniques to conduct dimension reduction for the parameter space.
Specifically, the coefficient matrices are first rearranged into a third-order tensor $ {\cm{A}} \in \mathbb{R}^{N\times N \times P}$ such that $\cm{A}_{(1)}=(\mathbf {A}_{1}, \ldots, \mathbf{A}_{P})$; see Figure \ref{fig_1} for an illustration.
We then assume multilinear low ranks $(r_1,r_2,r_3)$ to the coefficient tensor $ {\cm{A}}$, and accordingly there exists a Tucker decomposition \citep{DeLathauwer2000},
\begin{equation}\label{tucker-decomposition}
	{\cm{A}}=\mathbf{\cm{G}} \times_{1} \mathbf{U}_{1} \times{ }_{2} \mathbf{U}_{2} \times_{3} \mathbf{U}_{3},
\end{equation}
where $\mathbf{\cm{G}} \in \mathbb{R}^{r_{1} \times r_{2} \times r_{3}}$ is the core tensor, and $\mathbf{U}_{1} \in \mathbb{R}^{N \times r_{1}}$, $\mathbf{U}_{2} \in \mathbb{R}^{N \times r_{2}}$ and $\mathbf{U}_{3} \in \mathbb{R}^{P \times r_{3}}$ are factor matrices.
We call formulas \eqref{estimation error} and  \eqref{autoregressive}, together with the low-rank structure at \eqref{tucker-decomposition}, the \textit{multilinear low-rank HAR} (MLR-HAR) model for simplicity.

Note that $\cm{A}_{(2)}=(\mathbf{A}_1^{\top},\cdots,\mathbf{A}_P^{\top})$ and $\cm{A}_{(3)}=(\vectorize(\mathbf{A}_1),\ldots,\vectorize(\mathbf{A}_P))^{\top}$, and the spaces spanned by $\cm{A}_{(1)}$, $\cm{A}_{(2)}$ and $\cm{A}_{(3)}$ are the column, row and temporal spaces of coefficient matrices, respectively.
As a result, the low-rank assumption at \eqref{tucker-decomposition} restricts the parameter space from three directions simultaneously, and the number of parameters is reduced to $r_{1} r_{2} r_{3}+\left(N-r_{1}\right) r_{1}+\left(N-r_{2}\right) r_{2}+\left(P-r_{3}\right) r_{3}$.

The proposed MLR-HAR model has a form similar to that of vector autoregressive models with measurement errors \citep{Staudenmayer2005} and, by plugging \eqref{autoregressive} into \eqref{estimation error}, we have
\begin{equation}
		\widetilde{\mathbf{y}}_{n}=\sum\limits_{j=1}^{P}\mathbf{A}_j	\widetilde{\mathbf{y}}_{n-j}+\bm{\epsilon}_{n}\hspace{5mm}\text{with}\hspace{5mm}
		\bm{\epsilon}_{n}=\underbrace {\bm{\eta}_n-\sum\limits_{j=1}^{P}\mathbf{A}_j\bm {\eta}_{n-j}}_{estimation \quad  error}+\underbrace{\bm {\varepsilon}_n}_{model \quad error}, 
	\label{MHAR_1}
\end{equation}
where $P+1 \leq n\leq T$.
Suppose that the coefficient tensor admits the HOSVD, ${\cm{A}}=\mathbf{\cm{G}} \times_{1} \mathbf{U}_{1} \times{ }_{2} \mathbf{U}_{2} \times_{3} \mathbf{U}_{3}$, i.e. $\mathbf{\cm{G}}$ is all-orthogonal, and $\mathbf{U}_{j}$'s are orthonormal.
Let $\cm{H}=\cm{G}\times_{3} \mathbf{U}_{3}$, and $\mathbf{H}_j\in\mathbb{R}^{r_1\times r_2}$ be its $j$-th frontal slice for $1\leq j\leq P$, i.e. $\cm{H}_{(1)}=(\mathbf{H}_1,\mathbf{H}_2,\cdots,\mathbf{H}_P)$. Thus, ${\cm{A}}=\mathbf{\cm{H}} \times_{1} \mathbf{U}_{1} \times{ }_{2} \mathbf{U}_{2} $, and we can rewrite model \eqref{MHAR_1} into
\begin{equation*}\label{eq:var-factor}
\widetilde{\mathbf{y}}_{n}=\mathbf{U}_1\sum_{j=1}^{P}\mathbf{H}_j \mathbf{U}_2^{\top} \widetilde{\mathbf{y}}_{n-j}+\bm{\epsilon}_n \hspace{5mm}\text{or}\hspace{5mm}
\mathbf{U}_1^{\top}\widetilde{\mathbf{y}}_{n}=\sum_{j=1}^{P}\mathbf{H}_j \mathbf{U}_2^{\top} \widetilde{\mathbf{y}}_{n-j}+\mathbf{U}_1^{\top}\bm{\epsilon}_n,
\end{equation*}
where $\mathbf{U}_1^{\top}\widetilde{\mathbf{y}}_n$ and $\mathbf{U}_2^{\top} \widetilde{\mathbf{y}}_{n-j}$ are the summarized factors of responses and predictors, respectively.
The HAR model in \cite{Cubadda:2017} corresponds to the case with $r_1=N$ and $\mathbf{U}_1$ being an identity matrix, and it hence has more parameters than the proposed MLR-HAR model; see Remark \ref{rem-sec2} for more discussions.

\begin{remark}\label{rem-factor}
Let $\cm{S}=\cm{G}\times_{1} \mathbf{U}_{1}\times_{2} \mathbf{U}_{2}\in\mathbb{R}^{N\times N\times r_3}$, and  $\mathbf{S}_j\in\mathbb{R}^{N\times N}$ be its $j$-th frontal slice for $1\leq j\leq r_3$, i.e. $\cm{S}_{(1)}=(\mathbf{S}_1,\mathbf{S}_2,\cdots,\mathbf{S}_{r_3})$. 
Denote $\mathbf{U}_{3}=(\mathbf{u}^{(1)},\ldots,\mathbf{u}^{(r_3)})\in\mathbb{R}^{P\times r_3}$ and $\mathbf{u}^{(k)}=(u_1^{(k)},\ldots,u_P^{(k)})^\top\in\mathbb{R}^P$, where $1\leq k\leq r_3$.
As a result, ${\cm{A}}=\mathbf{\cm{S}} \times_{3} \mathbf{U}_{3}$, and model \eqref{MHAR_1} can be reformulated into
\begin{equation}\label{sec3-eq1}
	\widetilde{\mathbf{y}}_{n}= \mathbf{S}_1\widetilde{\mathbf{x}}_{n}^{(1)}+ \cdots+ \mathbf{S}_{r_3}\widetilde{\mathbf{x}}_{n}^{(r_3)}+ \bm{\epsilon}_{n} \hspace{5mm}\text{with}\hspace{5mm}
	\widetilde{\mathbf{x}}_{n}^{(k)}=\sum_{j=1}^{P}u_j^{(k)}\widetilde{\mathbf{y}}_{n-j}.
\end{equation}
Just like the daily, weekly and monthly realized volatilities, $\widetilde{\mathbf{x}}_{n}^{(k)}$'s are the summarized factors along the temporal direction, and they can be treated as $r_3$ heterogeneous volatility components, which are automatically selected by the estimation method.
Moreover, $\mathbf{U}_{3}$ is the corresponding loading matrix, and model \eqref{sec3-eq1} will reduce to the vector HAR model \citep{Bubak:2011, SoucekandTodorova:2013} when $\mathbf{U}_{3}=\mathbf{U}_{\mathrm{C}}$ at \eqref{har-rv}.
\end{remark}

\subsection{Ordinary least squares estimation}
\label{Sect3.3}

Suppose that the multilinear ranks $(r_1, r_2, r_3)$ of the coefficient tensor $\cm{A}$ are known.
From \eqref{tucker-decomposition} and \eqref{MHAR_1}, the ordinary least squares (OLS) estimator for MLR-HAR models can be defined as
\begin{align*}
	\widehat{{\cm{A}}}_{\rm MLR}\equiv [\![\widehat{\mathbf{\cm{G}}}; \widehat{\mathbf{U}}_{1}, \widehat{\mathbf{U}}_{2}, \widehat{\mathbf{U}}_{3}]\!]=	 \mathop{\argmin}L(\cm{G}, {\mathbf{U}}_{1}, {\mathbf{U}}_{2}, {\mathbf{U}}_{3}),
\end{align*}
where $\widetilde{\mathbf{x}}_{n}=(\widetilde{\mathbf{y}}_{n-1}^{\top},\ldots,\widetilde{\mathbf{y}}_{n-P}^{\top})^\top$, and
\begin{align*}
	L(\cm{G}, {\mathbf{U}}_{1}, {\mathbf{U}}_{2}, {\mathbf{U}}_{3})=\frac{1}{T}\sum_{n=P+1}^{T}||	\widetilde{\mathbf{y}}_{n}-(\mathbf{\cm{G}}\times _{1}\mathbf{U}_{1}\times _{2}\mathbf{U}_{2}\times_{3}\mathbf{U}_{3})_{(1)}\widetilde{\mathbf{x}}_{n}||_2^2.
\end{align*}
Although the components of Tucker decomposition, ${\mathbf{\cm{G}}}$, ${\mathbf{U}}_{1}$, ${\mathbf{U}}_{2}$ and ${\mathbf{U}}_{3}$, are not identifiable, $\cm{A}$ can be uniquely identified. 
For the low-dimensional case with both $N$ and $P$ being fixed, this subsection establishes the asymptotic normality of $\widehat{{\cm{A}}}_{\rm MLR}$ by adapting the technique for overparameterized models in \cite{shapiro:1986}.

Let $\bm{\varphi}=\left(\rm vec(\mathbf{\cm{G}_{(1)})}^{\top},  \rm vec(\mathbf{U}_1)^{\top},  \rm vec(\mathbf{U}_2)^{\top},   \rm vec(\mathbf{U}_3)^{\top}\right)^{\top}$, and $\bm{h}(\bm{\varphi})=\rm vec(\mathbf{\cm{A}}_{(1)})=\rm vec(\mathbf{U}_1\mathbf{\cm{G}}_{(1)}(\mathbf{U}_{3} \otimes\mathbf{U}_2 )^{\top})$
be a function of $\bm{\varphi}$.
Denote $\bm{\Sigma}_{{\varepsilon}}=\rm var(\bm{\varepsilon}_{n})$ and
\begin{align*}
	\bm{\Gamma}^{*}=
	\begin{pmatrix}
		\bm{\Gamma}_0 & \bm{\Gamma}_1 &\cdots& \bm {\Gamma}_{P-1}&\\
		\bm {\Gamma}_1^{\top} & \bm {\Gamma}_0 &\cdots& \bm {\Gamma}_{P-2}&\\
		\vdots & \vdots &\ddots&\vdots&\\
		\bm {\Gamma}_{P-1}^{\top} & \bm {\Gamma}_{P-2}^{\top} &\cdots& \bm{ \Gamma}_{0}&
	\end{pmatrix}\in\mathbb{R}^{NP \times NP},
\end{align*} 
where $\bm{\Gamma}_{j}=\mathrm{cov}(\mathbf{y}_{n+j},\mathbf{y}_n)\in\mathbb{R}^{N\times N}$ with $j  \geq 0$.
As a result, the Jacobian matrix $\mathbf{H}:={\partial \bm{h}(\bm{\varphi})}/{\partial\bm{\varphi}}$ has the form of
\begin{align*}
	\mathbf{H}=&\left( (\mathbf{U}_{3} \otimes \mathbf{U}_{2} \otimes \mathbf{U}_{1}),[( \mathbf{U}_{3}  \otimes \mathbf{U}_{2}) \mathbf{\cm{G}}_{(1)}^{\top}] \otimes\mathbf{ I}_{N}, \mathbf{T}_{21}	
	\left\{\left[(\mathbf{U}_{3} \otimes \mathbf{U}_{1} ) \mathbf{\cm{G}}_{(2)}^{\top}\right]\otimes\mathbf{I}_{N}\right\},\right.\\
	& \hspace{5mm} \left.	 \mathbf{T}_{31}\left\{\left[(\mathbf{U}_{2} \otimes \mathbf{U}_{1} ) \mathbf{\cm{G}}_{(3)}^{\top}\right] \otimes\mathbf{ I}_{P}\right\}
	\right) \in\mathbb{R}^{N^2P \times (r_1 r_2 r_3+N r_1+N r_2+ P r_3)},
\end{align*}
where $\mathbf{I}_{\ell}\in\mathbb{R}^{\ell \times \ell}$ is an identity matrix, and $\mathbf{T}_{ij}\in\mathbb{R}^{(N^2P)\times (N^2P)}$ is a permutation matrix such that $\mathrm{vec}({\cm{A}}_{(j)})=\mathbf{T}_{ij}\mathrm{vec}({\cm{A}}_{(i)})$ for $1\leq i,j\leq 3$.
Moreover, let $\mathbf{J}=\bm\Gamma^{*} \otimes \bm{\Sigma}_{\varepsilon}^{-1}$, and $\bm{\Sigma}_{\rm MLR} =\mathbf{H}(\mathbf {H}^{\top}\mathbf{J}\mathbf {H})^{\dagger}\mathbf{H}^{\top}$,  where $\dagger$ denotes the Moore-Penrose inverse.

%
%

\begin{thm}
	Suppose that \eqref{condition_estimation error} and Assumption \ref{assumption 1} hold, $\mathbb{E}\|\bm{\varepsilon}_{n}\|^4<\infty$, and both $N$ and $P$ are fixed.
	If $m  \rightarrow  \infty$, $T\rightarrow  \infty$ and $T^{4+\delta}m^{-1} \rightarrow  0$ for some $\delta>0$, then
	\begin{equation*}
		\sqrt{T}\{{\rm{vec}}((\widehat{\cm{A}}_{\rm{MLR}})_{(1)})-\rm{vec}(\mathbf{\cm {A}}_{(1)})\}  \rightarrow  N(\mathbf{0}, {\bm{\Sigma}}_{\rm{MLR}})
	\end{equation*}
	in distribution, where $\mathbf{0}$ is the vector of zeros.
	\label{theorem 1}
\end{thm}

The asymptotic variance matrix ${\bm{\Sigma}}_{\rm{MLR}}$ in the above theorem is degenerated since $\widehat{\cm{A}}_{\rm{MLR}}$ has low Tucker ranks.
Moreover, the number of intraday observations, $m$, is required to diverge with a rate faster than $T^4$, and this condition can be relaxed if we can achieve a higher order moment on estimation errors at \eqref{condition_estimation error}.

As a comparison, we also consider the cases without low-rank constraint and with low-rankness on $\cm {A}_{(2)}$ only, and their OLS estimators are given below,
\begin{equation*}
\widehat{\cm{A}}_{\rm{OLS}}= \mathop{\argmin} \frac{1}{T}\sum_{n=P+1}^{T}|| \widetilde{\mathbf{y}}_n- \cm{A}_{(1)}\widetilde{\mathbf{x}}_n	||_2^2\hspace{2mm}\text{and}\hspace{2mm}
\widehat{\cm{A}}_{\rm{MRI}}= \mathop{\argmin}\limits_{ {\rm {rank}}( \cmt{A}_{(2)}) \leq r_2}\frac{1}{T}\sum_{n=P+1}^{T} || \widetilde{\mathbf{y}}_n-  \cm{A}_{(1)}\widetilde{\mathbf{x}}_n	||_2^2	,
\end{equation*}
where $	\widehat{\cm{A}}_{\rm{MRI}}$ corresponds to the multivariate autoregressive index model (MRI) in \cite{Reinsel:1983}.
For two positive semi-definite matrices ${\bm{\Sigma}}_1$ and ${\bm{\Sigma}}_2$, denote ${\bm{\Sigma}}_1\leq {\bm{\Sigma}}_2$ if ${\bm{\Sigma}}_2-{\bm{\Sigma}}_1$ is positive semi-definite.
The next corollary theoretically verifies that $\widehat{\cm{A}}_{\rm{MLR}}$ is the most efficient, while $\widehat{\cm{A}}_{\rm{OLS}}$ performs worst.


\begin{cor}
	If the conditions of Theorem \ref{theorem 1} hold, then $	\sqrt{T}\{{\rm{vec}}((\widehat{\cm{A}}_{\rm{OLS}})_{(1)})-{\rm{vec}}(\mathbf{\cm {A}}_{(1)})\}  \rightarrow  N(\mathbf{0}, {\bm{\Sigma}}_{\rm{OLS}})$ and
	$	\sqrt{T}\{{\rm{vec}}((\widehat{\cm{A}}_{\rm{MRI}})_{(1)})-{\rm{vec}}(\mathbf{\cm {A}}_{(1)})\}  \rightarrow  N(\mathbf{0}, {\bm{\Sigma}}_{\rm{MRI}})$ in
	distribution as $m \rightarrow \infty$ and $T \rightarrow \infty$, where $\bm{\Sigma}_{\rm{OLS}} =\mathbf{J}^{-1}$ and ${\bm{\Sigma}}_{\rm{MRI}}$ is defined in the proof. Moreover, it holds that ${\bm{\Sigma}}_{\rm{MLR}} \leq{\bm{\Sigma}}_{\rm{MRI}} \leq {\bm{\Sigma}}_{\rm{OLS}}$.
	\label{corollary 1}
\end{cor}

\section{High-dimensional HAR modeling}\label{Sect4}

\subsection{High-dimensional HAR modeling}
\label{Sect4.1}

It is common to encounter many assets in real applications, and the number of assets $N$ can be very large, say growing with sample size $T$ with arbitrary rates. This corresponds to the high-dimensional setting, and the derived asymptotic properties in the previous section are no longer satisfied. This section alternatively establishes the non-asymptotic properties of the OLS estimation for the high-dimensional case. 


\begin{assum}
	Model error $\bm{\varepsilon}_{n} = \bm{\Sigma}_{\varepsilon}^{1/2} \bm{\xi}_{n}$, where $\{\bm{\xi}_{n}\}$ are $i.i.d.$ random vectors with $\mathbb{E}(\bm{\xi}_{n})=0$, ${\rm {var}}(\bm{\xi}_{n})=\mathbf{I}_{N}$, and $\bm{\Sigma}_{\varepsilon}={\rm{var}}(\bm{\varepsilon}_{n})$ is a positive definite matrix. In addition, the entries $(\bm{\xi}_{in})_{1\leq i \leq N}$ of $\bm{\xi}_{n}$ are mutually independent and $\kappa^2$-sub-Gaussian, and model errors $\{\bm{\varepsilon}_{n}\}$ are independent of estimation errors $\{\bm{\eta}_n\}$.
	\label{sub-Gaussian}
\end{assum}

The sub-Gaussianity in the above assumption is commonly used for high-dimensional settings in the literature \citep{Wainwright2019}. The independence between model errors and estimation errors is mainly used to simplify the technical proofs for theorems in this section, and it can be relaxed with lengthy proofs.

We next derive the non-asymptotic error bounds, which will rely on the temporal and cross-sectional dependence of $\{\widetilde{\mathbf{y}}_{n}\}$ \citep{basu2015regularized}.
To this end, two dependence measures are first defined below,
\begin{align*}
\mu_{\min}(\cm{A}):= \min_{|z|=1} \lambda_{\min}(\cm{A}^{*}(z)\cm{A}(z))\quad \text{and}\quad 
\mu_{\max}(\cm{A}):= \max_{|z|=1} \lambda_{\max}(\cm{A}^{*}(z)\cm{A}(z)),
\end{align*}
where the matrix polynomial $\cm{A}(z)=\mathbf{I}_{N}-\mathbf A_{1} z-\cdots-\mathbf A_{P} z^{P}$, $\cm{A}^{*}(z)$ is the conjugate transpose of $\cm{A}(z)$. For any two sequences $\{a_n\}$ and $\{b_n\}$, denote by $a_n \lesssim b_n$ (or $a_n \gtrsim b_n$) if there exists a constant $C$ such that $a_n\leq C b_n$ (or $a_n\geq C b_n$) for all $n$. Let $\kappa_L = \lambda_{\min}(\bm{\Sigma}_{\varepsilon})/\mu_{\max}(\cm{A})$, $\kappa_U = \lambda_{\max}(\bm{\Sigma}_{\varepsilon})/\mu_{\min}(\cm{A})$, and $d_{\mathcal{M}} = r_1r_2r_3+Nr_1 + Nr_2 + Pr_3$ be the number of parameters for MLR-HAR models.

\begin{thm}
	Suppose that \eqref{condition_estimation error} and Assumptions \ref{assumption 1} and \ref{sub-Gaussian} hold. If the sample size $T \gtrsim \max(\kappa^2,\kappa^4) (\kappa_U/\kappa_L)^{2}d_{\mathcal{M}}$, $m^{1/4}T^{2\delta} \gtrsim N^2P\exp(d_{\mathcal{M}})$ and $m^{1/4} \gtrsim T^{1+2\delta}$ for some $\delta>0$, then
	\[
	\|\widehat{{\cm{A}}}_{\rm MLR} - \cm{A}\|_{\rm F} \leq  \frac{C}{\kappa_L} \left[(\kappa^2 \sqrt{\lambda_{\max}(\bm{\Sigma}_{\varepsilon})\kappa_U} +\kappa\sqrt{\kappa_U})\sqrt{\frac{d_{\mathcal{M}}}{T}} + \frac{T^{1+2\delta}}{m^{1/4}}\right] \quad \text{and}
	\]
	\[
	\frac{1}{T}\sum_{n=P+1}^T\|(\widehat{{\cm{A}}}_{\rm MLR})_{(1)}\widetilde{\mathbf{x}}_{n} - \cm{A}_{(1)}\widetilde{\mathbf{x}}_{n}\|_{2}^2 \leq  \frac{C}{\kappa_L}  \left[(\kappa^2 \sqrt{\lambda_{\max}(\bm{\Sigma}_{\varepsilon})\kappa_U} +\kappa\sqrt{\kappa_U})\sqrt{\frac{d_{\mathcal{M}}}{T}} + \frac{T^{1+2\delta}}{m^{1/4}}\right]^2,
	\]
	with probability at least $1- \exp(-C d_{\mathcal{M}})-2\exp(-CT(\kappa_L/\kappa_U)^2\min\{\kappa^{-2},\kappa^{-4}\})$, where $C$ is a positive constant given in the proof.
	\label{thm2}
\end{thm}

The above theorem provides the upper bounds of both estimation and prediction errors, and they consist of two terms: the first one is due to the OLS estimation error for model \eqref{autoregressive}, and the second is caused by using realized volatilities to estimate integrated volatilities at \eqref{estimation error}.
When $\kappa_L$ and $\kappa_U $ are bounded away from zero and infinity, the estimation error $\|\widehat{{\cm{A}}}_{\rm MLR} - \cm{A}\|_{\rm F} =O_P(\sqrt{d_{\mathcal{M}}/T}+{T^{1+2\delta}}/{m^{1/4}})$, where $d_{\mathcal{M}}$ measures the complexity of MLR-HAR models, and the term of ${T^{1+2\delta}}/{m^{1/4}}$ converges to zero when $m$ diverges with a rate faster than $T^4$ as in the low-dimensional case in Theorem \ref{theorem 1}.
Similarly, the prediction error $T^{-1}\sum_{n=P+1}^T\|(\widehat{{\cm{A}}}_{\rm MLR})_{(1)}\widetilde{\mathbf{x}}_{n} - \cm{A}_{(1)}\widetilde{\mathbf{x}}_{n}\|_{2}^2= O_P({d_{\mathcal{M}}/T}+{T^{2+4\delta}}/{m^{1/2}})$, and their consistency can be achieved if  $m  \rightarrow  \infty$, $T\rightarrow  \infty$, ${d_{\mathcal{M}}/T}\rightarrow  0$ and $T^{4+8\delta}m^{-1} \rightarrow  0$.

\subsection{Projected gradient descent algorithm}
\label{Sect4.2}

From Section 3.3, it is a nonconvex problem to search for the OLS estimator $\widehat{{\cm{A}}}_{\rm MLR}$, and this makes the parameter estimation challenging numerically and theoretically. This subsection introduces a projected gradient descent (PGD) algorithm by adopting the method in \cite{chen2019non}, and its theoretical guarantee is also provided. 

Consider the parameter space of MLR-HAR models at  \eqref{estimation error}, \eqref{autoregressive} and \eqref{tucker-decomposition},
\begin{align*}
	\bm{\Theta}(r_1,r_2,r_3)=\{{\cm{A}} \in \mathbb{R}^{N\times N \times P}: \textrm{rank}({\cm{A}_{(i)}}) \leq r_i \quad {\rm {for}} \quad 1\leq i \leq 3 \}.
\end{align*}
We first introduce a projection of any tensor $\cm{B} \in \mathbb{R}^{N\times N \times P}$ onto $\bm{\Theta}(r_1,r_2,r_3)$.
For $1\leq i \leq 3$, let $\mathcal{M}_i$ be the matricization operator, which maps a tensor to its mode-$i$ matricization, and $\mathcal{M}_i^{-1}$ be the inverse operator, i.e. $\mathcal{M}_i(\cm{B})=\cm{B}_{(i)}$ and $\mathcal{M}_i^{-1}(\cm{B}_{(i)})=\cm{B}$.
Moreover, denote by $P_r$ a projection operator, which maps a matrix to its best rank $r$ approximation. Specifically, for a matrix, $P_r$ first conducts the SVD to it, and then the $r$ largest singular values are kept while the others are suppressed to zero.   
As a result, for a tensor $\cm{B} \in \mathbb{R}^{N\times N \times P}$, we can define its projection onto $\bm{\Theta}(r_1,r_2,r_3)$ below,
\begin{equation*}
	\widehat{P}_{\bm{\Theta}(r_1,r_2,r_3)}({\cm{B}}) :=  (\mathcal{M}_3^{-1}\circ P_{r_3} \circ \mathcal{M}_3) \circ (\mathcal{M}_2^{-1}\circ P_{r_2} \circ \mathcal{M}_2) \circ (\mathcal{M}_1^{-1}\circ P_{r_1} \circ \mathcal{M}_1) ({\cm{B}}).
\end{equation*}
Specifically, we first calculate mode-1 matricization of $\cm{B}$, then find out the best rank $r_1$ approximation by the SVD, and finally fold it back to a third-order tensor. The same action is further applied to the second and third modes sequentially; see Algorithm \ref{alg2} for details.
The order of which matricization is performed is nonessential, and it will not affect the forthcoming convergence analysis.
Moreover, $\widehat{P}_{\bm{\Theta}(r_1,r_2,r_3)}(\cdot)$ is an approximate projection onto $\bm{\Theta}(r_1,r_2,r_3)$ only, while the exact projection is well known to be an NP-hard problem \citep{hillar2013most}. 

\begin{algorithm}
	\caption{Projected gradient descent algorithm for HAR modeling}
	\begin{algorithmic}[H]
		\State {\bf Input :} data $\{\widetilde{\mathbf{y}}_n\}$, parameter space $\bm{\Theta}=\bm {\Theta}(r_{1},r_{2},r_{3})$, iterations $K$, step size $\eta$
		\State {\bf Initialize :} $k=0$ and $\widehat{{\cm{A}}}_0 \in \bm{\Theta}$.
		\For{$k=1,2,\dots, K$}
		\State $\widetilde{{\cm{A}}}_{k} = \widehat{{\cm{A}}}_{k-1} - \eta \nabla L(\widehat{{\cm{A}}}_{k-1})$ (gradient descent)
		\For{$j=1,2,3$}
		\State $\bm{B}_j = \mathcal{M}_j(\widetilde{{\cm{A}}}_{k})$ (mode-$j$ matricization)
		\State $\widehat{\bm{B}}_j =  P_{r_j}(\bm{B}_j)$ (best rank $r_j$ approximation by the SVD)
		\State $\widehat{{\cm{A}}}_{k} = \mathcal{M}_j^{-1}(\widehat{\bm{B}}_j)$ (fold into tensor by reversing the mode-$j$ matricization)
		\EndFor
		\EndFor 
		\State {\bf Output :} $\widehat{{\cm{A}}}_{K}$
	\end{algorithmic}
	\label{alg2}
\end{algorithm}

It is ready to introduce the PGD method to search for the OLS estimator $\widehat{{\cm{A}}}_{\rm MLR}$; see Algorithm \ref{alg2}.
Specifically, we first update the estimate by the commonly used gradient descent method, and the updated tensor is then projected onto $\bm{\Theta}(r_1,r_2,r_3)$ since it may not have low Tucker ranks.
Let $(r_1',r_2',r_3')$ be the running Tucker ranks used in Algorithm \ref{alg2}, and denote $d_{\mathcal{M}}^{\prime} = (r_1+r_1')(r_2+r_2')(r_3+r_3')+N(r_1+r_1') + N(r_2+r_2') + P(r_3+r_3')$.

\begin{thm}
	Suppose that \eqref{condition_estimation error} and Assumptions \ref{assumption 1} and \ref{sub-Gaussian} hold, step size $\eta = {2}/({3\kappa_U})$, and the running Tucker ranks  $r_i^\prime\geq\left(\sqrt[3]{1+{\kappa_{L}}/({24\kappa_{U}}})-1\right)^{-2}r_i$ with $1\leq i\leq 3$.
	If $T \gtrsim \max(\kappa^2,\kappa^4)$ $(\kappa_U/\kappa_L)^2d_{\mathcal{M}}^{\prime}$, $m^{1/4}T^{2\delta} \gtrsim N^2P\exp(d_{\mathcal{M}}^{\prime})$ and $m^{1/4} \gtrsim T^{1+2\delta}$ for some $\delta>0$, then
	\begin{equation*}
	\|\widehat{{\cm{A}}}_K- {\cm{A}}\|_{\rm F} \le (1-\frac{\kappa_{L}}{24\kappa_{U}})^K \|\widehat{\cm{A}}_0 - {\cm{A}}\|_{\rm F} + \frac{C}{\kappa_L} \left[(\kappa^2 \sqrt{\lambda_{\max}(\bm{\Sigma}_{\varepsilon})\kappa_U} +\kappa\sqrt{\kappa_U})\sqrt{\frac{d_{\mathcal{M}}^{\prime}}{T}} +\frac{T^{1+2\delta}}{m^{1/4}}\right],
	\end{equation*}
	with probability at least $1- \exp(-C d_{\mathcal{M}}^{\prime})-2\exp(-CT(\kappa_L/\kappa_U)^2\min\{\kappa^{-2},\kappa^{-4}\})$, where $C$ is a positive constant given in the proof. 
	\label{thm3}
\end{thm}

The two terms of the upper bound in the above theorem correspond to the optimization and statistical errors, respectively, and the statistical error has a form similar to that in Theorem \ref{thm2}.
Note that $\kappa_{L}<\kappa_{U}$, and hence the linear convergence rate is implied for the optimization error. Specifically, for any $\epsilon>0$, we can choose the number of iterations $K=[\log(\epsilon)-\log \|\widehat{\cm{A}}_0 - {\cm{A}}\|_{\rm F}]/\log[1-{\kappa_{L}}/{(24\kappa_{U})}]$ such that the optimization error is smaller than $\epsilon$.
For the initial value $\widehat{{\cm{A}}}_0$, we may simply set it to zero in practice.
Finally, Tucker ranks of $\cm{A}$ are usually unknown in real applications, and they can be chosen empirically or by a high-dimensional Bayesian information criterion (BIC),
\begin{equation}\label{eq:BIC}
	{\rm BIC}({\mathbf{r}}) = \log\left\{\frac{1}{T}\sum_{n=P+1}^T\|\tilde{\mathbf{y}}_n - \widehat{\cm{A}}({\mathbf{r}})\tilde{\mathbf{x}}_n\|_2^2\right\} + \frac{\lambda d_{\mathcal{M}}({\mathbf{r}})\log(T)}{T},
\end{equation}
where $\widehat{\cm{A}}({\mathbf{r}}) = \widehat{\cm{A}}_{\rm MLR}$ with Tucker ranks ${\mathbf{r}}=({r}_1,{r}_2,{r}_3)$, $d_{\mathcal{M}}({\mathbf{r}}) = {r}_1{r}_2{r}_3 + N{r}_1+N{r}_2 + P{r}_3$ is the number of parameters, and $\lambda$ is a tuning parameter.

\section{Simulations studies} \label{Sect5}
Three simulation experiments are conducted in this section: the first two are to evaluate the finite-sample performance of OLS estimation for the proposed MLR-HAR model under low- and high-dimensional settings in Sections 3.3 and 4.1, respectively, and the third one is to verify the convergence of the proposed algorithm in Section 4.2.

In the first experiment, the high-frequency data are generated by using the multivariate HAR-It\^o model at \eqref{high-dimensional-har} and \eqref{eq:mhar-ito} with the absence of jump components and drift terms, i.e. $L_{i,t}=0$ and $\mu_{i,t}=0$ with $1\leq i\leq N$.
The dimension is fixed at $N=5$, and we set the time interval to $\Delta=1/780$ during discretization.  
For the two Brownian motions ${\mathbf {B}}_{t}$ and $\mathbf{W}_t$, the increments ${\mathbf {B}}_{t+\Delta}-{\mathbf {B}}_{t}$ and ${\mathbf {W}}_{t+\Delta}-{\mathbf {W}}_{t}$ follow multivariate normal distributions with mean zero, variance matrix $\Delta \cdot\mathbf{I}_N$, and their correlation coefficient matrix being $-0.6\cdot\mathbf{I}_N$.
We fix the initial log prices to $X_{i,0}=\log(50)$ and the initial instantaneous volatility to $\sigma_{i,0}=0.1$ for all $1\leq i\leq N$.  
The model parameters are set to $(\omega_{i},v_i)=(0.2,0.4)$ with $1\leq i\leq N$, and 
$(\bm {\alpha}^{(1)}, \ldots, \bm {\alpha}^{(P)}) = \bm {\alpha}^*  (\mathbf {U}_{\mathrm{C}}^{\top} \otimes \mathbf{I}_N)\in \mathbb{R}^{N\times N\times 22}$, where $\mathbf {U}_{\mathrm{C}} \in \mathbb{R}^{22\times 3} $ is defined in \eqref{har-rv}, and $\bm {\alpha}^*\in \mathbb{R}^{N\times 3N}$ is a randomly generated matrix with rank two and Assumption \ref{assumption 1} being satisfied.
From Remark \ref{rem-sec2}, we can calculate the coefficient tensor, $\cm{A} \in \mathbb{R}^{N\times N \times 22}$, of its low-frequency representation, and it can be verified to have the low Tucker rank of $(r_{1},r_{2},r_{3})=(2,2,3)$.

We consider five different sample sizes, $T=500(1+j)$ with $1\leq j\leq 5$, and there are 2000 replications for each setting.
The realized volatility is first calculated with the number of intraday observations being $m=78$ or $780$, and then
Algorithm \ref{alg2} is employed to search for the estimate $\widehat{\cm{A}}_{\rm MLR}$ with step size $5\times 10^{-4}$, tolerance $10^{-7}$ and initial values $\widehat {\cm{A}}_{0}=0$.
As a comparison, we also calculate the estimators without low-rank constraint and with low-rankness on $\cm {A}_{(2)}$ only, i.e. $\widehat{\cm{A}}_{\rm OLS}$ and $\widehat{\cm{A}}_{\rm MRI}$ in Section 3.3.
Moreover, the asymptotic variance matrices of the three estimators can be obtained according to Theorem \ref{theorem 1} and Corollary \ref{corollary 1}.
Figure \ref{thm1-plot} presents the maximum singular value of empirical variance (EVar) matrices and averaged maximum singular values of estimated asymptotic variance (AVar) matrices. The bias is also given in terms of averaged absolute deviations from the true coefficient tensor, and it is squared to be comparable with the EVar and AVar.
We have four findings below. (i.) All bias and variance go to zero as the sample size $T$ increases, and we may conclude the consistency of the three estimators.
(ii.) Comparing with $\widehat{\cm{A}}_{\rm OLS}$ and $\widehat{\cm{A}}_{\rm MRI}$, $\widehat{\cm{A}}_{\rm MLR}$ has the smallest bias and variances, and it is consistent with our intuition that $\widehat{\cm{A}}_{\rm MLR}$ makes use of more low-rank structures.
(iii.) The EVar generally matches the corresponding AVar well, with their differences getting smaller as $T$ increases, although the AVar tends to underestimate the variances for all cases. 
(iv.) Finally, all three estimators have slightly smaller bias, EVar and AVar with a larger number of intraday observations, i.e. $m=780$. 

The second experiment is designed to verify the non-asymptotic estimation error bound in Theorem \ref{thm2}, which consists of two parts, $O_p(\sqrt{d_{\mathcal{M}}/T})$ and $O_p(T^{1+2\delta}/m^{1/4})$. The two parts are due to the model error in low-frequency representations and the high-frequency error brought in by the realized volatility, respectively, and they are hard to split and study separately.
As a result, we consider two different data generating processes to evaluate each of them individually.
The first data generating process produces the low-frequency data $\{y_{i,n}\}$ directly from model \eqref{autoregressive} with no high-frequency error involved, i.e. $m=\infty$, and hence the error bound reduces to $O_p(\sqrt{d_{\mathcal{M}}/T})$.
The coefficient tensor has the form of ${\cm{A}}=\mathbf{\cm{G}} \times_{1} \mathbf{U}_{1} \times_{2} \mathbf{U}_{2} \times_{3} \mathbf{U}_{3}\in \mathbb{R}^{N\times N \times 22}$, where the entries of core tensor $\cm{G}\in\mathbb{R}^{r_1\times r_2\times r_3}$ are first generated independently from the standard normal distribution and then rescaled such that $\| \cm{G} \|_{\rm F}=0.5$, and the factor matrices $\mathbf{U}_i$'s are generated by extracting the first $r_i$ left singular vectors of Gaussian random matrices while ensuring the stationarity condition in Assumption \ref{assumption 1}.
The error terms $\{\bm{\varepsilon}_{n}\}$ follow the multivariate standard normal distribution, and the Tucker rank is set to $(r_1,r_2,r_3) = (2,2,2)$ or $(3,3,3)$.
We consider three dimensions, $N=10$, 20 and 25, and five sample sizes, $T=50(4+3j+j^2)$ with $1\leq j\leq 5$, and
Algorithm \ref{alg2} is used again to search for the estimate $\widehat{\cm{A}}_{\rm MLR}$.
Figure \ref{thm2-plot} gives the estimation error $\| \widehat{\cm{A}}_{\rm MLR}-\cm{A} \|_{\rm F}$, averaged over 500 replications, and its linearity with respect to $\sqrt{d_{\mathcal{M}}/T}$ can be observed.
In fact, the estimation error will approach zero as the sample size $T$ keeps increasing.
As a result, the first part of error bounds is hence confirmed.
It can also be seen that these lines have different slopes, and this is due to the fact that the constant terms in the bound at Theorem \ref{thm2}, such as $\kappa$, $\kappa_{L}$ and $\kappa_{U}$, may vary for different dimensions of $N$. 

For the second data generating process in the second experiment, we generate an $i.i.d.$ high-frequency estimation error sequence $\{\bm\eta_n\}$ to model \eqref{estimation error}, where $\bm\eta_n$ follows multivariate normal distribution with mean zero and variance $m^{-1/2}\cdot \mathbf{I}_N$, and $\{\mathbf{y}_n\}$ are generated from the first data generating process.
We set five sample sizes $T=50(1+j)$ with $1\leq j \leq 5$, and all the other settings are the same as those for the first data generating process.
The realized volatility is calculated with the number of intraday observations being $m=T^4$. 
Note that, from Theorem \ref{thm2}, $\sqrt{T/d_{\mathcal{M}}}\| \widehat{\cm{A}}_{\rm MLR}-\cm{A} \|_{\rm F} =O_p(1)+O_p (T^{3/2+2\delta}/m^{1/4})$, and hence we plot the adjusted estimation error $\sqrt{T/d_{\mathcal{M}}}\| \widehat{\cm{A}}_{\rm MLR}-\cm{A} \|_{\rm F}$ against $T^{3/2}/m^{1/4}$ in Figure \ref{thm2-plot}. The clear linearity confirms the latter part of Theorem \ref{thm2}. Moreover, these lines have different intercepts, due to the term of $O_p(1)$.

The third experiment is to evaluate the convergence performance of Algorithm $\ref{alg2}$ in Section 4.2. A sample is generated using the data generating process in the first experiment with $(N,T)=(30,1000)$, and realized volatility is calculated with $m=78$, 390 or 780.
Note that the true Tucker ranks are $(r_1, r_2, r_3) = (2, 2, 3)$, and we consider three different running ranks, $(r_1^{\prime}, r_2^{\prime},  r_3^{\prime}) = (2, 2, 3)$, $(5, 5, 5)$ and $(10, 10, 10)$, in the algorithm.
Figure \ref{convergence_plot} gives the standardized root mean square errors $ \| \widehat{\cm{A}}_{K}-\cm{A} \|_{\rm F}/ \| \cm{A} \|_{\rm F}$ for the first 20 iterations, and it can be seen that all cases have a similar decay pattern.
In the meanwhile, lower estimation errors can be implied by more accurate pre-specified ranks and/or larger numbers of intraday observations.
We have also tried more replications and even different data generating processes, and a similar phenomenon can be observed.

\section{Real data analysis}
\label{Sect6}

This section analyzes the high-frequency trading data for the constituent stocks of S\&P 500 Index from April 1, 2009 to December 30, 2013, and the data from the first quarter of 2009 are dropped in order to reduce the effect of structural breaks.
Specifically, we consider $N=31$, 42 or 90 stocks with the largest trading volumes on January 2, 2013, and the data in 2013 are used to evaluate the out-of-sample performance. As a result, there are $T=937$ days for estimation and $M=249$ days for prediction.

The daily trading data from 9:30 am to 4:00 pm are downloaded from the Wharton Research Data Services, and we use the most commonly used five-minute returns, corresponding to the number of intraday observations $m=78$, in the literature \citep{Zhang:2011,Andersen:1998,LPS:2015}. 
Overnight returns are excluded since they usually have the jumps influenced by external factors.
Two realized measures, the RV and median RV (medRV), are used to estimate the integrated volatility, where the medRV can reduce the jump effects; see \cite{Anderson:2012}.
Both measures are first transformed into a logarithmic form and then centralized with mean zero.

The proposed MLR-HAR model with order $P=22$ is compared to two competitors: the vector HAR (VHAR) and vector HAR-index (VHARI) models.
The volatility components in both competitors are fixed to be the daily, weekly and monthly volatilities, while those in our model are chosen automatically.
We also consider the order $P=66$ for the MLR-HAR model to explore the possible longer-term volatility components.
No further dimension reduction is involved in the VHAR model, and a low rank of $r<N$ is assumed to the row space of coefficient matrices in the VHARI model. 
As a result, there are $3N^2$, $4Nr-r^2$ and $r_1r_2r_3+(N-r_1)r_1+(N-r_2)r_2+(P-r_3)r_3$ parameters in the VHAR, VHARI and MLR-HAR models, respectively, and our model has much less parameters.
The BIC at \eqref{eq:BIC} is used to search for the Tucker ranks of MLR-HAR models with tuning parameter $\lambda=10^{-4}$ and $1\leq r_1,r_2,r_3\leq 10$, leading to $(r_1,r_2,r_3)=(3,3,2)$ for both orders $P=22$ and 66 and all cases with $N=31$, 42 and 90 stocks by using the training data from 2009 to 2012.
It is also modified to select the rank of VHARI models, and we have $r=4$, 4 and 5 for $N=31$, 42 and 90 stocks, respectively.


A rolling forecast procedure is employed to evaluate the out-of-sample performance of the four models: the ending point of historical data iterates in the out-of-sample period of 2013 with the window size being fixed at $T=937$, and then one-step ahead prediction is conducted for each iteration.
The ranks for MLR-HAR and VHARI models are fixed as in the above during the prediction, and we adopt the most commonly used empirical quasi-likelihood (QLIKE) in the literature to evaluate the forecasting accuracy, 
\begin{equation*}
	\text{QLIKE}_i=\frac{1}{M}\sum_{n=1}^{M}\left(\frac{{\widetilde {y}}_{i,n}}{\widehat{y}_{i,n}}-{\rm log}\left( \frac{\widetilde{y}_{i,n}}{\widehat{y}_{i,n}}\right)-1\right) \hspace{2mm}\text{with}\hspace{2mm} 1\leq i\leq N,
\end{equation*}
where $\widehat{y}_{i,n}$ and $\widetilde {y}_{i,n}$ are the predicted and calculated realized measures for the $i$-th asset at the $n$-th trading day in 2013, respectively, and $M=249$ is the number of trading days in 2013; see \cite{Patton:2009,Patton:2011b,Patton:2011a}.
Figure \ref{figure:prediction_performance} gives the boxplots of QLIKEs from VHAR, VHARI and two MLR-HAR models with $N=31$, 42 or 90 stocks, and we have four findings below. 
(i.) The MLR-HAR models have much better forecasting accuracy than the VHAR and VHARI models, indicating the importance of exploring the low-rank structures among $N$ stocks from two directions.
(ii.) The two MLR-HAR models have a similar performance, and the model with $P=66$ even has a slightly worse performance.
This confirms the common practice in the literature to use volatility components within one month for forecasting realized measures, and the longer-term volatility component may have no significant contribution.
(iii.) The prediction error becomes larger generally when there are more stocks, and it may be due to the fact that a larger $N$ implies a more complicated model.
(iv.) Finally, the QLIKEs for medRV are much smaller than those for RV, and jumps may exist in the data.

To better understand the impact of high dimensionality on MLR-HAR models, we consider the training data with a shorter period from January 2, 2011 to December 30, 2012, and there are only $T=497$ days in total. 
The selected ranks of VHARI models are $r=3$, 4 and 5 for $N=31$, 42 and 90 stocks, respectively, while the Tucker ranks of MLR-HAR models are the same as before.
The corresponding boxplots of QLIKEs are presented in Figure \ref{figure:prediction_performance}, and it can be seen that all prediction errors become larger.
Especially, the VHAR model with $N=90$ stocks has the apparently worst performance, and this may be due to the less historical data but massive number of parameters.
As a comparison, both MLR-HAR models are not influenced too much, and we may argue that our model has the capability of handling many assets simultaneously. 

The cascade structure with three heterogeneous volatility components has been widely used in the literature, while the MLR-HAR modeling in the above insists on two data-driven components only for all cases. 
This motivates us to further study whether the three commonly used components with fixed forms in the literature are optimal in forecasting volatility.
To this end, we first refit the MLR-HAR model with order $P=22$ to the data from April 1, 2009, to December 30, 2013, and the Tucker ranks are set to $(r_1,r_2,r_3) = (3,3,3)$, i.e. the number of heterogeneous volatility components is fixed to three such that we can make a comparison between the estimated loading matrix $\widehat{\mathbf{U}}_3\in\mathbb{R}^{22\times 3}$ and $\mathbf{U}_{\mathrm{C}}\in\mathbb{R}^{22\times 3}$ at \eqref{har-rv}.
Note that, from Tucker decomposition at \eqref{tucker-decomposition}, $\widehat{\mathbf{U}}_3$ is not unique, while its column space, $\textrm{colspace}(\widehat{\mathbf{U}}_3)$, can be uniquely determined.
We next consider the discrepancy measure in \cite{pan2008modelling} to evaluate the distance between $\widehat{\mathbf{U}}_3$ and $\mathbf{U}_{\mathrm{C}}$, $\mathcal{D}\{\textrm{colspace}(\widehat{\mathbf{U}}_3),\textrm{colspace}(\mathbf{U}_{\mathrm{C}})\} = \{1-\rm{tr}(\mathbf{D}_1\mathbf{D}_1^\top\mathbf{D}_2\mathbf{D}_2^\top)/3\}^{1/2}$, where $\mathbf{D}_1$ and $\mathbf{D}_2$ are the orthonormal bases of $\textrm{colspace}(\widehat{\mathbf{U}}_3)$ and $\textrm{colspace}(\mathbf{U}_{\mathrm{C}})$, respectively, and it takes values within the range of $[0,1]$ with a larger value corresponding to more discrepancy between two spaces.
The metric has the values of $0.717$, $0.736$ and $0.652$ for the cases with $N=31$, $42$ and $90$ stocks, respectively, i.e. there exists a significant discrepancy between $\widehat{\mathbf{U}}_3$ and $\mathbf{U}_{\mathrm{C}}$.
This confirms the necessity of the proposed data-driven method to automatically select the heterogeneous volatility components in HAR models.

\section{Conclusion and discussion}
\label{Sect7}

This paper provides important extentions to the popular HAR model and solves its major drawbacks in both the high-frequency and low-frequency domains. Specifically, in the high-frequency regime, we establish the univariate and multivariate HAR-It\^o models to bridge the HAR-type models and their corresponding It\^o diffusion processes, which is the first to explore the high-frequency dynamics for HAR-type models. On the other hand, the multilinear low rank HAR model is proposed in the low-frequency regime, that exerts low rank assumptions on response, predictor and lag directions simultaneously. This low rank structure not only reduce the parameter space dramatically, enabling the model to handle the case with much more assets, but also replace the fixed heterogeneous volatility components in HAR model with a data-driven one. As a result, the flexibility is greatly enhanced, and the necessity is supported by the real data analysis. Finally, the theoretical properties of the high-dimensional HAR modeling are derived and projected gradient descent algorithm is suggested with theoretically justified linear convergence.

We next briefly discuss some possible extentions from this work. Firstly, the idea of HAR models \citep{Corsi:2009} stems from grouping different time horizons into different types of volatility components.
Along this line, we may rearrange $\widetilde{\mathbf{y}}_{n-1},\ldots, \widetilde{\mathbf{y}}_{n-P}$ at \eqref{autoregressive} into a third-order tensor such that the three modes correspond to assets, weeks and months, and the tensor technique can then be used to automatically select the weekly (medium-term) and then monthly (long-term) volatility components. 
This will lead to a fourth-order coefficient tensor $\cm{A}$, and the high-dimensional modeling tool in this paper may be adapted for it.
Secondly, the non-asymptotic properties derived in Section 4 depend on the sub-Gaussian assumption, while financial and economic data are usually heavy-tailed \citep{Shin2021}. It hence is of practical importance to discuss theoretical properties of the proposed high-dimensional modeling under a more heavy-tailed assumption. 
Finally, the realized covariance \citep{BauerandVorkink:2011} has recently attracted more and more attention, and they can be treated as matrix-valued time series  \citep{chen2021}. 
As in this paper, we may also bridge the HAR-covariance representation and It\^o diffusion processes, and then propose a low-rank HAR-covariance model to forecast realized covariances.

\bibliography{myReferences}

\newpage
\begin{figure}[H]
	\centerline{\includegraphics[width=1\textwidth]{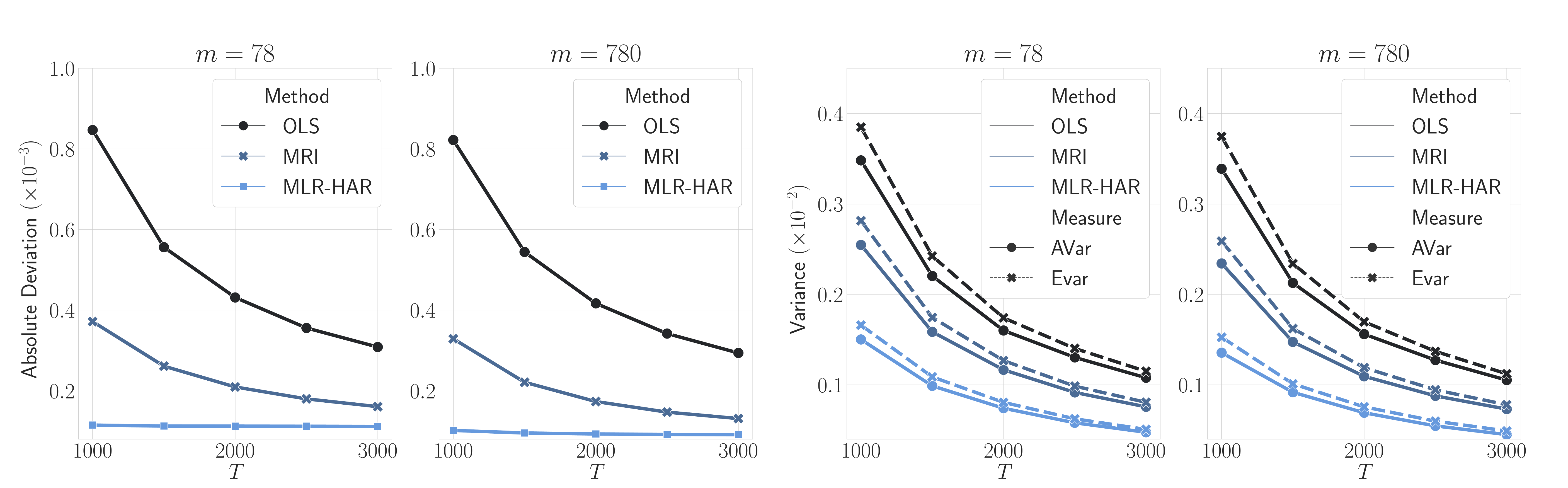}}
	\caption{Squared absolute deviations (two left panels), and empirical variance (EVar) and estimated asymptotic variance (AVar) matrices (two right panels) for the estimators $\widehat{\cm{A}}_{\rm MLR}$, $\widehat{\cm{A}}_{\rm MRI}$ and $\widehat{\cm{A}}_{\rm OLS}$ with the number of intraday observations being $m=78$ or $780$.}
	\label{thm1-plot}
\end{figure}

\begin{figure}[H]
	\centerline{\includegraphics[width=1\textwidth]{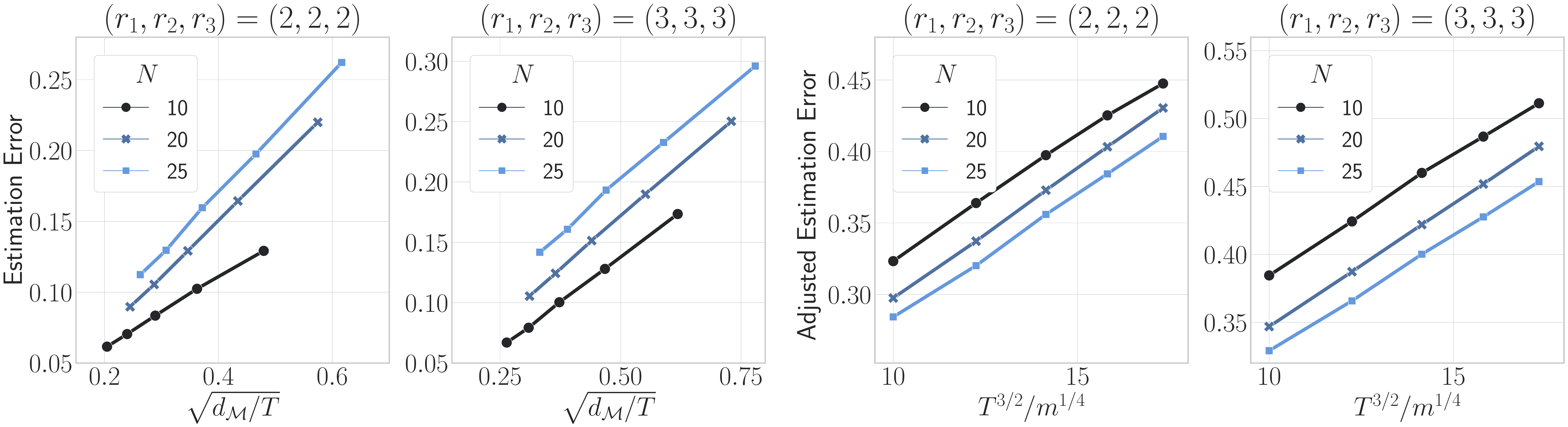}}
	\caption{Estimation errors $\| \widehat{\mathcal{A}}_{\rm MLR}-\mathcal{A} \|_{\rm F}$ for the first data generating process with respect to $\sqrt{d_{\mathcal{M}}/T}$ (two left panels), and adjusted estimation errors $\sqrt{T/d_{\mathcal{M}}}\| \widehat{\cm{A}}_{\rm MLR}-\cm{A} \|_{\rm F}$ from the second data generating process with respect to $T^{3/2}/m^{1/4}$ (two right panels).}
	\label{thm2-plot}
\end{figure}

\begin{figure}[H]
	\centerline{\includegraphics[width=1.0\textwidth]{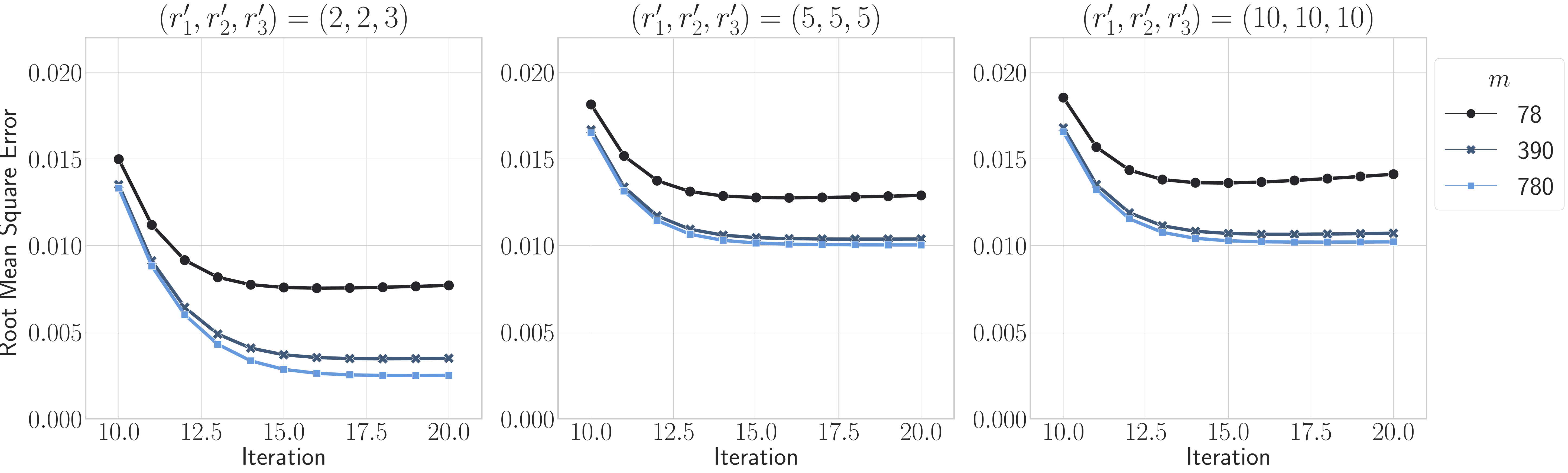}}
	\caption{Standardized root mean square errors $ \| \widehat{\cm{A}}_{\rm MLR}-{\cm{A}} \|_{\rm F}/\|{\cm{A}} \|_{\rm F}$ for the first 20 iterations with running ranks $(r_1^{\prime}, r_2^{\prime},  r_3^{\prime}) = (2, 2, 3)$, $ (5, 5, 5)$ and $ (10, 10, 10)$.}
	\label{convergence_plot}
\end{figure}

\begin{figure}[H]
	\centerline{\includegraphics[width=1.25\textwidth]{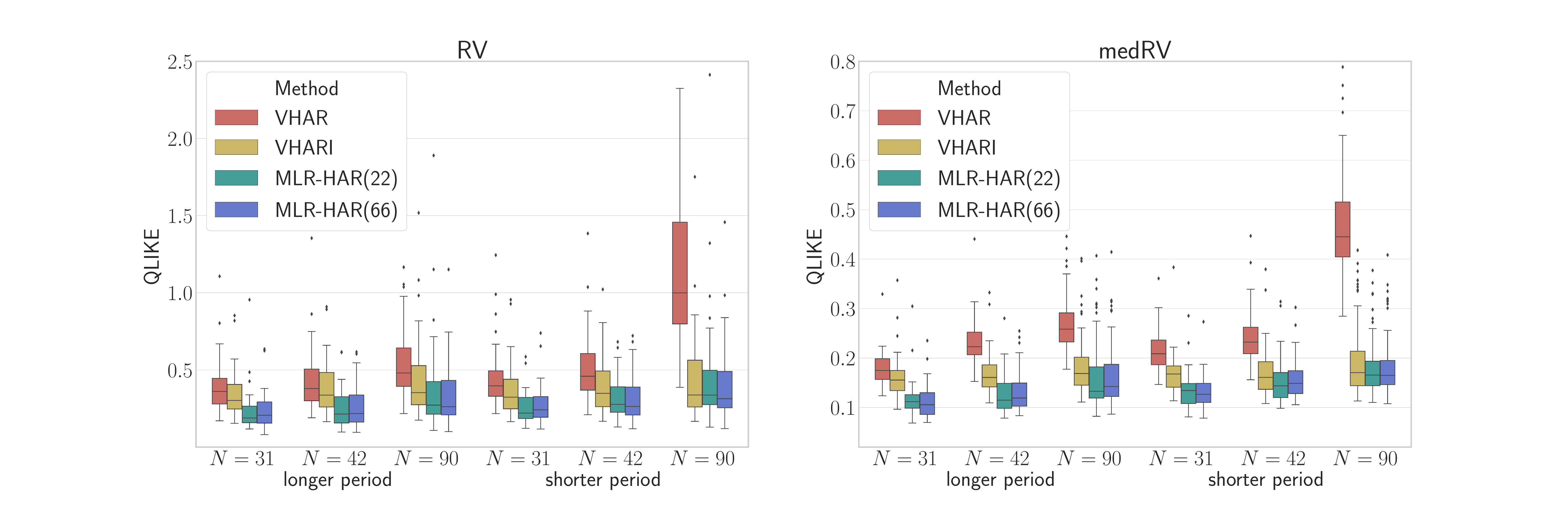}}
	\caption{Boxplots of QLIKEs from the VHAR model, VHARI model and MLR-HAR model with orders $P=22$ and 66 for $N=31$, 42 and 90 stocks with a longer period (2009.4--2013.12) and a shorter period (2011.1--2013.12) of data. Two realized measures are used: RV (left panel) and medRV (right panel).}
	\label{figure:prediction_performance}
\end{figure}

\newpage
\renewcommand{\thesection}{A}
\setcounter{subsection}{0} 
\renewcommand{\theequation}{A.\arabic{equation}}
\setcounter{equation}{0} 

\section*{Appendix A: technical details}
This appendix gives the technical proofs of theoretical results in Sections \ref{Sec2}-\ref{Sect4}, respectively.
\subsection{Proofs of Propositions \ref{prop:1} and \ref{prop:2} in Section \ref{Sec2}}
\begin{proof}[Proof of Proposition \ref{prop:1}]
	We define
	\begin{equation*}
		R(k)= \displaystyle \int_{n-1}^{n}\frac{(n-t)^{k}}{k!}\sigma_{t}^{2}dt =\frac{\sigma_{n-1}^{2}}{(k+1)!}+\displaystyle \int_{n-1}^{n}\frac{(n-t)^{k+1}}{(k+1)!}d\sigma_{t}^{2} 
	\end{equation*}
    for all $k\geq 0$ and $n\geq P+1$.
    From Definition \ref{def:har-ito} and It\^o lemma, it holds that
    \[
    d\sigma_{t}^{2}=[\omega+v(n-t)-\sigma_{n-1}^2- vZ_t^2+\alpha_1 \sigma_t^2]dt+2 \upsilon(n-t)Z_td Z_t+\beta L_t^2 d\Lambda_t
    \]
    for $n-1<t\leq n$, and
    \begin{equation*}
    		v \displaystyle\int_{n-1}^{n}\frac{(n-t)^{k+1}}{(k+1)!}Z_{t}^2dt
    		=\frac{v}{(k+3)!}+2v\displaystyle\int_{n-1}^{n}\frac{(n-t)^{k+2}}{(k+2)!}Z_{t}dZ_{t}.
    \end{equation*}
    As a result, for any $k\geq 0$,
    \begin{align}
    	\begin{split}\label{eq:iteration}
    	R(k) &=\frac{w}{(k+2)!}+\frac{(k+2)v}{(k+3)!}
    	+\frac{(k+1) \sigma_{n-1}^{2}}{(k+2)!} +2 v\displaystyle \int_{n-1}^{n}\frac{(n-t)^{k+2}}{(k+1)!}Z_{t}dZ_{t} \\ 
    	&\hspace{5mm}
    	- v\displaystyle\int_{n-1}^{n}\frac{(n-t)^{k+1}}{(k+1)!}Z_{t}^2dt +\beta \displaystyle \int_{n-1}^{n} \frac{(n-t)^{k+1}}{(k+1) !} L_{t}^2 d \Lambda_{t}+\alpha_1 R(k+1)\\
    	&:=R_1(k) +R_2(k) +R_3(k) +R_4(k)+\alpha_1 R(k+1),
    	\end{split}
    \end{align}
    where
    \begin{align*}
    	&R_1(k)=2 v
    \displaystyle\int_{n-1}^{n}\frac{(k+1)(n-t)^{k+2}}{(k+2)!}Z_{t}dZ_{t}, \hspace{3mm}R_2(k)=\beta \displaystyle\int_{n-1}^{n} \frac{(n-t)^{k+1}}{(k+1) !} L_{t}^2 d \Lambda_{t},\\
    	&R_3(k) =\frac{(k+3)w+(k+1)v}{(k+3)!}, \hspace{8mm}\text{and}\hspace{8mm} 
    	R_4(k)=\frac{ (k+1)\sigma_{n-1}^{2}}{(k+2)!}.
    \end{align*}

Since $R(k)\leq  \int_{n-1}^{n}\sigma_{t}^{2}dt$ and $0< \alpha_1<1$, it holds that $\alpha_1^{k}R(k)\rightarrow 0$ almost surely as $k\rightarrow\infty$ and, by iterating the formula at \eqref{eq:iteration}, we have
\begin{equation}\label{eq:r0}
	R(0) =\displaystyle \int_{n-1}^{n}\sigma_{t}^{2}dt =\sum_{k=0}^{\infty}\alpha_{1}^kR_1(k) +\sum_{k=0}^{\infty}\alpha_{1}^kR_2(k) +\sum_{k=0}^{\infty}\alpha_{1}^kR_3(k)+\sum_{k=0}^{\infty}\alpha_{1}^kR_4(k)
\end{equation}
with probability one.
By some algebra, it can be verified that
\begin{align}
	\begin{split}\label{eq:r0-1}
	\sum_{k=0}^{\infty}\alpha_{1}^kR_1(k) &=\sum_{k=0}^{\infty}2v\alpha_1^{-1}\displaystyle \int_{n-1}^{n}\frac{[\alpha_1(n-t)]^{k+1}(n-t)}{(k+1)!}Z_{t}dZ_{t}\\
	&\hspace{5mm}-\sum_{k=0}^{\infty}2v\alpha_1^{-2}\displaystyle \int_{n-1}^{n}\frac{[\alpha_1(n-t)]^{k+2}}{(k+2)!}Z_{t}dZ_{t}\\
	&=2 v \alpha_1^{-2} \displaystyle \int_{n-1}^{n}\left\{\alpha_1\left(n-t-\alpha_1^{-1}\right) e^{\alpha_1(n-t)}+1\right\} Z_{t} d Z_{t}:=\varepsilon_n^c,\\
	\sum_{k=0}^{\infty}\alpha_{1}^kR_2(k)&= \sum_{k=0}^{\infty}\beta\alpha_1^{-1}\int_{n-1}^{n}\frac{[\alpha_1(n-t)]^{k+1}}{(k+1)!}L_{t}^2 d \Lambda_{t} =\beta \alpha_{1}^{-1}\displaystyle \int_{n-1}^{n}\left(e^{\alpha_{1}(n-t)}-1\right) L_{t}^2 d \Lambda_{t}\\
	&=\varrho_{2}\beta\omega_L\lambda+\varepsilon_n^J,\\ 
	\sum_{k=0}^{\infty}\alpha_{1}^kR_3(k)&=\varrho_{2}\omega+(\varrho_{2}-2\varrho_{3})v \hspace{3mm}\text{and}\hspace{3mm}
	\sum_{k=0}^{\infty}\alpha_{1}^kR_4(k)=(\varrho_{1}-\varrho_{2})\sigma_{n-1}^{2},
	\end{split}
\end{align}
where $\varrho_{1}=\alpha_1^{-1}(e^{\alpha_1}-1)$, $\varrho_{2}=\alpha_1^{-2}(e^{\alpha_1}-1-\alpha_1)$, $\varrho_{3}=\alpha_1^{-3}(e^{\alpha_1}-1-\alpha_1-{\alpha_1^{2}}/{2})$, and
\[
\varepsilon_n^J=\beta \alpha_{1}^{-1}\left\{\displaystyle\int_{n-1}^{n}\left(e^{\alpha_{1}(n-t)}-1\right) [M_{t} d \Lambda_{t}+\omega_{L} \left(d \Lambda_{t}-\lambda d t\right)] \right\}.
\]
Moreover, the instantaneous volatility $\sigma_{t}^{2}$ at integer time point $n$ has the form of
$
	\sigma_{n}^{2}=\omega
	+\sum_{l=1}^{P}\alpha_l y_{n-l+1}+\beta J_n,
$
where
\begin{equation*}
	y_n=\displaystyle \int_{n-1}^{n}\sigma_{s}^{2}ds \hspace{5mm}\text{and}\hspace{5mm}
	J_n=\displaystyle \int_{n-1}^{n} L_{s}^{2} d \Lambda_{s}
\end{equation*}
are the integrated volatility and the jump variation, respectively. 
This, together with \eqref{eq:r0} and \eqref{eq:r0-1}, implies the equation at \eqref{har-model}.

Note that random variables $\varepsilon_{n}^{c}$ and $\varepsilon_{n}^{J}$ are both independent of the $\sigma$-field $\mathcal{F}_{n-1}$, and hence $\{\varepsilon_{n}^{c}\}$ and $\{\varepsilon_{n}^{J}\}$ are two $i.i.d.$ sequences. Moreover, it is readily to verify that both $\{\varepsilon_{n}^{c}\}$ and $\{\varepsilon_{n}^{J}\}$ are martingale difference sequences \citep{song2020realized}, i.e. $E(\varepsilon_{n}^{c})=E(\varepsilon_{n}^{J})=0$. We next show that both var$(\varepsilon_{n}^{c})$ and var$(\varepsilon_{n}^{J})$ are finite.

Let $\varrho_{4}(t)=\alpha_1\left(n-t-\alpha_1^{-1}\right) e^{\alpha_1(n-t)}+1$, and then $|\varrho_{4}(t)|<2e+1$ for all $n-1<t\leq n$. 
Moreover, $Z_{t}=\int_{[t]}^{t}dW_{t}=W_t-W_{n-1}$, $d Z_t=d W_t$ and $(d Z_t)^2=d t$. As a result, by It\^o isometry,
 \begin{align}
	\begin{split}\label{eq:error-c}
		E({\varepsilon_n^{c}})^2=&E\left(2 v \alpha_1^{-2} \displaystyle \int_{n-1}^{n}\varrho_{4}(t) Z_{t} d Z_{t}\right)^2
		=4 v^2 \alpha_1^{-4}E\displaystyle \int_{n-1}^{n}[\varrho_{4}(t)]^2Z_{t}^2(d Z_{t})^2\\
		< &4 v^2 \alpha_1^{-4}(2e+1)^2E\displaystyle \int_{n-1}^{n}Z_{t}^2dt.
	\end{split}
\end{align}
Due to It\^o lemma and definition of $Z_t$, it holds that $dZ_t^2=2Z_tdZ_t+dt$, $Z_{n-1}=0$, and $\{\int_{n-1}^{t}Z_sdZ_s\}$ is a martingale difference sequence. Thus, $Z_t^2=Z_t^2-Z_{n-1}^2=\int_{n-1}^{t}(2Z_sdZ_s+ds)$, and
\begin{equation*}
	E\int_{n-1}^{n}Z_{t}^2dt=2 E\left(\displaystyle \int_{n-1}^{n}\displaystyle \int_{n-1}^{t}Z_sdZ_sdt\right) +0.5 =0.5,
\end{equation*}
which, together with \eqref{eq:error-c}, implies that var$(\varepsilon_{n}^{c}) < 2v^2 \alpha_1^{-4}(2e+1)^2<\infty$.

We next consider var$(\varepsilon_n^{J})$. Let $\varrho_5(t)=e^{\alpha_{1}(n-t)}-1$, and then $|\varrho_5(t)|<e+1$ for all $n-1<t\leq n$.
From \eqref{eq:r0-1} and the fact that $E(\varepsilon_n^{J})=0$, it holds that
\begin{equation}\label{eq:error-j}
	\mathrm{var}(\varepsilon_n^{J})=\beta^2 \alpha_{1}^{-2}E\left(\displaystyle \int_{n-1}^{n}\varrho_5(t)L_t^2 d \Lambda_{t}\right)^2 -(\varrho_{2}\beta\omega_L\lambda)^2.
\end{equation}
Note that $\{L_t\}$ are independent of the Poisson process $\{\Lambda_{t}\}$ and, by the partitioning method, we have
\begin{align*}
	\begin{split}
		E\left(\int_{n-1}^{n}\varrho_5(t)L_t^2 d \Lambda_{t}\right)^2 &=E\left(\sum_{n-1<t\leq n}\varrho_5(t)L_{t}^2 \Delta\Lambda_{t}\right)^2
	<  (e+1)^2 E\left(\sum_{n-1<t\leq n}L_{t}^2 \Delta\Lambda_{t}\right)^2\\
		&	=(e+1)^2 E(L_t^4) E\left( \sum_{n-1<t\leq n}\Delta\Lambda_{t}^2\right)=	(e+1)^2 (\lambda^2+\lambda)E(L_t^4),
	\end{split}
\end{align*}
which, together with \eqref{eq:error-j} and the fact that $E(L_t^4)<\infty$, implies that var$(\varepsilon_n^{J})<\infty$.

\end{proof}

\begin{proof}[Proof of Proposition \ref{prop:2}]
	We define
	\begin{align*}
		R_i(k)= & \int_{n-1}^{n}\frac{(n-t)^{k}}{k!}\sigma_{i,t}^{2}dt=\frac{\sigma_{i,n-1}^{2}}{(k+1)!}+\displaystyle \int_{n-1}^{n}\frac{(n-t)^{k+1}}{(k+1)!}d\sigma_{i,t}^{2},\hspace{3mm}1\leq i\leq N,
	\end{align*}
	for all $k\geq 0$ and $n\geq P+1$.
	From Definition \ref{def:mhar-ito} and It\^o lemma, we have
	\begin{align*}
		d\sigma_{i,t}^{2}=&\left(\omega_{i}+v_i(n-t)-\sigma_{i,n-1}^2- \upsilon_i Z_{i,t}^2+\sum_{j=1}^{N}\alpha_{i,j}^{(1)} \sigma_{j,t}^2\right)dt
		+2 \upsilon_i(n-t)Z_{i,t}d Z_{i,t}\\
		&+\beta_i L_{i,t}^2 d\Lambda_{i,t}
	\end{align*}
	for $n-1<t\leq n$, and 
	\begin{equation*}
		\upsilon_i\int_{n-1}^{n}\frac{(n-t)^{k+1}}{(k+1)!}Z_{i,t}^2dt
		=\frac{v_i}{(k+3)!}+2\upsilon_i\int_{n-1}^{n}\frac{(n-t)^{k+2}}{(k+2)!}Z_{i,t}dZ_{i,t}.
	\end{equation*}
	As a result, for any $k\geq 0$, let
	\begin{align*}
		&R_i^{(1)}(k)=2 v_i
		\displaystyle \int_{n-1}^{n}\frac{(k+1)(n-t)^{k+2}}{(k+2)!}Z_{i,t}dZ_{i,t}, \hspace{3mm}	R_i^{(2)}(k)=\beta_i \displaystyle \int_{n-1}^{n} \frac{(n-t)^{k+1}}{(k+1) !}L_{i,t}^2 d \Lambda_{i,t}\\
		&R_i^{(3)}(k)=\frac{(k+3)w_{i}+(k+1)v_i}{(k+3)!}, \hspace{10mm}\text{and}\hspace{8mm}	R_i^{(4)}(k)=\frac{ (k+1)\sigma_{i,n-1}^{2}}{(k+2)!},
	\end{align*}
	and then it holds that
	\begin{align*}
		R_i(k)=&\frac{\omega_{i}}{(k+2)!}+\frac{(k+2)v_i}{(k+3)!}+\frac{(k+1)\sigma_{i,n-1}^{2}}{(k+2)!}\\
		&+2 v_i \displaystyle \int_{n-1}^{n}\frac{(n-t)^{k+2}}{(k+1)!}Z_{i,t}dZ_{i,t}- v_i
		\displaystyle \int_{n-1}^{n}\frac{(n-t)^{k+1}}{(k+1)!}Z_{i,t}^2dt\\
		&+\beta_i \displaystyle \int_{n-1}^{n} \frac{(n-t)^{k+1}}{(k+1) !} L_{i,t}^2 d \Lambda_{i,t}+\sum_{j=1}^{N}\alpha_{i,j}^{(1)} R_j(k+1)\\
		=&R_i^{(1)}(k)+R_i^{(2)}(k)+R_i^{(3)}(k)+R_i^{(4)}(k)+\sum_{j=1}^{N}\alpha_{i,j}^{(1)} R_j(k+1),
	\end{align*}
	which can be rewritten into
	\begin{align}
		\label{eq:iteration_vector}
		\mathbf{R}(k)=\mathbf{R}^{(1)}(k)+\mathbf{R}^{(2)}(k)+\mathbf{R}^{(3)}(k)+\mathbf{R}^{(4)}(k)+\bm{\alpha}^{(1)} \mathbf{R}(k+1),
	\end{align}
	where  $\mathbf{R}(k)=(R_1(k), R_2(k), \cdots, R_N(k))^{\top}$ and $\mathbf{R}^{(j)}(k)=(R_1^{(j)}(k), R_2^{(j)}(k), \cdots, R_N^{(j)}(k))^{\top}$ with $1\leq j\leq 4$.
	
	Due to the facts that $R_i(k)\leq  \int_{n-1}^{n}\sigma_{i,t}^{2}dt$ and the spectral radius of $\bm{\alpha}^{(1)}$ is less than one,  
	it holds that ${{\bm{\alpha}}^{(1)}}^{k}\mathbf{R}(k)\rightarrow 0$ almost surely as $k\rightarrow\infty$. Moreover, by iterating the formula at \eqref{eq:iteration_vector}, we have
	\begin{align}
		\begin{split}\label{eq:r0-vector}
			\mathbf R(0) =\displaystyle \int_{n-1}^{n}{\bm{\sigma}_{t}^{2}}dt =\sum_{k=0}^{\infty}{{\bm{\alpha}}^{(1)}}^k\mathbf {R}^{(1)}(k) &+\sum_{k=0}^{\infty}{{\bm{\alpha}}^{(1)}}^k\mathbf{R}^{(2)}(k) +\sum_{k=0}^{\infty}{{\bm{\alpha}}^{(1)}}^k\mathbf{R}^{(3)}(k)\\ &+\sum_{k=0}^{\infty}{{\bm {\alpha}}^{(1)}}^k\mathbf{R}^{(4)}(k)
		\end{split}
	\end{align}
	with probability one. Similar to Proposition \ref{prop:1}, it can be verified that
	\begin{align}
		\begin{split}\label{eq:vr0-1}
			\sum_{k=0}^{\infty}{{\bm{\alpha}}^{(1)}}^k \mathbf{R}_1(k) &=\sum_{k=0}^{\infty}{{\bm {\alpha}}^{(1)}}^k\left[2v_i \displaystyle \int_{n-1}^{n}\frac{(k+1)(n-t)^{k+2}}{(k+2)!}Z_{i,t}dZ_{i,t}\right]^{\top}_{i=1,\cdots,N}\\
			&=2 \mathbf{V}{{ \bm{\alpha}}^{(1)}}^{-2} \displaystyle \int_{n-1}^{n}\left\{\left((n-t){\bm{\alpha}}^{(1)}-\mathbf {I}_N\right)\mathbf{e}^{(n-t){ \bm{\alpha}}^{(1)}}+\mathbf{I}_N\right\} \mathbf{Z}_{t} d \mathbf{Z}_{t}\\
			&:=\bm{\varepsilon}_n^c,\\
			\sum_{k=0}^{\infty}{{\bm {\alpha}}^{(1)}}^k \mathbf{R}_2(k)&= \sum_{k=0}^{\infty}{{\bm{\alpha}}^{(1)}}^k \displaystyle \int_{n-1}^{n} \left[\beta_i\frac{(n-t)^{k+1}}{(k+1) !}L_{i,t}^2 d \Lambda_{i,t}\right]^{\top}_{i=1,\cdots,N}\\
			&=\bm{\varrho}_{2}\bm{\beta} \bm{\omega}_L \bm{\lambda}+\bm{\varepsilon}_n^J,\\ 
			\sum_{k=0}^{\infty}{{\bm{\alpha}}^{(1)}}^k \mathbf{R}_3(k)&=\bm {\varrho}_{2}\bm{\omega}+(\bm{\varrho}_{2}-2\bm{\varrho}_{3})\mathbf{V} \mathbf{1}_N \hspace{3mm}\text{and}\\
			\sum_{k=0}^{\infty}{{\boldsymbol\alpha}^{(1)}}^k \mathbf{R}_4(k)&=(\bm {\varrho}_{1}-\bm{\varrho}_{2})\bm{\sigma}_{n-1}^{2},
		\end{split}
	\end{align}
	where 
	\begin{align*}
		\bm{\varepsilon}_{n}^{J}=\bm \beta {{ \bm\alpha}^{(1)}}^{-1}\left\{\displaystyle \int_{n-1}^{n}\left(\mathbf{e}^{{{ \bm {\alpha}}^{(1)}}(n-t)}-\mathbf{I}_N\right) [ \mathbf{M}_{t} d \bm{\Lambda}_{t}+\bm{\omega}_{L}(d \bm{\Lambda}_{t}-\bm{\lambda} d t)]\right\}.
	\end{align*}
	Note that $\bm{\sigma}_{t}^{2}$ at integer time point $n$ has the form of
	\begin{equation*}
		\bm{\sigma}_{n}^{2}=\bm{\omega}
		+\sum\limits_{l=1}^{P}\bm{\alpha}^{(l)} \mathbf{y}_{n-l+1}+	\bm{\beta} \mathbf{J}_n,
	\end{equation*}
	and we then have the autoregressive form at \eqref{vector_integrate}.
	By a method similar to Proposition \ref{prop:1}, we further show that both $\{\bm{\varepsilon}_{n}^{c}\}$ and $\{\bm{\varepsilon}_{n}^{J}\}$ are $i.i.d.$ sequences with mean zero and finite variance matrices.
\end{proof}

\subsection{Proofs of Theorem \ref{theorem 1} and Corollary \ref{corollary 1} in Section \ref{Sec3}}

\begin{proof}[Proof of Theorem \ref{theorem 1}]
	The proof of this theorem consists of two steps: we first derive the asymptotic normality of the OLS estimator without low-rank constraint, and then adapt the proving technique in \cite{shapiro:1986} to establish the asymptotic normality of $\widehat{\cm{A}}_{\rm{MLR}}$.
	
	The first step considers the OLS estimator of model \eqref{MHAR_1} without any low-rank constraint, $\widehat{\cm{A}}_{\rm{OLS}}= \mathop{\argmin} T^{-1}\sum_{n=P+1}^{T}|| \widetilde{\mathbf{y}}_n- \cm{A}_{(1)}\widetilde{\mathbf{x}}_n||_2^2$, and it can be verified that
	\begin{equation}\label{proof-thm1-eq1}
	(\widehat{\cm{A}}_{\rm{OLS}})_{(1)}=\sum_{n=P+1}^T\widetilde{\mathbf{y}}_n\widetilde{\mathbf{x}}_n^{\top}\left(\sum_{n=P+1}^T\widetilde{\mathbf{x}}_n\widetilde{\mathbf{x}}_n^{\top}\right)^{-1} =\cm{A}_{(1)}+\sum_{n=P+1}^T\bm{\epsilon}_{n}\widetilde{\mathbf{x}}_n^{\top}\left(\sum_{n=P+1}^T\widetilde{\mathbf{x}}_n\widetilde{\mathbf{x}}_n^{\top}\right)^{-1},
	\end{equation} 
	where $\widetilde{\mathbf{x}}_{n}=(\widetilde{\mathbf{y}}_{n-1}^{\top},\ldots,\widetilde{\mathbf{y}}_{n-P}^{\top})^\top$, $\bm{\epsilon}_{n}=\bm{\eta}_n-\sum_{j=1}^{P}\mathbf{A}_j\bm {\eta}_{n-j}+\bm {\varepsilon}_n$, and $\bm{\eta}_n=\widetilde{\mathbf{y}}_n-{\mathbf{y}}_n$.
	Moreover, from \eqref{condition_estimation error}, we have
	\begin{equation}
	\mathbb{P}\left\{\max _{1 \leq n \leq T} \max _{1 \leq i \leq N}\left|{\eta}_{i,n}\right| \geq \frac{(NT)^{1/2+\delta/4}}{m^{1 / 4}}\right\} \leq \frac{C_1}{(NT)^{\delta/2}},
	\end{equation}
	where $C_1$ is a positive constant, and $\delta>0$ is given in the theorem.
	Let $\eta_{\mathrm{max}}=\max _{1 \leq n \leq T} \max _{1 \leq i \leq N}\left|{\eta}_{i,n}\right|$ for simplicity, and then $\eta_{\mathrm{max}}=O_p(T^{1/2+\delta/4}m^{-1 / 4})$. As a result, $T^{1/2}\eta_{\mathrm{max}} =O_p(T^{1+\delta/4}m^{-1 / 4}) =o_p(1)$ since $T^{4+\delta}m^{-1} \rightarrow  0$.
	
	We first handle the term of $T^{-1}\sum_{n=P+1}^T\widetilde{\mathbf{x}}_n\widetilde{\mathbf{x}}_n^{\top}$ at the right hand side of \eqref{proof-thm1-eq1}, and it holds that
	\begin{align}
	\begin{split}\label{proof-thm1-eq2}
	\frac{1}{T}\sum_{n=P+1}^T\widetilde{\mathbf{x}}_n\widetilde{\mathbf{x}}_n^{\top}=& \frac{1}{T}\sum_{n=P+1}^T{\mathbf{x}}_n{\mathbf{x}}_n^{\top} +\frac{1}{T}\sum_{n=P+1}^T(\widetilde{\mathbf{x}}_n-{\mathbf{x}}_n)(\widetilde{\mathbf{x}}_n-{\mathbf{x}}_n)^{\top}\\ &+\frac{1}{T}\sum_{n=P+1}^T{\mathbf{x}}_n(\widetilde{\mathbf{x}}_n-{\mathbf{x}}_n)^{\top} +\frac{1}{T}\sum_{n=P+1}^T(\widetilde{\mathbf{x}}_n-{\mathbf{x}}_n){\mathbf{x}}_n^{\top},
	\end{split}
	\end{align}
	where ${\mathbf{x}}_{n}=({\mathbf{y}}_{n-1}^{\top},\ldots,{\mathbf{y}}_{n-P}^{\top})^\top$.
	From Assumption \ref{assumption 1}, the fact that $\mathbb{E}\|\bm{\varepsilon}_{n}\|^4<\infty$, and the ergodic theorem, it holds that
	\begin{equation}\label{proof-thm1-eq3}
	\frac{1}{T}\sum_{n=P+1}^T{\mathbf{x}}_n{\mathbf{x}}_n^{\top} =\bm{\Gamma}^{*}+o_p(1) \hspace{5mm}\text{and}\hspace{5mm} 
	\frac{1}{T}\sum_{n=P+1}^T|{\mathbf{x}}_n| =\mathbb{E}(|{\mathbf{x}}_n|)+o_p(1),
	\end{equation}
	where $|\cdot|$ takes the absolute value in the element-wise sense.
	Moreover, it can be verified that $|\widetilde{\mathbf{x}}_{n}-{\mathbf{x}}_{n}|\leq \eta_{\mathrm{max}}\mathbf{1}_{NP}$, where $\mathbf{1}_{NP}$ is an $(NP)$-dimensional vector of ones. Thus,
	\begin{align*}
	|\frac{1}{T}\sum_{n=P+1}^T(\widetilde{\mathbf{x}}_n-{\mathbf{x}}_n)(\widetilde{\mathbf{x}}_n-{\mathbf{x}}_n)^{\top}| &\leq \eta_{\mathrm{max}}^2\mathbf{1}_{NP}\mathbf{1}_{NP}^\top \hspace{5mm}\text{and}\\
	|\frac{1}{T}\sum_{n=P+1}^T{\mathbf{x}}_n(\widetilde{\mathbf{x}}_n-{\mathbf{x}}_n)^{\top}|&\leq \eta_{\mathrm{max}}\left(\frac{1}{T}\sum_{n=P+1}^T|{\mathbf{x}}_n|\right)\bm{1}_{NP}^\top,
	\end{align*}
	which, together with \eqref{proof-thm1-eq2}, \eqref{proof-thm1-eq3} and the fact that $\eta_{\mathrm{max}}=o_p(1)$, implies that
	\begin{equation}\label{proof-thm1-eq4}
	\frac{1}{T}\sum_{n=P+1}^T\widetilde{\mathbf{x}}_n\widetilde{\mathbf{x}}_n^{\top}=\bm{\Gamma}^{*}+o_p(1).
	\end{equation}
	
	We next handle the term of $T^{-1/2}\sum_{n=P+1}^T\bm{\epsilon}_{n}\widetilde{\mathbf{x}}_n^{\top}$ at the right hand side of \eqref{proof-thm1-eq1}, and it holds that
	\begin{align}
	\begin{split}\label{proof-thm1-eq5}
	\frac{1}{\sqrt{T}}\sum_{n=P+1}^T\bm{\epsilon}_{n}\widetilde{\mathbf{x}}_n^{\top}=& \frac{1}{\sqrt{T}}\sum_{n=P+1}^T\bm{\varepsilon}_{n}{\mathbf{x}}_n^{\top} +\frac{1}{\sqrt{T}}\sum_{n=P+1}^T(\bm{\epsilon}_{n}-\bm{\varepsilon}_{n})(\widetilde{\mathbf{x}}_n-{\mathbf{x}}_n)^{\top}\\ &+\frac{1}{\sqrt{T}}\sum_{n=P+1}^T\bm{\varepsilon}_{n}(\widetilde{\mathbf{x}}_n-{\mathbf{x}}_n)^{\top} +\frac{1}{\sqrt{T}}\sum_{n=P+1}^T(\bm{\epsilon}_{n}-\bm{\varepsilon}_{n}){\mathbf{x}}_n^{\top}.
	\end{split}
	\end{align}
	From Assumption \ref{assumption 1}, there exists a positive constant $C$ such that $(\mathbf{I}_N+\sum_{j=1}^{P}|\mathbf{A}_j|)\mathbf{1}_N\leq C\cdot\mathbf{1}_N$ and then $|\bm{\epsilon}_{n}-\bm{\varepsilon}_{n}|=|\bm{\eta}_n-\sum_{j=1}^{P}\mathbf{A}_j\bm {\eta}_{n-j}|\leq C\eta_{\mathrm{max}}\mathbf{1}_N$.
	As a result,
	\begin{align*}
	\left|\frac{1}{\sqrt{T}}\sum_{n=P+1}^T(\bm{\epsilon}_{n}-\bm{\varepsilon}_{n})(\widetilde{\mathbf{x}}_n-{\mathbf{x}}_n)^{\top}\right|&\leq T^{1/2}\eta_{\mathrm{max}}^2\mathbf{1}_{NP}\mathbf{1}_{NP}^\top,\\
	\left|\frac{1}{\sqrt{T}}\sum_{n=P+1}^T\bm{\varepsilon}_{n}(\widetilde{\mathbf{x}}_n-{\mathbf{x}}_n)^{\top}\right|&\leq T^{1/2}\eta_{\mathrm{max}}\left(\frac{1}{T}\sum_{n=P+1}^T|\bm{\varepsilon}_{n}|\right)\bm{1}_{NP}^\top,\hspace{2mm}\text{and}\\
	\left|\frac{1}{\sqrt{T}}\sum_{n=P+1}^T(\bm{\epsilon}_{n}-\bm{\varepsilon}_{n}){\mathbf{x}}_n^{\top}\right|&\leq CT^{1/2}\eta_{\mathrm{max}}\bm{1}_{NP}\left(\frac{1}{T}\sum_{n=P+1}^T|{\mathbf{x}}_n^\top|\right).
	\end{align*}
	They, together with \eqref{proof-thm1-eq5}, \eqref{proof-thm1-eq3} and the facts that $T^{-1}\sum_{n=P+1}^T|\bm{\varepsilon}_{n}|=\mathbb{E}(|\bm{\varepsilon}_{n}|)+o_p(1)$ and $T^{1/2}\eta_{\mathrm{max}}=o_p(1)$, lead to
	\begin{equation}\label{proof-thm1-eq6}
	\frac{1}{\sqrt{T}}\sum_{n=P+1}^T\bm{\epsilon}_{n}\widetilde{\mathbf{x}}_n^{\top}= \frac{1}{\sqrt{T}}\sum_{n=P+1}^T\bm{\varepsilon}_{n}{\mathbf{x}}_n^{\top}+o_p(1).
	\end{equation}
	Combining \eqref{proof-thm1-eq1}, \eqref{proof-thm1-eq4}, \eqref{proof-thm1-eq6} and the central limit theorem for martingale difference sequences, we can obtain that
	\[
	\sqrt{T}\{{\rm{vec}}((\widehat{\cm{A}}_{\rm{OLS}})_{(1)})-\rm{vec}(\mathbf{\cm {A}}_{(1)})\}  \rightarrow  N(\mathbf{0}, {\bm{\Sigma}}_{\rm{OLS}})
	\]
	in distribution as $m\rightarrow\infty$ and $T\rightarrow\infty$, where ${\bm{\Sigma}}_{\rm{OLS}}=\bm{\Gamma}^{*-1}\otimes\bm{\Sigma}_{\varepsilon}$.
	
	The second step mainly follows Proposition 4.1 in \cite{shapiro:1986} for overparameterized models. 
	Let $\bm{\varphi}=\left(\rm vec(\mathbf{\cm{G}_{(1)})}^{\top},  \rm vec(\mathbf{U}_1)^{\top},  \rm vec(\mathbf{U}_2)^{\top},   \rm vec(\mathbf{U}_3)^{\top}\right)^{\top}$ be the component parameters in Tucker decomposition, and $\bm{h}(\bm{\varphi})=\rm vec(\mathbf{\cm{A}}_{(1)})=\rm vec(\mathbf{U}_1\mathbf{\cm{G}}_{(1)}(\mathbf{U}_{3} \otimes\mathbf{U}_2 )^{\top})$ is a function of $\bm{\varphi}$.
	Denote $\widehat{\bm{h}}_{\rm {OLS}}={\rm {vec}}((\widehat{\cm{A}}_{\rm {OLS}})_{(1)})$, and its asymptotic normality is established in the first step.	
	As in \cite{shapiro:1986}, we construct the following discrepancy function for any $\bm{h}$,
	\begin{align*}
	F(\bm{h}, \widehat{\bm{h}}_{\rm OLS})=\sum_{n=P+1}^{T}|| \widetilde{\mathbf{y}}_n- (\widetilde{\mathbf{x}}_n^\top \otimes\mathbf{I}_{N})\bm{h}||_2^2 -\sum_{n=P+1}^{T}|| \widetilde{\mathbf{y}}_n- (\widetilde{\mathbf{x}}_n^\top \otimes\mathbf{I}_{N})\widehat{\bm{h}}_{\rm OLS}||_2^2,
	\end{align*}
	which is a nonegative and twice continuously differentiable function. Moreover, it is equal to zero if and only if $\bm{h}=\widehat{\bm{h}}_{\rm OLS}$.
	
	Let ${\mathbf{T}_{ij}}(N, N, P) \in\mathbb{R}^{N ^2P \times N^2 P}$ be the tensor matricization transformation operator, which satisfies ${\rm {vec}} (\cm{A}_{(j)})={\mathbf{T}_{ij}}(N, N, P) \rm {vec}(\cm{A}_{(i)})$ for any tensor $\cm{A} \in \mathbb{R}^{N \times N \times P}$. Note that ${\mathbf{T}_{ij}}(N, N, P)$ is a full-rank matrix indicating the corresponding position of the tensor $\cm{A}$'s each entry in ${\rm {vec}} (\cm{A}_{(j)})$, and
	can be regarded as the natural extension of the permutation matrix for matrix transpose.  Although the ${\mathbf{T}_{ij} }(N, N, P)$ only depends on the value of $N$ and $P$,  we can simplify it to ${\mathbf{T}_{ij}}$ since both $N$ and $P$ are fixed in this theorem. 
	Thus,
	\begin{align*}
	{\rm vec}(\cm{A}_{(1)})=&{\rm vec}(\mathbf{U}_1\cm{G}_{(1)} ( \mathbf{U}_{3}\otimes\mathbf{U}_{2}) ^{\top}  )
	=\mathbf{T}_{21}{\rm vec}(\mathbf{U}_2\cm{G}_{(2)}( \mathbf{U}_{3}\otimes\mathbf{U}_{1})^{\top})\\
	=&\mathbf{T}_{31}{\rm vec}(\mathbf{U}_3\cm{G}_{(3)}( \mathbf{U}_{2}\otimes\mathbf{U}_{1})^{\top}),
	\end{align*}
	and the Jacobian matrix of $\bm h$ has the form of
	\begin{align*}
	\mathbf{H}={\partial \bm{h}(\bm{\varphi})}/{\partial\bm{\varphi}}=&\left( (\mathbf{U}_{3} \otimes \mathbf{U}_{2} \otimes \mathbf{U}_{1}),[( \mathbf{U}_{3}  \otimes \mathbf{U}_{2}) \mathbf{\cm{G}}_{(1)}^{\top}] \otimes\mathbf{ I}_{N}, \right.\\
	& \hspace{5mm} \left.	\mathbf{T}_{21}	
	\left\{\left[(\mathbf{U}_{3} \otimes \mathbf{U}_{1} ) \mathbf{\cm{G}}_{(2)}^{\top}\right]\otimes\mathbf{I}_{N}\right\}, \mathbf{T}_{31}\left\{\left[(\mathbf{U}_{2} \otimes \mathbf{U}_{1} ) \mathbf{\cm{G}}_{(3)}^{\top}\right] \otimes\mathbf{ I}_{P}\right\}
	\right).
	\end{align*}
	
	Denote by $\bm{h}(\widehat{\bm{\varphi}}_{\rm MLR})$ the minimizer of $F(\cdot,\widehat{\bm{h}}_{\rm OLS})$, and it corresponds to the MLR-HAR estimator, i.e. $\bm{h}(\widehat{\bm{\varphi}}_{\rm MLR})={\rm {vec}}((\widehat{\cm{A}}_{\rm {MLR}})_{(1)})$.
	Following Proposition 4.1 in \cite{shapiro:1986}, we can obtain the asymptotic normality below,
	\begin{align*}
	\sqrt{T} \{\bm{h}(\widehat{\bm{\varphi}}_{\rm MLR})  -  \bm{h} \}  \rightarrow N\left(\mathbf{0}, {\bm{\Sigma}}_{\rm MLR}\right)
	\end{align*}
	in distribution as $m\rightarrow\infty$ and $T\rightarrow\infty$, where ${\bm{\Sigma}}_{\rm MLR}=\mathbf{P}\bm{\Sigma}_{\rm{OLS}} \mathbf{P}^{\top}$, $\mathbf{P}= \mathbf{H}(\mathbf{H}^{\top}\mathbf{J}\mathbf{H})^{\dagger}\mathbf{H}^{\top}\mathbf{J}$ is the projection matrix,
	$\mathbf{J}=\bm{\Gamma}^{*} \otimes \bm{\Sigma}_{\varepsilon}^{-1} $ is the Fisher information matrix of $\bm{h}$, 
	and $\dagger$ denotes the Moore-Penrose inverse. Since $\bm{\Sigma}_{\rm{OLS}} =\mathbf{J}^{-1}$, we can easily obtain that $\bm{\Sigma}_{\rm MLR}=\mathbf{H}(\mathbf{H}^{\top}\mathbf{J}\mathbf{H})^{\dagger}\mathbf{H}^{\top}$.
	This accomplishes the proof of Theorem \ref{theorem 1}.
\end{proof}

\begin{proof}[Proof of Corollary \ref{corollary 1}]
	The asymptotic normality of $\widehat{\cm{A}}_{\rm OLS}$ has been proved at the first step of the proof of Theorem \ref{theorem 1}.
	Note that $\bm{\Sigma}_{\rm{OLS}} =\mathbf{J}^{-1}$ and $\bm{\Sigma}_{\rm MLR}=\mathbf{P}\mathbf{J}^{-1} \mathbf{P}^{\top}$, where $\mathbf{P}= \mathbf{H}(\mathbf{H}^{\top}\mathbf{J}\mathbf{H})^{\dagger}\mathbf{H}^{\top}\mathbf{J}$ is the projection matrix. As a result,
	\begin{equation}
	\bm{\Sigma}_{\rm{OLS}} -\bm{\Sigma}_{\rm MLR} =\mathbf{J}^{-1}-\mathbf{H}(\mathbf{H}^{\top}\mathbf{J}\mathbf{H})^{\dagger}\mathbf{H}^{\top}=\mathbf{J}^{-1/2}\mathbf{Q}_{\mathbf{J}^{1/2}\mathbf{H}}\mathbf{J}^{-1/2},
	\label{MLR}
	\end{equation}
	where $\mathbf{Q}_{\mathbf{J}^{1/2}\mathbf{H}}= \mathbf{ I}-\mathbf{J}^{1/2}\mathbf{H}(\mathbf{H}^{\top}\mathbf{J}\mathbf{H})^{\dagger}\mathbf{H}^{\top}\mathbf{J}^{1/2}$ is the projection matrix onto the orthogonal compliment  of span $({\mathbf{J}^{1/2}\mathbf{H}})$, and hence is positive semidefinite. Thus, ${\bm{\Sigma}}_{\rm{OLS}}\geq {\bm{\Sigma}}_{\rm{MLR}}$.
	
	For the MRI estimator, it is equivalent to restrict the coefficient tensor has Tucker decomposition of $\cm{A}=\cm{H}\times_{2}\mathbf{U}_2$ with a known rank $r_2$, and  the proving technique in Theorem \ref{theorem 1} is used to accomplish the corresponding proof. 
	
	Let $\bm{\theta}=\left(\rm{vec}(\cm{H}_{(1)})^{\top}, \rm{vec}(\mathbf{U}_2)^{\top} \right)^{\top}$, and $\bm{h}(\bm{\theta})=\rm{vec} (\cm{A}_{(1)}) =\rm{vec}(\mathbf{\cm{H}}_{(1)}(\mathbf{I}_{P} \otimes\mathbf{U}_2 )^{\top})$ be the function of $\bm{\theta}$. Similar to the case for $\widehat{\cm{A}}_{\rm MLR}$, we define the Jacobian matrix $\mathbf{R}:=\partial \bm{h}(\bm{\theta}) / \partial \bm{\theta}$ below, 
	\[
	\mathbf{R}=\left( \mathbf{I}_P \otimes \mathbf{U}_2 \otimes \mathbf{I}_N, \mathbf{T}_{21}	
	\left\{\left[(\mathbf{I}_{P} \otimes \mathbf{I}_{N} ) \mathbf{\cm{H}}_{(2)}^{\top}\right]\otimes \mathbf{I}_N\right\} \right)\in \mathbb{R}^{N^2P \times(NPr_2+Nr_2)}.
	\]
	Denote ${\bm\Sigma}_{\mathrm{MRI}}=\mathbf{R}\left(\mathbf{R}^{\top} \mathbf{J} \mathbf{R}\right)^{\dagger} \mathbf{R}^{\top}$. Similar to the proofs for $\widehat{\cm{A}}_{\rm MLR}$, we can show that $	\sqrt{T}\{{\rm{vec}}((\widehat{\cm{A}}_{\rm{MRI}})_{(1)})-{\rm{vec}}(\mathbf{\cm {A}}_{(1)})\}  \rightarrow  N(\mathbf{0}, {\bm{\Sigma}}_{\rm{MRI}})$ in
	distribution as $m \rightarrow \infty$ and $T \rightarrow \infty$, and $\bm\Sigma_{\mathrm{MRI}} \leq \mathbf{J}^{-1}=\bm{\Sigma}_{\rm{OLS}}$.
	
	Finally, since $\mathbf{U}_{2}$ in Tucker decomposition is exactly the same as the left singular vectors in the SVD of $\cm{A}_{(2)}$, we can view the Tucker decomposition as a further decomposition of the matrix $\cm{H}_{(1)}$, i.e. $\cm{A}=\cm{H}\times_{2}\mathbf{U}_2$ and $\cm{H}=\cm{G}\times_{1}\mathbf{U}_1\times_{3}\mathbf{U}_3$. Therefore, $\mathbf{H}=\partial \bm{h}(\bm{\theta}) / \partial  {\bm{\varphi}}=\partial \bm{h}(\bm{\theta}) / \partial \bm{\theta} \cdot \partial \bm{\theta}(\bm\varphi) / \partial {\bm\varphi}=\mathbf{R} \cdot \partial \bm{\theta}(\bm\varphi) / \partial {\bm \varphi}$. By a method similar to \eqref{MLR}, we can show that $\bm{\Sigma}_{\rm{MLR}} \leq \bm{\Sigma}_{\rm{MIR}}$, since $\rm{span}\left(\mathbf{J}^{1 / 2} \mathbf{H}\right) \subset \rm{span}\left(\mathbf{J}^{1 / 2} \mathbf{R}\right)$.
\end{proof}

\subsection{Proofs of Theorems \ref{thm2} and \ref{thm3} in Section \ref{Sect4}}

This subsection gives the technical proofs of Theorems \ref{thm2} and \ref{thm3} in Section \ref{Sect4}. In the meanwhile, we also provide four auxiliary lemmas.
Lemma \ref{lemma:covering} establishes covering number and discretization of low-multilinear-rank tensors,  Lemma \ref{lemma:RSC} derives Restricted strong convexity (RSC) and Restricted smoothness (RSM), Lemma \ref{lemma:DB} derives the deviation bound, and they will be used in the proof of Theorem \ref{thm2} and \ref{thm3}. Lemma \ref{lemma3} derives the contractive projection property (CPP), which is used in the proof of Theorem \ref{thm3}.
Throughout this subsection, we will use $C$ to represent generic positive numbers, whose value may vary from line to line.

\begin{proof}[Proof of Theorem \ref{thm2}]
	For simplicity, denote the multilinear low-rank estimator $\widehat{\cm{A}}_{\rm MLR}$ by $\widehat{\cm{A}}$, and let $\bm{\Delta}=\widehat{\cm{A}}-\cm{A}$, where $\cm{A}$ is the true parameter tensor. The loss function has the form of
	\begin{align*}
	L(\cm{A})=\frac{1}{T}\sum_{n=P+1}^{T}||	\widetilde{\mathbf{y}}_{n}-\cm{A}_{(1)}\widetilde{\mathbf{x}}_{n}||_2^2,
	\end{align*}
	where $\widetilde{\mathbf{x}}_{n}=(\widetilde{\mathbf{y}}_{n-1}^{\top},\ldots,\widetilde{\mathbf{y}}_{n-P}^{\top})^\top \in \mathbb{R}^{NP \times 1}$. 
	Due to the optimality of the $\widehat{\cm{A}}$, it holds that
	\begin{align*}
	&\frac{1}{T}\sum_{n=P+1}^T\|\widetilde{\mathbf{y}}_n - \widehat{\cm{A}}_{(1)}\widetilde{\mathbf{x}}_{n}\|_2^2 \le \frac{1}{T}\sum_{n=P+1}^T\|\widetilde{\mathbf{y}}_n - \cm{A}_{(1)}\widetilde{\mathbf{x}}_{n}\|_2^2,
	\end{align*}
	which implies that
	\begin{equation}\label{eq:add1}
	\frac{1}{T}\sum_{n=P+1}^T\|\bm{\Delta}_{(1)}\widetilde{\mathbf{x}}_{n}\|_2^2 \leq \frac{2}{T} \sum_{n=P+1}^T \langle \bm{\epsilon}_n, \bm{\Delta}_{(1)}\widetilde{\mathbf{x}}_{n} \rangle 
	\leq 2\langle \frac{1}{T}\sum_{n=P+1}^T \bm{\epsilon}_n \circ \widetilde{\mathbf{X}}_{n}, \bm{\Delta} \rangle,
	\end{equation}
	where $\widetilde{\mathbf{X}}_{n}=(\widetilde{\mathbf{y}}_{n-1},\ldots,\widetilde{\mathbf{y}}_{n-P}) \in \mathbb{R}^{N\times P} $, $\sum_{n=P+1}^T \langle \bm{\epsilon}_n, \bm{\Delta}_{(1)}\widetilde{\mathbf{x}}_{n} \rangle=\langle \sum_{n=P+1}^T \bm{\epsilon}_n \circ \widetilde{\mathbf{X}}_{n}, \bm{\Delta} \rangle$, and $\circ$ denotes the outer product.	
	
	Denote the set of tensors 
	$$\mathcal {S}(r_1,r_2,r_3) = \{\cm{A}\in \mathbb{R}^{N\times N \times P}: \|\cm{A}\|_{\rm F} = 1, \text{rank}_{i} (\cm{A}_i)\le r_i, 1\leq i \leq 3\}.$$ 
	Note that the Tucker ranks of both $\widehat{\cm{A}}$ and $\cm{A}$ are $(r_1,r_2,r_3)$, and hence the Tucker ranks of $\bm{\Delta}$ are at most $(2r_1,2r_2,2r_3)$. 
	As a result, from \eqref{eq:add1}, 
	\begin{equation*}
	\frac{1}{T}\sum_{n=P+1}^T\|\bm{\Delta}_{(1)}\widetilde{\mathbf{x}}_{n}\|_2^2 \leq 2\|\bm{\Delta}\|_{\rm F} \sup_{\bm{\Delta}\in \mathcal {S}(2r_1,2r_2,2r_3) } \langle \frac{1}{T}\sum_{n=P+1}^T \bm{\epsilon}_n \circ \widetilde{\mathbf{X}}_{n}, \bm{\Delta} \rangle,
	\end{equation*}
	and we hence can derive the estimation error bound by applying Lemmas \ref{lemma:RSC} and Lemma \ref{lemma:DB} with  $\delta$ being $\delta+1/2$ for simplicity. The prediction error bound can also be established from the above inequality, estimation error bound and Lemma \ref{lemma:DB}.
\end{proof}
\begin{proof}[Proof of Theorem \ref{thm3}]
	For a fixed $1\leq k\leq K$, define a linear space,
	\[
	\mathcal{A} = \{\alpha_1\widehat{ \cm{A}}_{k}+\alpha_2\cm{A},\hspace{2mm} \alpha_1,\alpha_2\in \mathbb{R}\},
	\]
	and denote by $(\cm{B})_{\mathcal{A}}$ the projection of $\cm{B}\in\mathbb{R}^{N\times N\times P}$ onto the space $\mathcal{A}$, where the dependence of $\mathcal{A}$ on $k$ is suppressed for simplicity.
	Since $\cm{A}\in \bm{\Theta}(r_1,r_2,r_3)$ and $\widehat{ \cm{A}}_{k}\in \bm{\Theta}(r_1^\prime,r_1^\prime,r_{3}^\prime)$, it holds that $\mathcal{A}  \subset\bm{\Theta}(r_1+r_1^\prime,r_2+r_2^\prime,r_{3}+r_{3}^\prime)$.
	
	Note that $\widetilde{\cm{A}}_{k}=\widehat{\cm{A}}_{k-1}-\eta\nabla L(\widehat{ \cm{A}}_{k-1})$, and $r_i^\prime\geq \left(\sqrt[3]{1+\frac{\kappa_{L}}{24\kappa_{U}}}-1\right)^{-2}r_i$ with $1\leq i \leq 3$. From Lemma \ref{lemma3}, we have 
	\begin{equation*}
	\|\widehat{ \cm{A}}_{k}-(\widetilde{ \cm{A}}_k)_{\mathcal{A}}\|_{\rm F} \le [\prod_{i=1}^3 (\sqrt{\frac{r_i}{r_i^{'}}}+1)-1]\| \cm{A}-(\widetilde{\cm{A}}_{k})_{\mathcal{A}}\|_{\rm F}
	\leq\frac{\kappa_{L}}{24\kappa_{U}}\| \cm{A}-(\widetilde{\cm{A}}_{k})_{\mathcal{A}}\|_{\rm F},
	\end{equation*}
	which, together with the fact that $1< 1+\frac{\kappa_{L}}{24\kappa_{U}}<2$, implies that
	\begin{align}
	\begin{split}\label{eq:a1a2}
	\|\widehat{ \cm{A}}_{k}- \cm{A}\|_{\rm F}\leq&\|\widehat{\cm{A}}_{k}-(\widetilde{\cm{A}}_{k})_{\mathcal{A}}\|_{\rm F}+\| \cm{A}-(\widetilde{\cm{A}}_{k})_{\mathcal{A}}\|_{\rm F}
	\leq(1+\frac{\kappa_{L}}{24\kappa_{U}})\| \cm{A}-(\widetilde{\cm{A}}_{k})_{\mathcal{A}}\|_{\rm F}\\
	\leq&(1+\frac{\kappa_{L}}{24\kappa_{U}})\|( \cm{A}-\widehat{ \cm{A}}_{k-1}-\eta[\nabla L( \cm{A})-\nabla L(\widehat{ \cm{A}}_{k-1})])_{\mathcal{A}}\|_{\rm F}+2\eta\|(\nabla L( \cm{A}))_\mathcal{A}\|_{\rm F}\\
	:=A_1+A_2.
	\end{split}
	\end{align}
	
	We first handle the term of $A_1$. Let $\mathbf{H} = T^{-1}\sum_{n=P+1}^T (\widetilde{\mathbf{x}}_n\widetilde{\mathbf{x}}_n^\top\otimes \mathbf{I}_N)$, and it holds that $T^{-1}\sum_{n=P+1}^T \| \bm{\Delta}_{(1)} \widetilde{\mathbf{x}}_{n} \|_2^2=T^{-1}\sum_{n=P+1}^T \| (\widetilde{\mathbf{x}}_{n}^{\top} \otimes \mathbf{I}_{N}){\rm{vec}}(\bm{\Delta}) \|_2^2 = \rm{vec}(\bm{\Delta})^{\top} \mathbf{H} \rm{vec}(\bm{\Delta})$.
	Then, from Lemma \ref{lemma:RSC} and for all $\bm{\Delta} \in \bm{\Theta}(r_1+r_1^\prime,r_2+r_2^\prime,r_{3}+r_{3}^\prime)$,
	\[
	 \frac{1}{8} \kappa_L \|\bm{\Delta}\|^2_{\rm F} \leq \rm{vec}(\bm{\Delta})^{\top} \mathbf{H} \rm{vec}(\bm{\Delta}) \leq  \frac{8}{3} \kappa_U \|\bm{\Delta}\|^2_{\rm F}
	\]
	with a probability at least
	\begin{equation}\label{eq:prob1}
	1-2\exp(-CT(\kappa_L/\kappa_U)^2\min\{\kappa^{-2},\kappa^{-4}\})-\exp(-Cd_{\mathcal{M}}^{\prime}).
	\end{equation}
	Note that $\eta={2}/({3\kappa_U})$, and $\mathbf{H}$ is the Hessian matrix of the loss function $L(\cm{B})$ with respect to $\rm{vec}(\cm{B})$. It holds that ${\rm{vec}}(\nabla L(\cm{A})-\nabla L(\cm{\widehat{A}}_{k-1}))=\mathbf{H}\rm{vec}(\cm{A}-\cm{\widehat{A}}_{k-1})$, and 
	\begin{align}
	\begin{split}\label{eq:with_hessian}
	A_1 &=(1+\frac{\kappa_{L}}{24\kappa_{U}})\|(( \mathbf{I}-\eta \mathbf{H}){\rm{vec}}( \cm{A}-\widehat{ \cm{A}}_{k-1}))_\mathcal{A}\|_2\\
	& \leq (1+\frac{\kappa_{L}}{24\kappa_{U}}) (1-\frac{\kappa_L}{12\kappa_U})\|\widehat{ \cm{A}}_{k-1}- \cm{A}\|_{\rm F}\\
	&\leq (1-\frac{\kappa_L}{24\kappa_U})\|\widehat{ \cm{A}}_{k-1}- \cm{A}\|_{\rm F},
	\end{split}
	\end{align}
	with the probability at \eqref{eq:prob1}, where $(\rm{vec}(\cm{B}))_{\mathcal{A}}=(\cm{B})_{\mathcal{A}}$, and the first inequality is by Lemma 4 of \cite{chen2019non}.
		
	We next handle the term of $A_2$ and, by Lemma 5 in \cite{chen2019non},
	\[
	A_2\leq \frac{32}{\kappa_L}\|(\nabla L( \cm{A}))_\mathcal{A}\|_{\rm F} =\frac{32}{\kappa_L} \sup_{\cm{S} \in\mathcal{A} , \|\cm{S}\|_{\rm F}=1}\langle \nabla L( \cm{A}), \cm{S} \rangle \leq \frac{32}{\kappa_L}\xi,
	\]
	where $\mathcal{S}(r_1+r_1^\prime,r_2+r_2^\prime,r_{3}+r_{3}^\prime)=\bm{\Theta}(r_1+r_1^\prime,r_2+r_2^\prime,r_{3}+r_{3}^\prime)\bigcap \{\|\cm{S}\|_{\rm F}=1\}$ and
	\begin{align*}
	\xi&= \sup_{\cm{S} \in\mathcal{S}(r_1+r_1^\prime,r_2+r_2^\prime,r_{3}+r_{3}^\prime)}\langle \nabla L( \cm{A}), \cm{S} \rangle = \sup_{\cm{S} \in\mathcal{S}(r_1+r_1^\prime,r_2+r_2^\prime,r_{3}+r_{3}^\prime)} \left\langle \frac{2}{T}\sum_{n=P+1}^T \bm{\epsilon}_n \circ \widetilde{\mathbf{X}}_{n}, \cm{S} \right\rangle\\
	&\leq C \left[(\kappa^2 \sqrt{\lambda_{\max}(\bm{\Sigma}_{{\varepsilon}})\kappa_U} +\kappa\sqrt{\kappa_U})\sqrt{\frac{d_{\mathcal{M}}^{\prime}}{T}} +\frac{T^{1+2\delta}}{m^{1/4}}\right]
	\end{align*}
	with a probability at least $1-\exp(-Cd_{\mathcal{M}}^{\prime})- 2\exp(-CT(\kappa_L/\kappa_U)^2\min\{\kappa^{-2},\kappa^{-4}\})$.
	This, together with \eqref{eq:a1a2} and \eqref{eq:with_hessian}, accomplishes the proof. 
	
\end{proof}

\begin{lemma}
	(Covering number and discretization of low-multilinear-rank tensors). 
	Suppose that $\bar{\mathcal{S}}(r_1,r_2,r_3)$ is an $\epsilon$-net of the set $S(r_1,r_2,r_3):=\{\cm{A}\in \mathbb{R}^{N\times N \times P}: \left\|\cm{A}\right\|_{\rm F}=1, {\rm{rank}}_{i}(\cm{A}_{(i)})\leq r_i, 1\leq i \leq 3\}$.
	\begin{itemize}
		\item [(i)] The cardinality of $\bar{\mathcal{S}}(r_1,r_2,r_3)$ satisfies
		\begin{align*}
		\left|\bar{\mathcal{S}}(r_1,r_2,r_3)\right| \leq (12/\epsilon)^{(r_1r_2r_3+Nr_1+Nr_2+Pr_3)}.
		\end{align*}
		\item [(ii)] For any tensor $\cm{N} \in \mathbb{R}^{N\times N \times P}$ and matrix $\mathbf{Z} \in \mathbb{R}^{NP \times T}$, it holds that,
		\begin{align*}
		\sup_{\bm{\Delta}\in \mathcal{S}(2r_1,2r_2,2r_3)} \left\langle \cm{N},\bm{\Delta} \right\rangle \leq &(1- 2\sqrt{2} \epsilon)^{-1} \max_{\bar{\bm{\Delta}}\in \bar{\mathcal{S}}(2r_1,2r_2,2r_3)}\left \langle \cm{N},\bar{\bm{\Delta}} \right \rangle, \hspace{5mm}\text{and}\\
		\sup_{\bm{\Delta}\in \mathcal{S}(2r_1,2r_2,2r_3)} \left\| \bm{\Delta}_{(1)}\mathbf{Z} \right\|_{\rm F} \leq & (1- 2\sqrt{2} \epsilon)^{-1} \max_{\bar{\bm{\Delta}}\in \bar{\mathcal{S}}(2r_1,2r_2,2r_3)}\left\| \bar{\bm{\Delta}}_{(1)}\mathbf{Z} \right\|_{\rm F}.
		\end{align*}
	\end{itemize}
	\label{lemma:covering}
\end{lemma}
\begin{proof}[Proof of Lemma \ref{lemma:covering}]
	Result at (i) is from Lemma A.1 of \cite{wang2020compact}, and we here prove (ii) only.
	
	Consider an $\epsilon$-net $\bar{\mathcal{S}}(2r_1,2r_2,2r_3)$ for $\mathcal{S}(2r_1,2r_2,2r_3)$. Then for any tensor $\bm{\Delta} \in \mathcal{S}(2r_1,2r_2,2r_3)$, there exists a $\bar{\bm{\Delta}} \in \bar{\mathcal{S}}(2r_1,2r_2,2r_3)$ such that $\| \bm{\Delta} - \bar{\bm{\Delta}}\|_{\rm F} \le \epsilon$. Since the rank of $\widebar{\cm{W}} =  \bm{\Delta} -\bar{\bm{\Delta}}$ are at most $(4r_1,4r_2,4r_3)$, we can split the HOSVD of $\widebar{\cm{W}} $ into 8 parts such that $\widebar{\cm{W}} =\sum_{i=1}^{8}\widebar{\cm{W}} _i$, where ${\rm{rank}}_j(\widebar{\cm{W}} _i) \le 2 r_j$ for $1\leq i \leq 8$ and $1\leq j \leq 3$, and $\left\langle \widebar{\cm{W}} _j, \widebar{\cm{W}} _k \right\rangle=0$ for any $j \ne k$. Then for any $\cm{N}\in \mathbb{R}^{N\times N \times P}$, we have
	\begin{equation}\label{eq:proof1}
	\left \langle \cm{N},\bm{\Delta} \right\rangle = \left\langle \cm{N},\bar{\bm{\Delta}} \right\rangle + \sum_{i=1}^{8} \langle \cm{N},\widebar{\cm{W}} _i  \rangle = \left\langle \cm{N},\bar{\bm{\Delta}} \right\rangle + \sum_{i=1}^{8} \left \langle \cm{N},\widebar{\cm{W}} _i/\|\widebar{\cm{W}} _i\|_{\rm F}  \right \rangle \left\|\widebar{\cm{W}} _i \right\|_{\rm F},
	\end{equation}
	where $\widebar{\cm{W}} _i/\left\|\widebar{\cm{W}} _i\right\|_{\rm F} \in S(2r_1,2r_2,2r_3)$, and $\left\langle \cm{N},\widebar{\cm{W}} _i/\|\widebar{\cm{W}} _i\|_{\rm F} \right\rangle \le \sup_{\bm{\Delta}\in S(2r_1,2r_2,2r_3)} \left \langle \cm{N},\bm{\Delta}\right\rangle$.
	
	Note that $\left\|\widebar{\cm{W}}\right\|_{\rm F}^2=\sum_{i=1}^{8}\left\|\widebar{\cm{W}}_i\right\|_{\rm F}^2$, and it holds that $\sum_{i=1}^{8}\left\|\widebar{\cm{W}}_i\|_{\rm F} \le 2\sqrt{2} \|\widebar{\cm{W}}\right\|_{\rm F} \le 2\sqrt{2} \epsilon$, which, together with \eqref{eq:proof1}, implies that
	\begin{align*}
	\gamma &:= \sup_{\bm{\Delta}\in \mathcal{S}(2r_1,2r_2,2r_3)} \left\langle \cm{N},\bm{\Delta} \right\rangle \le \max_{\bar{\bm{\Delta}}\in \bar{\mathcal{S}}(2r_1,2r_2,2r_3)}\left\langle \cm{N},\bar{\bm{\Delta}} \right\rangle + 2\sqrt{2}\gamma \epsilon, \hspace{5mm}\text{or}\\
	\gamma&=\sup_{\bm{\Delta}\in S(2r_1,2r_2,2r_3)}\left \langle \cm{N},\bm{\Delta} \right\rangle \le (1- 2\sqrt{2} \epsilon)^{-1} \max_{\bar{\bm{\Delta}}\in \bar{\mathcal{S}}(2r_1,2r_2,2r_3)}\left\langle \cm{N},\bar{\bm{\Delta}} \right\rangle.
	\end{align*}
	
	For matrix $\mathbf{Z} \in \mathbb{R}^{NP \times T}$, it holds that 
	\begin{align*}
	\left\| \bm{\Delta}_{(1)}\mathbf{Z} \right\|_{\rm F} &\leq \left\| \bar{\bm{\Delta}}_{(1)}\mathbf{Z} \right\|_{\rm F} + \sum_{i=1}^{8} \left\| (\widebar{\cm{W}}_i)_{(1)}\mathbf{Z} \right\|_{\rm F}\\
	& =  \left\| \bar{\bm{\Delta}}_{(1)}\mathbf{Z} \right\|_{\rm F} +  \sum_{i=1}^{8} \left\| \widebar{\cm{W}}_i \right\|_{\rm F} \left\| (\widebar{\cm{W}}_i)_{(1)}/ \| \widebar{\cm{W}}_i\|_{\rm F}\mathbf{Z} \right\|_{\rm F}\\
	&\leq \left\| \bar{\bm{\Delta}}_{(1)}\mathbf{Z} \right\|_{\rm F} + \sum_{i=1}^{8} \left\| \widebar{\cm{W}}_i \right\|_{\rm F} \sup_{\bm{\Delta} \in \mathcal{S}(2r_1,2r_2,2r_3)}\left\| \bm{\Delta}_{(1)}\mathbf{Z} \right\|_{\rm F}\\
	&\leq \left\| \bar{\bm{\Delta}}_{(1)}\mathbf{Z} \right\|_{\rm F} + 2\sqrt{2}\epsilon \sup_{\bm{\Delta} \in \mathcal{S}(2r_1,2r_2,2r_3)} \left\| \bm{\Delta}_{(1)}\mathbf{Z} \right\|_{\rm F},
	\end{align*}
	and we accomplish the proof by taking supremum on both sides.
\end{proof}

\begin{lemma}
	(Restricted strong convexity and smoothness). 
	Suppose that \eqref{condition_estimation error} and Assumptions \ref{assumption 1} and \ref{sub-Gaussian} hold. If $T \gtrsim \max (\kappa^2, \kappa^4) (\kappa_U/\kappa_L)^2d_{\mathcal{M}}$, $m^{1/4} \gtrsim T^{2\delta}$, $m^{1/4}T^{2\delta-1}\gtrsim N^2P\exp(d_{\mathcal{M}})$ for some $\delta > 1/2$, then  
	\begin{align*}
	\frac{1}{8} \kappa_L \left\|\bm{\Delta}\right\|^2_{\rm F} \leq \frac{1}{T}\sum_{n=P+1}^{T} \left\| \bm{\Delta}_{(1)} \widetilde{\mathbf{x}}_{n} \right\|_2^2 \leq \frac{8}{3} \kappa_U \left\| \bm{\Delta} \right\|^2_{\rm F},
	\end{align*}
	for all $\bm{\Delta} \in \mathcal{S}(2r_1,2r_2,2r_3)$ with probability at least 
	\begin{align*}
	1-2\exp(-CT(\kappa_L/\kappa_U)^2\min\{\kappa^{-2},\kappa^{-4}\})-\exp(-Cd_{\mathcal{M}}),
	\end{align*}
	where $\kappa$, $\kappa_U$, $\kappa_U$ and $d_{\mathcal{M}}$ are defined in Theorem \ref{thm2}.
	\label{lemma:RSC}
\end{lemma}

\begin{proof}[Proof of Lemma \ref{lemma:RSC}]
	Denote $R_T(\bm{\Delta})=\sum_{n=P+1}^{T} \|\bm{\Delta}_{(1)}\widetilde{\mathbf{x}}_n\|_2^2$, and it holds that
	\begin{align}
	\begin{split}
	{R_T(\bm{\Delta})}
	&= \sum_{n=P+1}^{T}\rm \left\|\bm{\Delta}_{(1)}(\mathbf{x}_{n}+\bm{\vartheta}_n)\right\|_2^2\\
	&= \sum_{n=P+1}^{T}{\mathbf{x}}_{n}^{\top} \bm{\Delta}_{(1)}^{\top}\bm{\Delta}_{(1)}{\mathbf{x}}_{n} + 2\sum_{n=P+1}^{T}{\mathbf{x}}_{n}^{\top} \bm{\Delta}_{(1)}^{\top}\bm{\Delta}_{(1)}{\bm{ \vartheta}}_{n} + \sum_{n=P+1}^{T}{\bm{ \vartheta}}_{n}^{\top} \bm{\Delta}_{(1)}^{\top} \bm{\Delta}_{(1)}{\bm{ \vartheta}}_{n}\\
	&:={\rm R}_1+ {\rm R}_2+ {\rm R}_3,
	\end{split}
	\label{R_split}
	\end{align}
	where $\bm{ \vartheta}_n =\widetilde{\mathbf{x}}_n - \mathbf{x}_{n}=(\bm{\eta}_{n-1}^\top,\cdots,\bm{\eta}_{n-P}^\top)^\top \in \mathbb{R}^{NP \times 1}$.
	By the spectral measure of ARMA processes in \cite{basu2015regularized}, we have $\lambda_{\min} \{\mathbb{E}(\mathbf{x}_n\mathbf{x}_n^\top)\} \geq \lambda_{\min}(\bm{\Sigma}_{\varepsilon})/\mu_{\max}(\cm{A}) = \kappa_L$ and $\lambda_{\max} \{\mathbb{E}(\mathbf{x}_n\mathbf{x}_n^\top)\} \leq \lambda_{\max}(\bm{\Sigma}_{\varepsilon})/\mu_{\min}(\cm{A}) = \kappa_U$, and it then holds that
	\begin{equation}\label{exp_R}
	(T-P) \kappa_L \le \mathbb{E}({\rm R}_1) = \mathbb{E}\left(\sum_{n=P+1}^{T}\text{vec}(\bm{\Delta})^\top(\mathbf{I}_N\otimes\mathbf{x}_n\mathbf{x}_n^\top)\text{vec}(\bm{\Delta})\right) \le (T-P) \kappa_U \le  T \kappa_U,
	\end{equation}
	as $\left\|\bm{\Delta}\right\|_{\rm F} = 1$. 
	Furthermore, $R_T(\bm{\Delta}) \ge \mathbb{E}({\rm R}_1) - \sup_{\bm{\Delta}\in \mathcal{S}}\{|{\rm R}_1 - \mathbb{E}({\rm R}_1)| + |{\rm R}_2| + |{\rm R}_3|\}$, and $R_T(\bm{\Delta}) \leq \mathbb{E}({\rm R}_1) + \sup_{\bm{\Delta}\in \mathcal{S}}\{|{\rm R}_1 - \mathbb{E}({\rm R}_1)| + |{\rm R}_2| + |{\rm R}_3|\}$, where $\mathcal{S}=\mathcal{S}(2r_1,2r_2,2r_3)$. We next first bound $|{\rm R}_1 - \mathbb{E}({\rm R}_1)| + |{\rm R}_2| + |{\rm R}_3|$ for each fixed $\bm{\Delta}\in\mathcal{S}(2r_1,2r_2,2r_3)$.
	
	Consider the term of ${\rm R}_1 - \mathbb{E}({\rm R}_1)$. Note that $\mathbf{x}_{n}=(\mathbf{y}_{n-1}^{\top}, \ldots,\mathbf{y}_{n-P}^{\top})^{\top}\in \mathbb{R}^{NP\times 1}$, and $\cm{A}_{(1)}=(\mathbf{A}_1,\dots,\mathbf{A}_P)$ with each $\mathbf{A}_{i}$ being an $N$-by-$N$ matrix. From \eqref{autoregressive}, we have $\mathbf{y}_n = \cm{A}_{(1)}\mathbf{x}_{n} + \bm{\varepsilon}_n$. It can be further rewritten into an VAR(1) form, $\mathbf{x}_{n} = \mathbf{B} \mathbf{x}_{n-1}  + \bm{e}_{n}$, and hence the VMA representation of $\mathbf{x}_n = \sum_{j=0}^{\infty}\mathbf{B}^j\bm{e}_{n-j}$, or $\mathbf{z}=\mathbf{P}\bm{e}$,
	where 
	\begin{align*}
	\mathbf{B}=
	\begin{pmatrix}
	\mathbf{A}_1 & \mathbf{A}_{2}  & \cdots & \mathbf{A}_{P-1} & \mathbf{A}_{P}  \\
	\mathbf{I}_N & \mathbf{0} & \cdots & \mathbf{0} & \mathbf{0}   \\
	\mathbf{0} & \mathbf{I}_N & \cdots & \mathbf{0} & \mathbf{0}  \\
	\vdots & \vdots & \ddots & \vdots & \vdots \\
	\mathbf{0} & \mathbf{0} & \cdots & \mathbf{I}_N & \mathbf{0}
	\end{pmatrix},\hspace{5mm} 
	\mathbf{P}=
	\begin{pmatrix}
	\mathbf{I}_{NP} & \mathbf{B} & \mathbf{B} ^2 & \cdots & \mathbf{B} ^{T-1} & \cdots\\
	\mathbf{0} & 	\mathbf{I}_{NP} & \mathbf{B} & \cdots & \mathbf{B}^{T-2} &\cdots  \\
	\mathbf{0} & \mathbf{0} & 	\mathbf{I}_{NP} & \cdots &\mathbf{B}^{T-3} & \cdots\\
	\vdots & \vdots & \vdots & \ddots & \vdots &\cdots\\
	\mathbf{0} &\mathbf{0} & \mathbf{0}& \cdots & 	\mathbf{I}_{NP}  & \cdots
	\end{pmatrix},
	\end{align*}
	$\bm{e}_{n}= ({\bm{\varepsilon}}_{n}^\top, \dots,\bm{0})^\top \in \mathbb{R}^{NP\times 1}$, $\bm{e}=(\bm{e}^{\top}_{T-1},\dots,\bm{e}^{\top}_{P}, \ldots)^{\top}$, and $\mathbf{z}=(\mathbf{x}^{\top}_{T},\dots,\mathbf{x}^{\top}_{P+1})^{\top}\in \mathbb{R}^{NP(T-P)\times 1}$. 
	Moreover, by Assumption \ref{sub-Gaussian}, the error term has the form of $\bm{e}=\bar{\bm{\Sigma}}\bar{\bm{\xi}}$, where $\bar{\bm{\xi}}_{n}=({\bm{\xi}}_{n}^\top, \ldots,\mathbf{0})^{\top}\in \mathbb{R}^{NP\times 1}$, $\bar{\bm{\xi}} = (\bar{\bm{\xi}}^{\top}_{T-1},\dots,\bar{\bm{\xi}}^{\top}_{P}, \ldots)^{\top}$, 
	\begin{align*}
	\bar{\bm{\Sigma}}_{\varepsilon}=
	\begin{pmatrix}
	\bm{\Sigma}_{\varepsilon} & \mathbf{0}& \cdots & \mathbf{0}\\
	\mathbf{0} & \bm{\Sigma}_{\varepsilon}   & \cdots & \mathbf{0}  \\
	\vdots & \vdots & \ddots & \vdots\\
	\mathbf{0} & \mathbf{0}   & \cdots & \bm{\Sigma}_{\varepsilon}
	\end{pmatrix} \in\mathbb{R}^{NP\times NP} \hspace{5mm}\text{and}\hspace{5mm} 	\bar{\bm{\Sigma}}=
	\begin{pmatrix}
	\bar{\bm{\Sigma}}_{\varepsilon}^{1/2} & \mathbf{0} &\mathbf{0}& \cdots \\
	\mathbf{0}& \bar{\bm{\Sigma}}_{\varepsilon}^{1/2} & \mathbf{0}& \cdots  \\
	\mathbf{0} & \mathbf{0} & \bar{\bm{\Sigma}}_{\varepsilon}^{1/2} & \cdots \\
	\vdots & \vdots & \vdots & \ddots \\
	\end{pmatrix}.
	\end{align*}
	Denote $\bm{\Sigma}_{\bm{\Delta}}=\bar{\bm{\Sigma}}\mathbf{P}^{\top}(\mathbf{I}_{T-P}\otimes\bm{\Delta}_{(1)}^{\top}\bm{\Delta}_{(1)})\mathbf{P}\bar{\bm{\Sigma}}$, and then
	\begin{equation*}
	{\rm R}_1 = \sum_{n=P+1}^{T}{\mathbf{x}}_{n}^{\top} \bm{\Delta}_{(1)}^{\top}\bm{\Delta}_{(1)}{\mathbf{x}}_{n} =\bar{\bm{\xi}}^{\top}\bar{\bm{\Sigma}}\mathbf{P}^{\top}(\mathbf{I}_{T-P}\otimes\bm{\Delta}_{(1)}^{\top}\bm{\Delta}_{(1)})\mathbf{P}\bar{\bm{\Sigma}}\bar{\bm{\xi}} =\bar{\bm{\xi}}^\top\bm{\Sigma}_{\bm{\Delta}}\bar{\bm{\xi}}.
	\end{equation*}
	
	Note that $\lambda_{\max} (\mathbf{P}\mathbf{P}^{\top}) = 1/\mu_{\min}(\cm{A})$, $\left\|\bm{\Sigma}_{\bm{\Delta}}\right\|_{\text{op}} \le \kappa_U$ and
	\begin{equation*}
	\left\|\bm{\Sigma}_{\bm{\Delta}}\right\|_{\rm F} \le \|\bar{\bm{\Sigma}}\|_{\text{op}}^2 \|\mathbf{P}\|_{\text{op}}\|\mathbf{P}^{\top}\|_{\text{op}} \|\mathbf{I}_{T-P}\otimes\bm{\Delta}_{(1)}^{\top}\bm{\Delta}_{(1)}\|_{\rm F} = \sqrt{T-P} \kappa_U \le \sqrt{T} \kappa_U.
	\end{equation*} 
	For any $t>0$, by Hanson-Wright inequality, we can bound ${\rm R}_1- \mathbb{E}({\rm R}_1)$ below,
	\begin{equation}
	\begin{split}
	\mathbb{P}\left[\left|{\rm R}_1- \mathbb{E}({\rm R}_1)\right|\ge t\right] 
	& \le 2\exp \left(-C \min \left(\frac{t}{\kappa^2 \left\|\bm{\Sigma}_{\bm{\Delta}}\right\|_{\text{op}}}, \frac{t^2}{\kappa^4 \left\|\bm{\Sigma}_{\bm{\Delta}}\right\|_{\rm F}^2}\right)\right)\\
	& \le 2\exp \left( -C \min \left(\frac{t}{\kappa^2 \kappa_U}, \frac{t^2}{\kappa^4 T \kappa_U^2}\right)\right).
	\end{split}
	\label{concen-I}
	\end{equation}
	
	We next consider the term of ${\rm R}_2$, which has the form of
	\begin{equation*}
	{\rm R}_2  = 2 \mathbf{z}^{\top}(\mathbf{I}_{T-P}\otimes\bm{\Delta}_{(1)}^{\top}\bm{\Delta}_{(1)})\mathbf{c} =  2\bar{\bm{\xi}}^\top\bar{\bm{\Sigma}}\bm{P}^\top(\mathbf{I}_{T-P}\otimes\bm{\Delta}_{(1)}^{\top}\bm{\Delta}_{(1)})\mathbf{c} =2{\bm{m}}_1^\top\bar{\bm{\xi}},
	\end{equation*}
	where ${\bm{m}}_1 = \bar{\bm{\Sigma}}\bm{P}^\top(\mathbf{I}_{T-P}\otimes\bm{\Delta}_{(1)}^{\top}\bm{\Delta}_{(1)})\mathbf{c}$, $\mathbf{c} = (\bm{\vartheta}_{T}^\top,\cdots, \bm{\vartheta}_{P+1}^\top)^\top \in \mathbb{R}^{NP(T-P) \times 1}$.
	Let $\eta_{\max} = \max _{1 \leq n \leq T} \max _{1 \leq i \leq N}\left|{\eta}_{i,n}\right|$, and it holds that $\|\mathbf{c}\|_2^2 \leq NPT\eta_{\max}^2$ and $\|{\bm{m}_1}\|_2^2 \leq NPT\eta_{\max}^2\kappa_U$.
	Moreover, by the inequality at \eqref{condition_estimation error} and Markov inequality, we have
	\begin{equation}
	\mathbb{P}\left\{\eta_{\max} \geq \frac{x}{m^{1 / 4}}\right\} \leq \frac{CNT}{x^2},
	\label{markov inquality}
	\end{equation}
	which can be used to control the rate of $\eta_{\max}$ by varying the value of $x$.
	By the sub-Gaussian condition of $\bar{\bm{\xi}}$ at Assumption \ref{sub-Gaussian} and letting $x=m^{1/8}T^{\delta}/N^{1/2}P^{1/2}$ in \eqref{markov inquality}, 
	\begin{equation}
		\label{concen-III}
	\mathbb{P}\left[2\left|{\bm{m}_1}^\top\bar{\bm{\xi}}\right|\geq t\right] \leq 2\exp\left(-\frac{Cm^{1/4}t^2}{\kappa^2T^{1+2\delta}\kappa_U}\right) + \frac{CN^2P}{m^{1/4}T^{2\delta-1}}
	\end{equation}
	for any $t>0$.
	
	Finally, for the third term of ${\rm R}_3$ at \eqref{R_split}, it holds that
	\begin{equation*}
	\sup_{\bm{\Delta} \in \mathcal{S}(2r_1,2r_2,2r_3)}|{\rm R}_3| = \sup_{\bm{\Delta} \in \mathcal{S}(2r_1,2r_2,2r_3)} \left|\mathbf{c}^{\top}(\mathbf{I}_{T-P}\otimes\bm{\Delta}_{(1)}^{\top}\bm{\Delta}_{(1)})\mathbf{c}\right|  \leq NPT\eta_{\max}^2,
	\end{equation*}
	and, by letting $x=m^{1/8}T^{\delta}/N^{1/2}P^{1/2}$, we have 
	\begin{equation}
	\begin{split}
	\mathbb{P}\left[\sup_{\bm{\Delta} \in \mathcal{S}(2r_1,2r_2,2r_3)}|{\rm R}_3|\geq \frac{T^{1+2\delta}}{m^{1/4}}\right] \leq \frac{CN^2P}{m^{1/4}T^{2\delta-1}}.
	\end{split}
	\label{concen-II}
	\end{equation}
	Thus, combining \eqref{concen-I}, \eqref{concen-III} and \eqref{concen-II}, we have
	\begin{equation*}
	\begin{split}
	& \mathbb{P}\left[| {\rm R}_1+{\rm R}_2+{\rm R}_3 - \mathbb{E}({\rm R}_1)| \geq t_1 + t_2 + \frac{T^{1+2\delta}}{m^{1/4}}\right]\\
	&\hspace{10mm} \leq \mathbb{P}\left[| {\rm R}_1- \mathbb{E}({\rm R}_1)| \geq t_1\right]+\mathbb{P}\left[|{\rm R}_2| \geq  t_2\right] + \mathbb{P}\left[| {\rm R}_3| \geq \frac{T^{1+2\delta}}{m^{1/4}}\right]\\
	&\hspace{10mm}\leq 2\exp \left( -C \min \left(\frac{t_1}{\kappa^2 \kappa_U}, \frac{t_1^2}{\kappa^4T \kappa_U^2}\right)\right) + 2\exp\left(-\frac{Cm^{1/4}t^2}{\kappa^2T^{1+2\delta}\kappa_U}\right) + \frac{CN^2P}{m^{1/4}T^{2\delta-1}}.
	\end{split}
	\end{equation*}
	Let $t_1 = t_2 = T  \kappa_L/6$ and, from \eqref{exp_R}, it holds that 
	\begin{equation}\label{point-RSC}
	\begin{split}
	&\mathbb{P}\left[0.5\kappa_L \leq T^{-1}{R_T(\bm{\Delta})} \leq 1.5\kappa_U\right] \\
	&\hspace{10mm}\geq 1 - 2\exp \left( -C \min \left(\frac{T\kappa_L}{\kappa^2 \kappa_U}, \frac{T\kappa_L^2}{\kappa^4 \kappa_U^2}\right)\right) -2\exp\left(- \frac{C\kappa_L^2m^{1/4}}{T^{2\delta-1}\kappa^2\kappa_U}\right) - \frac{CN^2P}{m^{1/4}T^{2\delta-1}},
	\end{split}
	\end{equation}
	as $m^{1/4} \geq {6T^{2\delta}}/{\kappa_L}$.
	Let $\bar{\mathcal{S}}$ to be an $\epsilon$-covering net of $S(2r_1,2r_2,2r_3)$. To construct the union bound, we rewrite $R_T(\bm{\Delta})$ as $R_T(\bm{\Delta}) = \|\bm{\Delta}_{(1)}\widetilde{\mathbf{X}}\|_{\rm F}^2$, where $\widetilde{\mathbf{X}} = (\tilde{\mathbf{x}}_{P+1}, \cdots, \tilde{\mathbf{x}}_{T})\in \mathbb{R}^{NP \times (T-P)}$. Define the event
	\begin{equation*}
	\mathcal{E}(\epsilon) = \left\{\forall \bm{\Delta} \in \bar{\mathcal{S}}(2r_1,2r_2,2r_3): \sqrt{0.5\kappa_L} \leq \frac{1}{\sqrt{T}}\|\bm{\Delta}_{(1)}\widetilde{\mathbf{X}}\|_{\rm F} \leq \sqrt{1.5\kappa_U}\right\}.
	\end{equation*}
	Then, by the pointwise bound in \eqref{point-RSC} and the covering number in Lemma \ref{lemma:covering}(i), 
	\begin{align*}
	&\mathbb{P}\left[\mathcal{E}^{c}(\epsilon)\right] \leq 2 \exp \left( Cd_{\mathcal{M}}-C \min \left(\frac{T\kappa_L}{\kappa^2 \kappa_U}, \frac{T\kappa_L^2}{\kappa^4 \kappa_U^2}\right)\right)\\
	&\hspace{20mm}+ 2\exp\left(Cd_{\mathcal{M}}- \frac{C\kappa_L^2m^{1/4}}{T^{2\delta-1}\kappa^2\kappa_U}\right)+\exp(Cd_{\mathcal{M}})\frac{CN^2P}{m^{1/4}T^{2\delta-1}}.
	\end{align*}
	
	Note that, by Lemma \ref{lemma:covering} (ii), 
	\begin{align*}
	\begin{split}
	\mathcal{E}(\epsilon)& \subset \left\{\max_{\bm{\Delta} \in \bar{\mathcal{S}}(2r_1,2r_2,2r_3)} \frac{1}{\sqrt{T}}\|\bm{\Delta}_{(1)}\widetilde{\mathbf{X}}\|_{\rm F} \leq \sqrt{1.5\kappa_U}\right\}\\
	&\subset \left\{\sup_{\bm{\Delta} \in \mathcal{S}(2r_1,2r_2,2r_3)} \frac{1}{\sqrt{T}}\|\bm{\Delta}_{(1)}\widetilde{\mathbf{X}}\|_{\rm F} \leq \frac{\sqrt{1.5 \kappa_U}}{1-2\sqrt{2}\epsilon}\right\}.
	\end{split}
	\end{align*}
	Moreover, similarly to Lemma \ref{lemma:covering}(ii), we can show that
	\begin{align*}
	\| \bm{\Delta}_{(1)}\mathbf{Z} \|_{\rm F} &\geq \| \bar{\bm{\Delta}}_{(1)}\mathbf{Z} \|_{\rm F} - \sum_{i=1}^{8} \| (\widebar{\cm{W}}_i)_{(1)}\mathbf{Z} \|_{\rm F}\\
	&\geq \| \bar{\bm{\Delta}}_{(1)}\mathbf{Z} \|_{\rm F} - \sum_{i=1}^{8} \| \widebar{\cm{W}}_i\|_{\rm F} \sup_{\bm{\Delta} \in \mathcal{S}(2r_1,2r_2,2r_3)}\| \bm{\Delta}_{(1)}\mathbf{Z} \|_{\rm F}\\
	&\geq \| \bar{\bm{\Delta}}_{(1)}\mathbf{Z} \|_{\rm F} - 2\sqrt{2}\epsilon \sup_{\bm{\Delta} \in \mathcal{S}(2r_1,2r_2,2r_3)}\| \bm{\Delta}_{(1)}\mathbf{Z} \|_{\rm F},
	\end{align*}
	where the last inequality is due to $\sum_{i=1}^{8} \| \widebar{\cm{W}}_i\|_{\rm F} \leq 2\sqrt{2}\epsilon$.
	Taking infimum on both sides, if $0 \leq \epsilon \leq \frac{1}{4\sqrt{2}}$, we have
	\begin{align*}
	\inf_{\bm{\Delta} \in S(2r_1,2r_2,2r_3)} \frac{1}{\sqrt {T}}\| \bm{\Delta}_{(1)}\mathbf{Z} \|_{\rm F} \geq& \min_{\bm{\Delta} \in \bar{\mathcal{S}}(2r_1,2r_2,2r_3)} \frac{1}{\sqrt {T}}\| \bm{\Delta}_{(1)}\mathbf{Z} \|_{\rm F} \\
	&- 2\sqrt{2}\epsilon \frac{1}{\sqrt {T}}\sup_{\bm{\Delta} \in \mathcal{S}(2r_1,2r_2,2r_3)}\| \bm{\Delta}_{(1)}\mathbf{Z} \|_{\rm F}\\
	\geq& \sqrt{0.5\kappa_L} - 2\sqrt{2}\epsilon \frac{\sqrt{1.5 \kappa_U}}{1-2\sqrt{2}\epsilon}
	\geq \sqrt{0.5\kappa_L} - 4\epsilon\sqrt{3 \kappa_U}.
	\end{align*}
	When $\epsilon$ is chosen to be $\frac{1}{8}\sqrt{\frac{\kappa_L}{6\kappa_U}}$, 
	\begin{align*}
	\mathcal{E}(\epsilon) \subset \left\{\inf_{\bm{\Delta} \in \mathcal{S}(2r_1,2r_2,2r_3)} \frac{1}{\sqrt{T}}\|\bm{\Delta}_{(1)}\widetilde{\mathbf{X}}\|_{\rm F} \geq \frac{\sqrt{0.5 \kappa_L}}{2}\right\},
	\end{align*}
	As a result, with the above choice of $\epsilon$, 
	\begin{align*}
	\mathcal{E}(\epsilon) \subset \left\{\frac{\kappa_L}{8}\leq \inf_{\bm{\Delta} \in \mathcal{S}(2r_1,2r_2,2r_3)} \frac{1}{T}\|\bm{\Delta}_{(1)}\widetilde{\mathbf{X}}\|_{\rm F}^2 \leq \sup_{\bm{\Delta} \in \mathcal{S}(2r_1,2r_2,2r_3)} \frac{1}{T}\|\bm{\Delta}_{(1)}\widetilde{\mathbf{X}}\|_{\rm F}^2 \leq \frac{8\kappa_U}{3} \right\}.
	\end{align*}
	Given the conditions that $T \gtrsim (\kappa_U/\kappa_L)^2 \max(\kappa^2, \kappa^4)d_{\mathcal{M}}$, $m^{1/4} \gtrsim T^{2\delta}$, $ m^{1/4}T^{2\delta-1}\gtrsim N^2P\exp(d_{\mathcal{M}})$ for some $\delta > 1/2$, we have that for all $\bm{\Delta} \in \mathcal{S}(2r_1, 2r_2, 2r_3)$,
	\begin{align*}
	\mathbb{P}\left[\frac{\kappa_L}{8}\leq \frac{1}{T}\sum_{n=P+1}^T\|\bm{\Delta}_{(1)}\widetilde{\mathbf{x}}_n\|_{2}^2 \leq \frac{8\kappa_U}{3}\right]
	&\geq 1- 2\exp \left( - C T (\kappa_L/\kappa_U)^2 \min (\kappa^{-2}, \kappa^{-4})\right)\\ &\hspace{5mm}-\exp\left(-Cd_{\mathcal{M}}\right).
	\end{align*}
	This accomplishes the proof. 
\end{proof}

\begin{lemma}
	Suppose that \eqref{condition_estimation error} and Assumptions \ref{assumption 1} and \ref{sub-Gaussian} hold. If sample size $T \gtrsim \max(\kappa^2,\kappa^4) (\kappa_U/\kappa_L)^2d_{\mathcal{M}}$,  $m^{1/4} \gtrsim T^{2\delta}$ and $m^{1/4}T^{2\delta-1}\gtrsim N^2P^{1/2}\exp(d_{\mathcal{M}})$ for some $\delta > 1/2$, then
	\begin{align*}
	\sup_{\bm{\Delta}\in \mathcal{S}(2r_1,2r_2,2r_3)} \langle \frac{1}{T}\sum_{n=P+1}^T \bm{\epsilon}_n \circ \widetilde{\mathbf{X}}_n, \bm{\Delta} \rangle \leq C(\kappa^2 \sqrt{\lambda_{\max}(\bm{\Sigma}_{{\varepsilon}})\kappa_U} +\kappa\sqrt{\kappa_U} )\sqrt{\frac{d_{\mathcal{M}}}{T}} + \frac{T^{2\delta}}{m^{1/4}}
	\end{align*}
	with probability at least 
	$$1-\exp(-Cd_{\mathcal{M}})- 2\exp(-CT(\kappa_L/\kappa_U)^2\min\{\kappa^{-2},\kappa^{-4}\}),$$
	where $\kappa$, $\kappa_U$, $\kappa_U$ and $d_{\mathcal{M}}$ are defined in Theorem \ref{thm2}.
	\label{lemma:DB}
\end{lemma}

\begin{proof}[Proof of Lemma \ref{lemma:DB}]
	To separate the influence of model error and estimation error, we let $\bm {\upsilon}_n=\bm{\eta}_n-\sum\limits_{j=1}^{P}\mathbf{A}_j\bm {\eta}_{n-j} \in \mathbb{R}^{N \times 1}$, ${\mathbf{X}}_{n}=({\mathbf{y}}_{n-1},\ldots,{\mathbf{y}}_{n-P}) \in \mathbb{R}^{N\times P}$, and $\bm{\Upsilon}_n=({\bm{\eta}}_{n-1},\ldots,{\bm{\eta}}_{n-P}) \in \mathbb{R}^{N\times P}$. Note that $\bm{\epsilon}_n = \bm{\varepsilon}_n + \bm{\upsilon}_n$ and $\widetilde{\mathbf{X}}_n = \mathbf{X}_n + \bm{\Upsilon}_n$. As a result,
	\begin{equation}
	\begin{split}
	&\sup_{\bm{\Delta}\in \mathcal {S}(2r_1,2r_2,2r_3) } \langle \frac{1}{T}\sum_{n=P+1}^T \bm{\epsilon}_n \circ \widetilde{\mathbf{X}}_{n}, \bm{\Delta} \rangle \\
	&\hspace{10mm}\leq \sup_{\bm{\Delta}\in \mathcal {S}(2r_1,2r_2,2r_3) } \langle \frac{1}{T}\sum_{n=P+1}^T \bm{\varepsilon}_n \circ \widetilde{\mathbf{X}}_{n}, \bm{\Delta} \rangle\\
	&\hspace{15mm}+\sup_{\bm{\Delta}\in \mathcal {S}(2r_1,2r_2,2r_3) } \langle \frac{1}{T}\sum_{t=P+1}^T \bm {\upsilon}_n \circ \mathbf{X}_{n}, \bm{\Delta} \rangle+\sup_{\bm{\Delta}\in \mathcal {S}(2r_1,2r_2,2r_3) } \langle \frac{1}{T}\sum_{n=P+1}^T \bm{\upsilon}_n \circ \bm{\Upsilon}_n, \bm{\Delta} \rangle. \label{eq:db-breakup}
	\end{split}
	\end{equation}
	We shall provide upper bounds for the three terms one by one.
	
	For the first term of (\ref{eq:db-breakup}), it is easily verified that $\langle  \bm{\varepsilon}_n \circ \widetilde{\mathbf{X}}_{n}, \bm{\Delta} \rangle =\langle  \bm{\varepsilon}_n , \bm{\Delta}_{(1)} \widetilde{\mathbf{x}}_{n} \rangle $.
	Denote $S_t(\bm{\Delta})=\sum_{n=P+1}^t \langle \bm{\varepsilon}_n, \bm{\Delta}_{(1)}\widetilde{\mathbf{x}}_{n} \rangle$ and $R_t(\bm{\Delta})=\sum_{n=P+1}^{t} \|\bm{\Delta}_{(1)}\widetilde{\mathbf{x}}_n\|_2^2$, for $P+1\leq t \leq T$. By the Chernoff bound of errors, for any $\alpha>0, \beta>0$ and $c>0$, there exists $\eta>0$,
	\begin{align}
	\begin{split}
	&\mathbb{P}\left[\left\{S_T(\bm{\Delta})\geq \alpha\right\} \bigcap \left\{R_T(\bm{\Delta}) \leq \beta\right\}\right]\\
	=& \inf_{\eta>0} \mathbb{P}\left[\left\{\exp\left(\eta S_T(\bm{\Delta})\right)\geq \exp\left(\eta \alpha\right)\right\} \bigcap \left\{R_T(\bm{\Delta}) \leq \beta\right\}\right]\\
	=& \inf_{\eta>0} \mathbb{P}\left[ \exp\left(\eta S_T\left(\bm{\Delta}\right)\right) \mathbb{I}\left(R_T(\bm{\Delta}\right) \leq \beta) \geq \exp(\eta \alpha)\right]\\
	\leq& \inf_{\eta>0} \exp\left(-\eta \alpha\right) \mathbb{E}\left[\exp(\eta S_T(\bm{\Delta}))  \mathbb{I}(R_T(\bm{\Delta})\leq \beta)\right]\\
	=& \inf_{\eta>0} \exp(-\eta\alpha+c\eta^2\beta) \mathbb{E}\left[\exp(\eta S_T(\bm{\Delta}) - c\eta^2\beta) \mathbb{I}(R_T(\bm{\Delta})\leq \beta)\right]\\
	\le& \inf_{\eta>0} \exp(-\eta\alpha+c\eta^2\beta)\mathbb{E}\left[\exp(\eta S_T(\bm{\Delta})-c\eta^2R_T(\bm{\Delta}))\right].
	\end{split}
	\label{probality_inter}
	\end{align}
	By the tower property for conditional expectations, we have
	\begin{align*}
	&\mathbb{E}\left[\exp\left(\eta S_T(\bm{\Delta})-c\eta^2R_T(\bm{\Delta})\right)\right]\\
	=& \mathbb{E}\left[\mathbb{E}\left[\exp\left(\eta S_T(\bm{\Delta})-c\eta^2R_T(\bm{\Delta})\right)|\mathcal{F}_{T-1}\right]\right]\\
	=&\mathbb{E}\left[\exp\left(\eta S_{T-1}(\bm{\Delta})-c\eta^2R_{T-1}(\bm{\Delta})\right)\mathbb{E}\left[\exp\left(\eta\langle \bm{\varepsilon}_T, \bm{\Delta}_{(1)}\widetilde{\mathbf{x}}_{T}\rangle - c\eta^2\|\bm{\Delta}_{(1)}\widetilde{\mathbf{x}}_{T}\|_2^2\right)|\mathcal{F}_{T-1}\right]\right].
	\end{align*}
	With the sub-Gaussianity condition in Assumption \ref{sub-Gaussian}, then $\langle \bm{\varepsilon}_{T}, \bm{\Delta}_{(1)}\widetilde{\mathbf{x}}_{T}\rangle = \langle \bm{\xi}_{T}, \bm{\Sigma}_{\varepsilon}^{1/2}\bm{\Delta}_{(1)}\widetilde{\mathbf{x}}_{T} \rangle$, and 
	$\mathbb{E}\left[\exp(\eta\langle \bm{\varepsilon}_{T}, \bm{\Delta}_{(1)}\widetilde{\mathbf{x}}_{T}\rangle)\right] \leq \exp\left(\eta^2 \kappa^2\lambda_{\max}(\bm{\Sigma}_{\varepsilon})\|\bm{\Delta}_{(1)}\widetilde{\mathbf{x}}_{T}\|_2^2/2\right)$. Since $\widetilde{\mathbf{x}}_{n}$ is $\mathcal{F}_{n-1}$-measurable, $\bm{\varepsilon}_n$ is $\mathcal{F}_{n}$-measurable and $\bm{\varepsilon}_n|\mathcal{F}_{n-1}$ is mean-zero,  let $c=\kappa^2\lambda_{\max}(\bm{\Sigma}_{\varepsilon})/2$, and the following inequalities can be easily deduced,
	\begin{align*}
	&\mathbb{E}[\exp(\eta S_T(\bm{\Delta})-\eta^2\kappa^2\lambda_{\max}(\bm{\Sigma}_{\varepsilon})R_T(\bm{\Delta})/2)]\\
	\leq &\mathbb{E}[\exp(\eta S_{T-1}(\bm{\Delta})-\eta^2\kappa^2\lambda_{\max}(\bm{\Sigma}_{\varepsilon})R_{T-1}(\bm{\Delta})/2)]\\
	\leq
	& \cdots 
	\leq \mathbb{E}[\exp(\eta S_{P+1}(\bm{\Delta})-\eta^2\kappa^2\lambda_{\max}(\bm{\Sigma}_{\varepsilon})R_{P+1}(\bm{\Delta})/2)]
	\leq 1.
	\end{align*}
	As a result, for any $\alpha >0$ and $\beta>0$, we can have the following inequality of (\ref{probality_inter}),
	\begin{align}
	\begin{split}
	&\mathbb{P}[\left\{S_T(\bm{\Delta})\geq \alpha \right\} \bigcap \left\{R_T(\bm{\Delta}) \leq \beta \right\}]\\
	\leq & \inf_{\eta>0} \exp(-\eta \alpha + \eta^2 \kappa^2 \lambda_{\max}(\bm{\Sigma}_{{\varepsilon}})\beta/2)\\
	=& \exp\left(-\frac{\alpha^2}{2\kappa^2\lambda_{\max}(\bm{\Sigma}_{{\varepsilon}})\beta}\right).
	\end{split}
	\label{S_T and R_T}
	\end{align}
	Moreover, according to Lemma \ref{lemma:RSC}, the following bounds for $R_T(\bm{\Delta})$ hold that
	\begin{equation}\label{R_T}
	\frac{T}{8} \kappa_L \leq R_T(\bm{\Delta}) \le \frac{8T}{3} \kappa_U
	\end{equation}
	with probability at least $1-2\exp(-CT(\kappa_L/\kappa_U)^2\min\{\kappa^{-2},\kappa^{-4}\})-\exp(-Cd_{\mathcal{M}})$.
	
	By Lemma \ref{lemma:covering} (ii), for any $x>0$,
	\begin{equation}\label{DB-covering}
	\begin{split}
	&\mathbb{P}\left[\sup_{\bm{\Delta}\in S(2r_1,2r_2,2r_3)} \left\langle \frac{1}{T}\sum_{n=P+1}^T \bm{\varepsilon}_n \circ \widetilde{\mathbf{X}}_{n}, \bm{\Delta} \right\rangle \ge x\right]\\
	\le& \mathbb{P}\left[ \max_{\bm{\Delta}\in \widebar{\mathcal{S}}(2r_1,2r_2,2r_3)} \left\langle \frac{1}{T}\sum_{n=P+1}^T\bm{\varepsilon}_n \circ \widetilde{\mathbf{X}}_{n}, \bm{\Delta} \right\rangle \ge (1-2\sqrt{2}\epsilon) x \right]\\
	\le& |\bar{\mathcal{S}}(2r_1,2r_2,2r_3)| \cdot \mathbb{P}\left[ \left\langle \frac{1}{T}\sum_{n=P+1}^T \bm{\varepsilon}_n \circ \widetilde{\mathbf{X}}_{n}, \bm{\Delta} \right\rangle \ge (1-2\sqrt{2}\epsilon) x \right],
	\end{split}
	\end{equation}
	which, together with \eqref{S_T and R_T} and \eqref{R_T}, implies that
	\begin{align*}
	&\mathbb{P}\left[ \left\langle \frac{1}{T}\sum_{n=P+1}^T \bm{\varepsilon}_n \circ \widetilde{\mathbf{X}}_{n}, \bm{\Delta} \right\rangle \geq (1-2\sqrt{2}\epsilon) x \right]\\
	\leq& \mathbb{P}[\{S_T(\bm{\Delta})\geq T(1-2\sqrt{2}\epsilon)x\} \bigcap \{R_T(\bm{\Delta}) \le CT  \kappa^2\kappa_U\}]\\
	\vspace{10mm}
	& + \mathbb{P}[R_T(\bm{\Delta}) \ge CT \kappa^2\kappa_U]\\
	\vspace{15mm}
	\leq & \exp\left[-\frac{(1-2\sqrt{2}\epsilon)^2Tx^2}{2C\kappa^4 \lambda_{\max}(\bm{\Sigma}_{{\varepsilon}})\kappa_U}\right] + 2\exp(-CT(\kappa_L/\kappa_U)^2\min\{\kappa^{-2},\kappa^{-4}\}) + \exp(-Cd_{\mathcal{M}})
	\end{align*}
	for any $x>0$.
	Note that, from Lemma \ref{lemma:covering}, $|\bar{\mathcal{S}}(r_1,r_2,r_3)| \leq (12/\epsilon)^{d_{\mathcal{M}}}$. By letting $\epsilon = 1/10$, and $x=C\sqrt{d_{\mathcal{M}}\kappa^4  \lambda_{\max}(\bm{\Sigma}_{{\varepsilon}})\kappa_U/T}$, we then have
	\begin{equation}\label{eq:db-I}
	\begin{split}
	&\mathbb{P}\left[\sup_{\bm{\Delta}\in \mathcal{S}(2r_1,2r_2,2r_3)} \left\langle \frac{1}{T}\sum_{t=1}^T \bm{\varepsilon}_n \circ \widetilde{\mathbf{X}}_{n}, \bm{\Delta} \right\rangle \ge C\kappa^2 \sqrt{\lambda_{\max}(\bm{\Sigma}_{{\varepsilon}})\kappa_U } \sqrt{\frac{d_{\mathcal{M}}}{T}}\right]\\
	\vspace{6mm}
	& \leq \exp(-Cd_{\mathcal{M}})+ 2\exp(-CT(\kappa_L/\kappa_U)^2\min\{\kappa^{-2},\kappa^{-4}\}).
	\end{split}
	\end{equation}
	
	For the second term in \eqref{eq:db-breakup}, since $\mathbf{x}_n=(\mathbf{y}_{n-1}^{\top}, \ldots, \mathbf{y}_{n-P}^{\top})^{\top}\in \mathbb{R}^{NP \times 1}$, it can be verified that $\langle  \bm {\upsilon}_n \circ \mathbf{X}_{n}, \bm{\Delta} \rangle=\langle \bm {\upsilon}_n, \bm{\Delta}_{(1)} \mathbf{x}_n \rangle $.
	Denote $\mathbf{z}=(\mathbf{x}_{T}^{\top}, \ldots, \mathbf{x}_{P+1}^{\top})^{\top}\in \mathbb{R}^{NP(T-P)\times 1}$, and $\bm{\upsilon} = (\bm{\upsilon}_{T}^{\top}, \ldots, \bm{\upsilon}_{P+1}^{\top})^{\top}\in \mathbb{R}^{N(T-P)\times 1}$, and it holds that
	\begin{align*}
	\sum_{n=P+1}^T\langle \bm{\upsilon}_n, \bm{\Delta}_{(1)} \mathbf{x}_n \rangle = \bm{\upsilon}^\top (\mathbf{I}_{T-P}\otimes \bm{\Delta}_{(1)}) \mathbf{z} := \bm{m}_2^\top\bar{\bm{\xi}},
	\end{align*}
	where $\bm{m}_2=\bar{\bm{\Sigma}}\mathbf{P}^\top(\mathbf{I}_{T-P}\otimes \bm{\Delta}_{(1)})\bm{\upsilon}$, $\mathbf{z}=\mathbf{P}\bm{e}$, and $\bm{e}=\bar{\bm{\Sigma}} \bar{\bm{\xi}}$. 
	Note that $\|\bm{\upsilon}\|_2^2 \leq C NT\eta_{\max}^2$ and $\|\bm{m}_2\|_2^2 \leq CNT\eta_{\max}^2\kappa_U$, where $\eta_{\max} = \max_{1\leq n \leq T}\max_{1\leq i \leq N}|\eta_{n,i}|$. Refer to the proof of Lemma \ref{lemma:RSC} for more details.
	
	With the sub-Gaussianity condition of $\bm{\xi}$ at Assumption \ref{sub-Gaussian}, similar to the proof of Lemma \ref{lemma:RSC}, we can obtain that
	\begin{align*}
	\mathbb{P}[|{\bm{m}}_2^\top\bar{\bm{\xi}}|\geq t|\eta_{\max}] \leq \exp\left(-\frac{Ct^2}{\kappa^2\kappa_UNT\eta_{\max}^2}\right) \hspace{5mm}\text{for all $t > 0$}
	\end{align*}
	and, by letting $x=m^{1/8}T^{\delta}/N^{1/2}$ in (\ref{markov inquality}), 
	\begin{equation*}
	\begin{split}
	\mathbb{P}\left[|{\bm{m}}_2^\top\bar{\bm{\xi}}|\geq t\right] \leq \exp\left(-\frac{Cm^{1/4}t^2}{\kappa^2\kappa_UT^{1+2\delta}}\right) + \frac{CN^2}{m^{1/4}T^{2\delta-1}}.
	\end{split}
	\end{equation*}
	Therefore, by a method similar to (\ref{DB-covering}),
	\begin{align*}
	&\mathbb{P}\left[ \sup_{\bm{\Delta}\in \mathcal{S}(2r_1,2r_2,2r_3)}\langle \frac{1}{T}\sum_{n=P+1}^T \bm {\upsilon}_n \circ \mathbf{X}_{t}, \bm{\Delta} \rangle 
	\geq t\right] \\
	&\leq \mathbb{P}\left[ \max_{\bm{\Delta}\in \bar{\mathcal{S}}(2r_1,2r_2,2r_3)}\langle \frac{1}{T}\sum_{n=P+1}^T \bm {\upsilon}_n \circ \mathbf{X}_{t}, \bm{\Delta} \rangle 
	\geq (1-2\sqrt{2}\epsilon)t\right]\\
	&\leq |\bar{\mathcal{S}}(2r_1,2r_2,2r_3)| \mathbb{P}\left[ \langle \sum_{n=P+1}^T \bm {\upsilon}_n \circ \mathbf{X}_{t}, \bm{\Delta} \rangle 
	\geq (1-2\sqrt{2}\epsilon)Tt\right]\\
	&\leq \exp(Cd_{\mathcal{M}})\frac{CN^2}{m^{1/4}T^{2\delta-1}} + \exp\left(Cd_{\mathcal{M}}-\frac{Cm^{1/4}t^2}{\kappa^2\kappa_UT^{2\delta-1}}\right).
	\end{align*}
	By letting $t=C\sqrt{\kappa^2\kappa_UT^{2\delta-1}d_{\mathcal{M}}/m^{1/4}}$ and with the conditions that $m^{1/4}\gtrsim T^{2\delta}$ and  $m^{1/4}T^{2\delta-1} \gtrsim N^2\exp(d_{\mathcal{M}})$ for some $\delta > 1/2$, we have 
	\begin{align}\label{eq:db-II}
	\mathbb{P}\left[ \sup_{\bm{\Delta}\in \mathcal{S}(2r_1,2r_2,2r_3)}\langle \frac{1}{T}\sum_{n=P+1}^T \bm {\upsilon}_n \circ \mathbf{X}_{t}, \bm{\Delta} \rangle 
	\geq C\kappa\sqrt{\kappa_U}\sqrt{\frac{d_{\mathcal{M}}}{T}}\right] 
	\leq \exp\left(-Cd_{\mathcal{M}}\right).
	\end{align}
	
	Finally for the third term at \eqref{eq:db-breakup}, denote $\bm{ \vartheta}_n=(\bm{\eta}_{n-1}^{\top}, \ldots, \bm{\eta}_{n-P}^{\top})^{\top}\in \mathbb{R}^{NP \times 1}$, $\bm {\upsilon}_n=\bm{\eta}_n-\sum\limits_{j=1}^{P}\mathbf{A}_j\bm {\eta}_{n-j}\in \mathbb{R}^{N \times 1}$, $\bm{\Upsilon}_n=(\bm{ \eta}_{n-1},\ldots,{\bm{ \eta}}_{n-P}) \in \mathbb{R}^{N \times P}$, and it holds that $\langle \bm {\upsilon}_n \circ \bm{\Upsilon}_n, \bm{\Delta} \rangle=\langle \bm {\upsilon}_n, \bm{\Delta}_{(1)} \bm{ \vartheta}_n \rangle $. Let $\mathbf{c} = (\bm{\vartheta}_{T}^\top,\cdots, \bm{\vartheta}_{P+1}^\top)^\top$, and
	\begin{align*}
	\sup_{\bm{\Delta}\in \mathcal{S}(2r_1,2r_2,2r_3)}\sum_{n=P+1}^T\langle \bm{\upsilon}_n, \bm{\Delta}_{(1)} \bm{ \vartheta}_n \rangle 
	&= \sup_{\bm{\Delta}\in \mathcal{S}(2r_1,2r_2,2r_3)} \bm{\upsilon}^\top(\mathbf{I}_{T-P}\otimes \bm{\Delta}_{(1)}) \mathbf{c}\\
	&\leq \sup_{\bm{\Delta}\in \mathcal{S}(2r_1,2r_2,2r_3)} CNP^{1/2}T\eta_{\max}^2
	\end{align*}
	as $\|\bm{\upsilon}\|_2 \leq C \sqrt{NT}\eta_{\max}$ and $\|\mathbf{c}\|_2 \leq \sqrt{NTP}\eta_{\max}$.
	Thus, letting $x = m^{1/8}T^{\delta}/N^{1/2}P^{1/4}$ at \eqref{markov inquality}, we have
	\begin{equation}
	\mathbb{P}\left[\sup_{\bm{\Delta}\in \mathcal{S}(2r_1,2r_2,2r_3)} \sum_{n=P+1}^T\langle \bm{\upsilon}_n, \bm{\Delta}_{(1)} \bm{\vartheta}_n \rangle \geq \frac{T^{1+2\delta}}{m^{1/4}}\right] \leq \frac{CN^2P^{1/2}}{m^{1/4}T^{2\delta-1}}.
	\label{eq:db-III}
	\end{equation}
	
	Combining \eqref{eq:db-I}, \eqref{eq:db-II}, and \eqref{eq:db-III}, the complete deviation bound can be obtained,
	\begin{align*}
	&\mathbb{P}\left[\sup_{\bm{\Delta}\in \mathcal{S}(2r_1,2r_2,2r_3)} \frac{1}{T}\sum_{n=P+1}^T\langle \bm{\epsilon}_n, \bm{\Delta}_{(1)} \widetilde{\mathbf{x}}_n \rangle \geq C(\kappa^2 \sqrt{\lambda_{\max}(\bm{\Sigma}_{{\varepsilon}}) \kappa_U} +\kappa\sqrt{\kappa_U} )\sqrt{\frac{d_{\mathcal{M}}}{T}}  + \frac{T^{2\delta}}{m^{1/4}}\right] \\
	&\leq \exp(-Cd_{\mathcal{M}})+ 2\exp(-CT(\kappa_L/\kappa_U)^2\min\{\kappa^{-2},\kappa^{-4}\}).
	\end{align*} We hence complete the proof.
\end{proof}

\begin{lemma}\label{lemma3}
	Consider three Tucker ranks $(r_1^{(j)},r_2^{(j)},r_3^{(j)})$ with $1\leq j\leq 3$, and assume that $r_i^{(1)}<r_i^{(2)}<r_i^{(0)}$ with $1\leq i\leq 3$. For $\cm{X}\in \bm{\Theta}(r_1^{(0)},r_2^{(0)},r_3^{(0)})$, it holds that 
	\begin{equation}
	\|P_{\bm{\Theta}(r_1^{(2)},r_2^{(2)},r_3^{(2)})}(\cm{X})- \cm{X}\|_{\rm F} \le [\Pi_{i=1}^3 (\beta_i +1)-1] \|\cm{Y} - \cm{X}\|_{\rm F}
	\end{equation}
	for any $\cm{Y} \in \bm{\Theta}(r_1^{(1)},r_2^{(1)},r_3^{(1)})$,
	where $\beta_i = \sqrt{(r_i^{(0)}- r_i^{(2)})/(r_i^{(0)}- r_i^{(1)})}$.
\end{lemma}
\begin{proof}[Proof of Lemma \ref{lemma3}]
	Refer to lemma 3 of \cite{chen2019non} for technical proofs. 
\end{proof}

\end{document}